\documentclass[12pt]{amsart}
\usepackage[margin=1in]{geometry}
\usepackage{graphicx} % Required for inserting images
\usepackage{amssymb,amsmath,amsthm,amsfonts,amscd,bbm}
\usepackage{tikz-cd}
\usepackage{tikz}
\usepackage[noadjust]{cite}
\usepackage{hyperref}
\usepackage{authblk}
 \usepackage{titling}

\tikzset{
  symbol/.style={
    draw=none,
    every to/.append style={
      edge node={node [sloped, allow upside down, auto=false]{$#1$}}}
  }
}

\newtheorem{thm}{Theorem}[section]
\newtheorem{prop}[thm]{Proposition}
\newtheorem{lem}[thm]{Lemma}
\newtheorem{cor}[thm]{Corollary}

\newtheorem{thmalpha}{Theorem}

\newcommand{\doi}[1]{\href{https://dx.doi.org/#1}{{\small DOI:#1}}}

\newcommand{\googlebooks}[1]{(preview at \href{https://books.google.com/books?id=#1}{google books})}

\newcommand{\numdam}[1]{}

\theoremstyle{remark}

\theoremstyle{definition}
\newtheorem{ex}[thm]{Example}
\newtheorem{defn}[thm]{Definition}
\newtheorem{quest}[thm]{Question}

\newtheorem{remark}[thm]{Remark}

\usepackage{graphicx,tikz}
\usetikzlibrary{arrows,backgrounds,scopes}
\usetikzlibrary{calc,cd}

\tikzset{coupon/.style={rectangle,rounded corners=1.5pt,draw,fill=white,inner sep=1.5,minimum size=12pt}}

\title{An operator algebraic approach to\\ fusion category symmetry on the lattice}
\author{David E. Evans, Corey Jones}
\address{David E Evans, School of Mathematics, Cardiff University, Senghennydd Road, Cardiff CF24 4AG, Wales, United Kingdom}

\address{Corey Jones, Department of Mathematics,
North Carolina State University, Raleigh, NC 27695, USA}

\begin{document}

\begin{abstract}
We propose a framework for fusion category symmetry on the (1+1)D lattice in the infinite-volume limit by giving a formal interpretation of SymTFT decompositions. Our approach is based on axiomatizing physical boundary subalgebra of quasi-local observables, and applying ideas from algebraic quantum field theory to derive the expected categorical structures. We show that given a physical boundary subalgebra $B$ of a quasi-local algebra $A$, there is a canonical fusion category $\mathcal{C}$ that acts on $A$ by bimodules and whose fusion ring acts by locality preserving quantum channels on the quasi-local algebra such that $B$ is recovered as the fixed point operators. We show that a fusion category can be realized as symmetries on a tensor product quasi-local algebra if and only if all of its objects have integer dimensions, and that it admits an ``on-site" action on a tensor product spin chain if and only if it admits a fiber functor. We give a formal definition of a topological symmetric state, and prove two anomaly enforced gaplessness theorems, one for internal categorical symmetries and one for anomalous duality channels. Using the first, we show that for any fusion category $\mathcal{C}$ with no fiber functor there always exist gapless pure symmetric states on an anyon chain. 
\end{abstract}

\maketitle

\tableofcontents

\section{Introduction}

Global symmetries play a fundamental role in understanding the structure of quantum many-body systems and quantum field theories. Recently, there has been significant interest in extending the traditional conception of global group symmetry to a higher categorical, non-invertible context (\cite{GKSW, PhysRevResearch.2.043086,SCHAFERNAMEKI20241,shao2024whatsundonetasilectures} for overview and references). This has given rise to a rich interplay between quantum field theory and higher categorical algebra. In (1+1)D, (finite) categorical symmetries of a theory are characterized by a unitary fusion category \cite{TW24, TW242}. We contribute to this story by proposing a mathematical formulation of fusion category symmetry on the (1+1)D lattice in infinite volume via an operator algebraic interpretation of the \textit{SymTFT} picture of categorical symmetry.

%In this paper, we propose a framework for describing fusion category symmetry of (1+1)D lattice systems in the infinite volume limit by axiomatizing the algebra of symmetric local operators. While there is already a well-established local picture for fusion category symmetries on a spin chain in terms of matrix product operators (MPOs) \cite{MR3543452,MR3719546,aasen2020topologicaldefectslatticedualities,PRXQuantum.4.020357,PRXQuantum.5.010338, In22, Ho23}, this perspective does not immediately provide access to many of the technical tools of quantum field theory that are useful in the IR.  %Ideally, however, we would like to mathematically formulate and prove important physical results, such as the classification of categorical SPTs \cite{GarreRubio2023classifyingphases,PhysRevLett.133.161601} and categorical versions of the Lieb-Schultz-Mattis theorem \cite{TW24, 10.21468/SciPostPhys.16.6.154}. As low energy effective theories are emergent phenomenon, we would only expect these results at best to be formally valid in the thermodynamic limit. The standard mathematical picture for formalizing the (infinite volume) thermodynamic limit is defined in terms of the C*-algebra of quasi-local observables \cite{MR1441540}, and thus it is natural to look for a model of fusion category symmetry that captures the standard MPO story, but applies in this setting. 
\textit{SymTFT} \cite{freed2024topological, PhysRevResearch.2.033417, GK21,ABGIH}, also called \textit{topological holography} \cite{ PhysRevB.108.075105, 10.21468/SciPostPhysCore.6.4.066,PhysRevB.111.115161,huang2025topologicalholographyquantumcriticality, MR3763324, PhysRevResearch.2.043044,KLWZZ,PhysRevResearch.1.033054,LZI,LZII, LZIII}, see in particular Section II \cite{chatterjee2024emergentgeneralizedsymmetrymaximal} for a detailed overview), has emerged as a powerful framework for studying categorical symmetries in all dimensions. SymTFT imports mathematical ideas and techniques from topological quantum field theories (TQFT) to characterize categorical symmetries of arbitrary quantum field theories. In this picture, a quantum field theory $\mathcal{F}$ with categorical symmetry $\mathcal{C}$ has a decomposition of the system into a topological boundary $\mathcal{B}_{\text{top}}$ and physical boundary $\mathcal{B}_{\text{phys}}$, coupled via a bulk topological theory $\mathcal{T}$ in one higher dimension, called the SymTFT. 

$$\begin{tikzpicture}
\draw [fill=gray!30, gray!30] (4,0.5) rectangle (6,1.5);
\draw node at (1,1.5) {$\mathcal{F}$};
\draw (0,1)--(2,1);
\draw node at (5,1.8) {$\mathcal{B}_{\text{top}}$};
\draw (4,0.5)--(6,0.5);
\draw node at (3,1) {=};
\draw (4,1.5)--(6,1.5);
\draw node at (5,0.18) {$\mathcal{B}_{\text{phys}}$};
\draw node at (5,1) {$\mathcal{T}$};
\end{tikzpicture}$$

\noindent The topological boundary $\mathcal{B}_{\text{top}}$ should realize $\mathcal{C}$ via its boundary topological defects, and this implements the symmetry $\mathcal{C}$ on the whole theory. The picture above is often referred to as the ``sandwich" picture \cite{KONG201762,freed2024topological}. Not only does this perspective make higher categorical symmetries of arbitrary theories tractable from a mathematical standpoint, but it provides a deeper understanding of their meaning and significance.

Typically this framework is applied to continuum theories. Recent work has demonstrated how to represent SymTFT ideas on the lattice concretely in examples \cite{bhardwaj2025latticemodelsphasestransitions, PhysRevB.111.054432}. Ideally, however, we would like our formulation of categorical symmetry on the lattice to give ``in-situ" access to the algebraic structures underlying a SymTFT decomposition in a model independent way. In particular, we would like to directly see the bulk topological order (described algebraically by a unitary modular tensor category), the gapped boundary (described algebraically by a Lagrangian algebra in this category), and the symmetry category of defects (described algebraically by the fusion category of modules over the algebras), all internally to the formal mathematical structure of the infinite volume limit system. 

We propose a definition of a SymTFT decomposition in terms of the \textit{quasi-local algebra} of a 1+1D lattice model. Recall the standard mathematical formulation of a spin chain is a pair $(A,H)$, where $A$ is the quasi-local algebra of observables in the infinite volume limit over $\mathbbm{Z}$ (traditionally this is the infinite tensor product of local matrix algebras, but more generally can be an ``anyon" chain), and $H$ is the Hamiltonian operator (formally represented by an unbounded derivation on the quasi-local algebra) \cite{MR1441540}. If we view this as a kind of discrete-space quantum field theory, we can try to make the above SymTFT description mathematically precise in terms of the operator algebra $A$ and its Hamiltonian $H$. Our starting point is the observation that if we had a SymTFT decomposition, the quasi-local operators localized near the physical boundary should form a subalgebra of $A$, and as was argued in \cite{PhysRevB.107.155136}, this should ``remember" the emergent bulk topological order. We propose formal axioms for a subalgebra $B\subseteq A$ to be the local observables of the physical boundary in a SymTFT decomposition, and we reconstruct from physical boundary subalgebras an explicit fusion categorical symmetry.

The basic idea uses the theory of $\text{DHR}$-bimodules of an abstract quasi-local algebra introduced in \cite{jones2024dhr} (see also \cite{NSz97}). This is a version of DHR superselection theory from algebraic quantum field theory \cite{MR0297259,DHR2,GFr93}, but is optimized for use in spin systems, where the local algebras are finite-dimensional, and without a fixed choice of Hilbert space representation or vacuum state. The DHR category of an abstract quasi-local algebra $B$ is a certain unitary braided tensor category of $B$-$B$ bimodules, and recently under fairly general conditions it has been shown that in 1+1D, $\text{DHR}(B)$ is a non-chiral unitary modular tensor category \cite{hataishi2025structuredhrbimodulesabstract}. Furthermore, it was argued in \cite{jones2025localtopologicalorderboundary} that the category of DHR bimodules of the physical boundary algebra of locally topologically ordered spin systems should agree with the bulk topological order (see also Section \ref{subsec:physpic}).

This leads us to propose for a subalgebra $B\subseteq A$, the bulk TQFT $\mathcal{T}$ in the SymTFT decomposition should be identified with the 2+1D TQFT associated with the unitary modular tensor category $\text{DHR}(B)$. In our set-up, $A$ is not only an algebra but an \textit{algebra object} internal to $\text{Bim}(B)$, so as part of our axiomatization we ask that it is in fact a Lagrangian algebra in $\text{DHR}(B)$, thus corresponding to the gapped boundary. This leads us to a formal definition of a \text{physical boundary subalgebra} $B\subseteq A$. Modulo some minor technicalities, we can summarize the definition here.

\begin{defn} If $A$ is a quasi-local algebra over $\mathbbm{Z}$, a subalgebra $B\subseteq A$ is called a \textit{physical boundary subalgebra} if $A\in \text{DHR}(B)$ is a Lagrangian algebra object.
\end{defn}

A large inventory of examples arise from matrix product operator (MPO) symmetry \cite{MR3543452,MR3719546,aasen2020topologicaldefectslatticedualities,PRXQuantum.4.020357,PRXQuantum.5.010338, In22, Ho23}. In particular, this definition is satisfied by the quasi-local algebra $A$ built from local operators on a spin or anyon chain, with the subalgebra $B$ generated by local operators invariant under a symmetry (for the standard MPO picture, (\cite{Kaw23,jones2025quantumcellularautomatacategorical}, see Examples \ref{ex:FusionSpin} and \ref{ex:fusionspinsubalg}). We point out that our subalgebra approach is agnostic toward the actual algebraic structure of the symmetry operators, i.e. whether we are using MPOs, Hopf algebras, weak Hopf algebras, etc.

Given this set-up, the general philosophy of SymTFT suggests we should now have access to a symmetry fusion category $\mathcal{C}$, for which $\text{DHR}(B)\cong \mathcal{Z}(\mathcal{C})$. Inspired by $\alpha$-induction for conformal nets \cite{MR1332979, MR1617550, MR1729094, MR1777347,
MR1785458, MR4642115}, for a physical boundary subalgebra $B\subseteq A$, we define a fusion category $\mathcal{C}$ (formally defined as $A$ modules internal to $\text{DHR}(B)$) as a particular full subcategory of bimodules of the quasi-local algebra $A$ that satisfy a solitonic type of localization as in \cite{MR1671970}. We also obtain an interpretation of the objects $\mathcal{C}$ as topological lattice defects relative to our symmetry subalgebra (c.f. \cite{10.21468/SciPostPhys.16.6.154,Ho23}). We then show that the fusion ring of $\mathcal{C}$ acts by \textit{quantum channels} on $A$ (formulated in the Heisenberg picture as unital completely positive maps form $A$ to $A$), giving an action on states in the infinite volume limit. 

We emphasize that we are starting only with a subalgebra $B\subseteq A$ satisfying some conditions, and then we \textit{derive} the symmetry category $\mathcal{C}$ and its realization as defects on our system. The Theorem below is a summary of the results of Section \ref{sec:localsubalg}.

\begin{thmalpha}{Symmetry from SymTFT.} Let $B\subseteq A$ be a physical boundary subalgebra of a quasi-local algebra $A$ over $\mathbbm{Z}$.
  \begin{enumerate}
    \item 
    There is a canonically associated fusion category $\mathcal{C}$ called the \textit{symmetry category} with $\mathcal{Z}(\mathcal{C})\cong \text{DHR}(B)$.
    \item 
    There is a canonical embedding of unitary tensor categories $\mathcal{C}\hookrightarrow \text{Bim}(A)$, as the bimodules with solitonic type localization.
    \item 
    There is an action of the fusion ring of $\mathcal{C}$ on $A$ by unital, completely positive maps with bounded spread such that the fixed point subalgebra is precisely $B$.
\end{enumerate}

We say $\mathcal{C}$ is the symmetry fusion category of the physical boundary subalgebra.

\end{thmalpha}

The above theorem also gives us an alternative perspective on our story that puts symmetry (rather than SymTFT) first. Instead of starting from the subalgebra $B\subseteq A$, we could start with an action of the fusion ring of $\mathcal{C}$ (thought of as a ``hypergroup") on $A$ by bounded spread unital completely positive (ucp) maps, and consider the fixed point algebra $B$. If $B\subseteq A$ is a physical boundary subalgebra (which we would expect in a SymTFT decomposition), then we can apply our theory to recover the SymTFT and $\mathcal{C}$. Indeed, given an MPO representation of the fusion category $\mathcal{C}$, after renormalization there is a natural way to obtain an action of the fusion ring of $\mathcal{C}$ by u.c.p maps on the quasi-local algebra of the \textit{constrained Hilbert space} defined by the unit MPO \cite{jones2025quantumcellularautomatacategorical}. Applying the chain of reasoning described above will precisely recover these MPO \textit{channels} (rather than the traced out MPO operators with periodic boundary conditions) from the symmetric subalgebra. 

The mathematical structures we described above are closely parallel to (and are directly inspired by) results in the operator algebraic approach to conformal field theory in 1+1 dimensions \cite{GFr93}. Despite obvious differences (particularly with regard to the symmetry covariance requirements), there are many structural similarities between our set-up on the lattice and conformal nets\footnote{We note the more recent work \cite{Ho23} also studies categorical symmetries on the lattice in close analogy with conformal field theory.}. For example, the idea that local extensions correspond to commutative Q-systems (see Section \ref{sec:Q-sys}) in the DHR category originates in \cite{MR1332979}, and the realization of non-invertible symmetries as locality preserving quantum channels (see \ref{sebsec:SymasChan}) first appeared in the conformal net setting \cite{MR3595480}. We also point to the related work \cite{ahmad2026facesnoninvertiblesymmetries}, which emphasizes the role of weak Hopf algebras, but also utilizes subfactor theory in the context of categorical symmetries of quantum field theories.

%We also employ a form of $\alpha$-induction in the setting of DHR symmetries on 
%spin chains. Whilst one could deduce the chiral extensions and some of their categorical properties from categorical considerations as indicated above,  we go beyond that to understand $\alpha$ induction in terms of DHR symmetries of spin chains.

In this paper, our applications primarily focus on symmetric kinematics, namely the observables and states at a fixed time slice. The dynamical aspects of categorical symmetries are very rich, and we plan to pursue development of this theory in future work. In the next two subsections, we give an overview of our results in the paper.

\subsection{Symmetries of tensor product quasi-local algebras} 
In our framework, we make very few assumptions on the nature of the original quasi-local algebra $A$. While the usual choice for a spin system corresponds to the tensor product UHF quasi-local algebra $\otimes_{\mathbbm{Z}} M_{d}(\mathbbm{C})$ \cite{MR1441540}, other interesting examples include ``anyon chains", which are important for capturing the full scope of potential MPO symmetries. Anyon chains are locally embeddable as \textit{non-unital} subalgebras of tensor product quasi-local algebras, and symmetries in these contexts are sometimes called \textit{emergent symmetries}. However, the question of realizability of symmetries on a genuine tensor product quasi-local algebra is still interesting from a physical point of view, since tensor product Hilbert spaces naturally occur for physical spin systems and are the basic model utilized for states on quantum processors.

For many fusion categories, the only known examples of quasi-local algebras hosting these categories as symmetries are anyon chains, and it is not clear whether these categories admit a realization on tensor product quasi-local algebras. Any fusion category which admits a fiber functor can easily be realized on tensor product Hilbert spaces by utilizing an on-site action of a Hopf algebra. A fusion category is called \textit{anomalous} if it does not admit a fiber functor. There are some examples of anomalous fusion categories (in particular, $\text{Vec}(\mathbbm{Z}/2\mathbbm{Z},\omega)$ which nevertheless admit actions on tensor product quasi-local algebras (for example, see \cite{kapustin2024anomaloussymmetriesquantumspin, bols2025classificationlocalitypreservingsymmetries}) raising the question of precisely characterizing which fusion categories admit actions on tensor product Hilbert spaces. Recall that a fusion category is \textit{integral} if the quantum dimension of each object is an integer, We have the following theorem that settles this question, which is a combination of Proposition \ref{prop:obstruction} and Theorem \ref{thm:intrealiz}.

\begin{thmalpha}
    A fusion category $\mathcal{C}$ is realized as symmetries of a tensor product quasi-local algebra if and only if it is integral.
\end{thmalpha}

The proof that any fusion category realized as symmetries on a tensor product quasi-local algebra must be integral uses a standard $K$-theory obstruction for categorical symmetries on C*-algebras \cite{MR4419534}. Our construction of tensor product symmetries for any integral fusion categories is new, and utilizes a standard MPO action on an anyon chain, but then conjugates by a bounded spread isomorphism to a tensor product quasi-local algebra. 

In our framework, we define a categorical symmetry to be \textit{on-site} if the associated symmetry channels have $0$-spread (i.e. are strictly locality preserving). As mentioned above, if $\mathcal{C}$ is anomaly-free, we can construct \textit{on-site} actions using a fiber functor (see \cite{meng2025noninvertiblesptsonsiterealization, seifnashri2025disentanglinganomalyfreesymmetriesquantum}). Examining our construction above, however, we see that it generally produces non on-site actions. This suggests a connection between ``onsite-ability" and the existence of a fiber functor. We prove the following theorem, which is Corollary \ref{cor:on-siteanomaly} in the main body of the paper.

\begin{thmalpha} A fusion category is realized as on-site symmetries of a tensor product quasi-local algebra if and only if $\mathcal{C}$ admits a fiber functor.
\end{thmalpha}

Again, these arguments only use the structure arising from the inclusion of algebras $B\subseteq A$, and not on any particular representation of the symmetry (e.g. matrix product operators vs weak Hopf algebra, etc.).

\subsection{Symmetric states}

A motivating problem for categorical symmetries is understanding the universality classes of symmetric states and phase transitions between these in terms of categorical data (this program has been fittingly called the Categorical Landau Paradigm \cite{PhysRevLett.133.161601,PhysRevB.111.054432, PhysRevB.108.075105}). A symmetric state is a state $\phi$ on $A$ invariant under all symmetry channels. The definition of a physical boundary subalgebra $B\subseteq A$ guarantees the existence of a unique conditional expectation $E:A\rightarrow B$. Thus the symmetry condition can be equivalently restated as $\phi\circ E=\phi$. The conditional expectation establishes a bijection between symmetric states on $A$ and arbitrary states on $B$. We can view a state on $B$ as a concrete boundary condition for the TQFT $\mathcal{T}$ (which is precisely interpreted in the context of Levin-Wen type lattice models in \cite{jones2025localtopologicalorderboundary}). This gives us a way to formally interpret the SymTFT idea of gluing various boundary conditions onto the physical boundary (for example, see \cite{PhysRevLett.133.161601}), pictured below.

$$\begin{tikzpicture}
\draw (0,1.5)--(2,1.5);
\draw [dotted,red] (0,.7)--(2,0.7);
\draw (0,0.2)--(2,0.2);
\draw node at (1,1.8) {$\mathcal{B}_{\text{top}}$};
\draw [red] node at (-0.4,0.6) {$\mathcal{B}_{\text{phys}}$};
\draw node at (1,-0.3) {$B_{alt}$};
\end{tikzpicture}$$

How can we leverage the ``boundary of TQFT" interpretation for a state on $B$ to obtain a useful order parameter? From a physical perspective, the primary order parameters for boundaries of 2+1D TQFTs are \textit{algebra objects} internal to the bulk unitary modular tensor category $\mathcal{A}:=\mathcal{Z}(\mathcal{C})$ \cite{PhysRevB.108.075105, PhysRevLett.133.161601, PhysRevB.111.054432}. The idea is that bulk topological defects in $\mathcal{T}$ can be pushed onto the boundary, making the boundary defects into a module category for $\mathcal{A}$. Furthermore, the GNS representation $L^{2}(B,\phi)$ of the state itself provides a distinguished object in this module category. Taking the internal hom provides an algebra object\footnote{we usually assume our states are pure on $B$, which is equivalent to this algebra being connected}  $H_{\phi}\in \mathcal{A}$, which should be an invariant under any natural notion of equivalence of boundaries (and in particular, we show is invariant under symmetric finite-depth circuit equivalence in Corollary \ref{thm:finitedepth}). 

If the algebra object $H_{\phi}$ is commutative (in which case we say that the state is \textit{local}), then the distinguished object can be thought of ``vacuum" defect, and the module category acquires the structure of a unitary fusion category as a quotient of $\text{DHR}(B)$. If $H_{\phi}$ is a Lagrangian algebra object, we say $\phi$ is \textit{topological} (usually referred to in the physics literature as ``gapped"), and the Lagrangian algebra describes the \textit{topological order} of the symmetric state. We say a state is \textit{gapless} if $H_{\phi}$ is not Morita equivalent to a \textit{Lagrangian algebra}. While we do not have any Hamiltonian at hand to rigorously talk about gaplessness (or gapped, for that matter), we believe this word captures what most physicists mean when they use the term, since these states cannot be realized as defects of a topological state (at least, not a state topological with respect to $\mathcal{C}$). This is compatible with various notions in the literature of categorical symmetry protected topological order \cite{GarreRubio2023classifyingphases, PhysRevLett.133.161601, MR4861493}. 

Topological states realizing any given symmetric topological order are fairly easy to construct (for example, see \ref{ex:buildingtopstate}). However, gapless states, which we expect to have interesting universality classes (e.g. CFTs), are much more difficult to explicitly write down and analyze. Thus any abstract criterion that guarantees a state is gapless is very valuable, and the search for such criteria is one of the main motivations for studying categorical symmetry in the first place.

There are two general results from the physics literature which achieve this end: a categorical anomaly enforced gaplessness theorem \cite{TW24} and a version of Kramers-Wannier duality invariance \cite{PRXQuantum.4.020357,PRXQuantum.5.010338, 10.21468/SciPostPhys.16.6.154}. One of the main results of this paper is a formulation and proof of versions of these ideas on the lattice in the infinite volume limit.

Associated to a topological state $\phi$ is a module category $\mathcal{M}_{\phi}$ over the symmetry category $\mathcal{C}$ corresponding to the Lagrangian algebra $H_{\phi}$. We show that if $\phi$ is topological, then $\phi$ decomposes as a convex combination of inequivalent pure states on $A$, indexed by the simple objects of $\mathcal{M}_{\phi}$. In particular, a pure symmetric state on $A$ is topological if and only if the associated module category is rank one, in which case $\mathcal{C}$ has a fiber functor. This leads us to the following version anomaly-enforced gaplessness theorem, which could be viewed as a kind of categorical Lieb-Schultz-Mattis theorem (see Corollary \ref{cor:LSM}).

\begin{thmalpha}(Anomaly enforced gaplessness I). Let $A\subseteq B$ be a physical boundary subalgebra with symmetry category $\mathcal{C}$, and let $\phi$ be a pure symmetric state on $A$.  If $\mathcal{C}$ has no fiber functor, then $\phi$ is gapless.
\end{thmalpha}

\noindent We then utilize some technical results from the theory of C*-algebras to prove such states always exist in the setting of fusion spin chains (see Corollary \ref{cor:existence of gapless}).

\begin{thmalpha}\label{thm:gaplessE} If $B$ is a fusion spin chain and $B\subseteq A$ is a physical boundary inclusion, there exists a symmetric pure state on $A$. Thus, if $\mathcal{C}$ has no fiber functor, then there exists a symmetric gapless state on $A$.
\end{thmalpha}

\bigskip

The above theorem shows that for any of the Haagerup-Izumi \cite{MR1832764,MR2837122, IzumiUnpublished} categories or even the extended Haagerup fusion categories \cite{MR2979509, MR4598730}, there exists a gapless state on a quasi-local algebra. Following the standard physics mantra, the RG fixed point of this state, enforcing categorical symmetry, is expected to be a CFT with the same categorical symmetry. There has been significant effort spent looking for a CFT with Haagerup symmetry by searching for a symmetric critical Hamiltonian \cite{PhysRevLett.128.231602, PhysRevLett.128.231603, PhysRevLett.134.191602}, while our result provides a purely state-based alternative (though it is non-constructive). It would be very interesting to have a concrete expression for the state constructed about so that one could try to get their hands on an actual CFT in the universality class.

We also investigate Kramers-Wannier type dualities, which have emerged recently as objects of interest in their own right \cite{COGN,PRXQuantum.4.020357,jones2025quantumcellularautomatacategorical, JoLi24, jones2024dhr}. Given a physical boundary subalgebra $B\subseteq A$, a (self)duality can be described by a bounded spread automorphism $\alpha:B\rightarrow B$. Composing with the conditional expectation $E:A\rightarrow B\subseteq A$, the composition $\widetilde{\alpha}:=\alpha\circ E: A\rightarrow A$ induces a map on the state space of $A$ that is bijective on the set of symmetric states.

One of the main results of \cite{jones2024dhr} is that conjugation by $\alpha$ induces a braided autoequivalence of $\text{DHR}(B)$. We show that if a state $\phi$ on $B$ is invariant (or more generally covariant), then the induced autoequivalence of $\text{DHR}(B)$ must fix the algebra objects $H_{\phi}$. We call a duality \textit{anomalous} if its action on $\text{DHR}(B)$ does not fix any Lagrangian algebras. By the main result of \cite{jones2025quantumcellularautomatacategorical}, this is equivalent to the inability to extend $\alpha$ is a bounded spread isomorphism on \textit{any} physical boundary subalgebra.

There are many examples of anomalous dualities, the most famous of which is Kramers-Wannier duality. In this case $\mathcal{C}=\text{Vec}(\mathbbm{Z}/2\mathbbm{Z})$, and the induced action on the center is the $e\leftrightarrow m$ swap autoequivalence, which leaves neither of the two Lagrangian algebras $1+e$ and $1+m$ invariant. We discuss a handful of other examples that generalize the Kramers-Wannier example in Section \ref{sec:duaofstates}. The significance of anomalous dualities lies in the following (see Corollary \ref{cor:anomalousLagrangian}).

\begin{thmalpha}(Anomaly enforced gaplessness II) Let $B\subseteq A$ is a physical boundary subalgebra and $\alpha:B\rightarrow B$ an anomalous bounded spread isomorphism. If $\phi$ is a connected symmetric state on $A$ such that $\phi|_{B}$ is covariant under $\alpha$, then $\phi$ is gapless. Furthermore, if $B$ is a simple C*-algebra, such states always exist.
\end{thmalpha}

\subsection{Remarks}
As mentioned above, our framework derives both technically and philosophically from the theory of subfactors \cite{MR0696688, MR1642584}. A unitary fusion category always acts (uniquely) on the  hyperfinite II$_1$  factor by Popa  - see e.g. the discussion in \cite{MR4419534}.
Hence, there is no loss of generality in studying fusion categories and fusion modules through factors and subfactors.
Note here also that a fusion module over a unitary fusion category is always equivalent to a unitary fusion module by \cite{MR4616673, ciamprone2025weakquasihopfalgebrasctensor}.
So subfactors suffice to understand unitary fusion categories and all their fusion modules
and related notions of $\alpha$ induction in purely categorical settings.
The basic tenet of subfactor theory is that symmetries are encoded in inclusions of von Neumann algebras.
There is a very similar story for categorical symmetries on C*-algebras via properties of their inclusions. 

These ideas date back
to \cite{MR985307, MR1234394}, where properties of group actions on prime C$^*$ algebras could be recovered from inclusion of the fixed point algebra in the ambient algebra. We are able to directly leverage the results from this earlier work together with recent results of \cite{IzumiBEK} to find pure invariant states in great generality, which allow us to derive the existence of gapless states with Haagerup symmetry (see Theorem \ref{thm:gaplessE}). Thus the analytical perspective afforded by subfactor theory and its C*-version is a useful new tool in this area. We expect this subfactor perspective to lead to many more new applications.

We also point out that the subfactor perspective suggest a very natural question: given a quasi-local algebra $A$, can we classify physical boundary subalgebras $B\subseteq A$ up to bounded spread equivalence (see Definition \ref{def:bddspreadequiv})? The work \cite{chen2022ktheoretic} can be used to give a partial answer to this question, but unfortunately the K-theoretic arguments used there only classify things up to equivalence by C*-isomorphism, rather than bounded spread equivalence. 

Our paper is formulated with $\mathbbm{Z}$ as the underlying metric space, but nothing is lost by replacing this with $\mathbbm{R}$. Most definitions, theorems, and proofs go through verbatim. In addition, if we require our algebras to be von Neumann algebras, then everything goes through except we should (probably) also add the hypothesis that our states and correspondences are locally normal, which is not assumed in the strictly C*-setting we use here but makes no real technical difference.

It is also natural to ask if our theory works in higher dimensions. The short answer is ``no". The problem is that DHR bimodules for 2-dimensional quasi-local algebras are only symmetric braided tensor category, hence cannot possibly recover the full topological order of a 3+1D bulk TQFT. While there is some progress in this direction \cite{jones2025holographybulkboundarylocaltopological}, we expect there to be an as-yet-unknown satisfactory extension of the DHR category to capture the putative bulk topological order for a physical boundary subalgebra.

Finally, we wish to point out several related works. The MPO story for fusion categorical symmetry emerges from the theory of \textit{tensor networks}, and there is a version of the SymTFT picture in this framework in \cite{PhysRevLett.121.177203}. The theory of fusion categorical symmetry has also been well-developed in the context of (1+1)D rational conformal field theory, see \cite{MR4642115} for a review, and the foundational work of Fuchs, Runkel, and Schweigert \cite{MR1940282, MR2076134, FRSIV}.

\subsection{Structure of the paper}
In Section \ref{Section:quasi-localandDHR}, we introduce formal definitions for the basic objects we are studying, in particular quasi-local algebras and DHR bimodules. In this section, we also attempt to give a more physically-oriented motivation for the abstract definition of DHR bimodules in terms of topological defects in the bulk theory. In Section \ref{sec:localsubalg}, we give the definition for local subalgebra and physical boundary subalgebra for quasi-local algebras, and show how to extract the symmetry category in this situation. In the last two sections establish the results on tensor product symmetries and symmetric states.

We warn the reader that familiarity with the theory of (unitary) fusion categories is required for much of this paper, as well as some knowledge of operator algebras. We refer the reader to \cite{MR3242743, MR3308880, MR3204665} for background on tensor categories and \cite{MR887100,MR1441540, MR3617688} for background on operator algebras with an emphasis on their applications in formalizing the thermodynamic limit of spin systems. A comprehensive reference on subfactor theory is \cite{MR1642584} with \cite{MR4642115} for a recent review oriented towards algebraic quantum field theory. A topic which we make frequent use of are \textit{Hilbert C*-modules and C*-correspondences}. We refer the reader to \cite{MR1325694} or \cite{MR4419534} for a more categorical perspective (see also \cite{moore2017quantummechanicsnoncommutativeamplitudes} for an interesting physical take).

\subsection{Acknowledgements}

The authors would like to thank Lea Bottini, Liang Kong, Pieter Naaijkens, Jake McNamara, David Penneys, Abhinav Prem, Sakura Schafer-Nameki, Sahand Seifnashri, Shu-Heng Shao, Nikita Sopenko, Xiao-Gang Wen, Dominic Williamson, and Xinping Yang for very interesting conversations. The first author was supported  by an Emeritus Fellowship from the Leverhulme Trust. This work was initiated when he visited North Carolina State University, Raleigh. The second author was supported by NSF DMS 2247202, and was inspired by his participation in the KITP program "Generalized Symmetries in Quantum Field Theory: High Energy Physics, Condensed Matter, and Quantum Gravity".

\section{Quasi-local algebras and DHR bimodules}\label{Section:quasi-localandDHR}

In the SymTFT picture for a quantum field theory, we have a decomposition of the system into a topological boundary and physical boundary, coupled via a bulk TQFT. In 1+1D for a fixed time slice, the operators $A_{I}$ localized in a spatial interval $I$ are represented in the SymTFT picture as operators supported in a rectangle which intersects the physical and topological boundaries along I, represented graphically by

$$\begin{tikzpicture}
\filldraw [violet!30]  (4.5,0.5) rectangle (5.5,1.5);
\draw node at (1,1.5) {$\mathcal{F}$};
\draw (0,1)--(2,1);
\draw node at (5,1.8) {$\mathcal{B}_{\text{top}}$};
\draw  (3.7,0.5)--(6.3,0.5);
\draw node at (3,1) {=};
\draw  (3.7,1.5)--(6.3,1.5);
\draw node at (5,0.18) {$\mathcal{B}_{\text{phys}}$};
\draw [dotted] (3.7,0.5)--(3.7,1.5);
\draw [dotted] (6.3,0.5)--(6.3,1.5);
\draw node at (4.2,1) {$\color{violet}A_{I}$};
\draw [violet] (0.5,1)--(1.5,1);
\draw [violet] (0.5,0.9)--(0.5,1.1);
\draw [violet] (1.5,0.9)--(1.5,1.1);
\draw node at (1,0.7) {$\color{violet}A_{I}$};
\end{tikzpicture}$$

Our key observation is that operators localized near intervals of the physical boundary form subalgebras $B_{I}\subseteq A_{I}$, which assemble into a global subalgebra $B\subseteq A$ of the quasi-local algebra of observables which we call the \textit{physical boundary subalgebra}\footnote{this will later have a formal abstract definition}.

$$\begin{tikzpicture}
\filldraw [violet!30]  (4.5,0.5) rectangle (5.5,0.7);
\draw node at (5,1.8) {$\mathcal{B}_{\text{top}}$};
\draw  (4,0.5)--(6,0.5);
\draw  (4,1.5)--(6,1.5);
\draw node at (5,0.18) {$\mathcal{B}_{\text{phys}}$};
\draw node at (5,1) {$\color{violet}B_{I}$};
\draw [dotted] (4,0.5)--(4,1.5);
\draw [dotted] (6,0.5)--(6,1.5);
\end{tikzpicture}$$

Furthermore, for a spin system with local Hamiltonian $H=\sum_{I} H_{I}$, the assertion that the physical boundary should contain the dynamical content of the theory, so that the local terms of the Hamiltonian are contained in the physical subalgebra, i.e. $H_{I}\in B_{I}$.

The canonical examples of symmetries on spin chains are onsite global group symmetries encoded by a homomorphism $G\rightarrow \text{Aut}(A)$, where $A$ is the quasi-local algebra of observables. Here we take $G$ to be finite. The fusion category $\text{Vec}(G)$ is the symmetry fusion category in this context. Then the physical boundary observables are defined by the local subalgebras $B_{I}:=A^{G}_{I}$ consisting of local operators which commute with the on-site symmetry. The physical boundary subalgebra is $B=A^{G}$. 

In this section, we address the question: which subalgebras of the quasi-local algebra $A$ can be interpreted as a physical boundary subalgebra for some SymTFT decomposition? In other words,  which subalgebras of $A$ are reasonable to interpret as the observables of a physical boundary theory, which couples via a bulk TQFT couple to a topological boundary to recover the original spin system?

The main goal of this paper is to propose an answer to this question using the theory of DHR bimodules, introduced in \cite{jones2024dhr}. In order to formulate our definitions, we first set up some background terminology. In particular, we need a good notion of an ``abstract quasi-local algebra" over $\mathbbm{Z}$ which can be used to capture the structure of the physical boundary as well as the entire spin system.

\begin{defn}\label{defn:quasi-localalg} Let $A$ be a unital C*-algebra. A quasi-local structure on $A$ over $\mathbbm{Z}$  consists of an assignment $I\mapsto A_{I}$, of a unital, C*-subalgebra $A_{I}\subseteq A$ to every (finite) interval $I\subseteq \mathbbm{Z}$ subject to the following axioms:

\begin{enumerate}
\item 
(Unital) $A_{\varnothing}=\mathbbm{C}1_{A}$.
\item 
(Isotony) If $I\subseteq J$, then $A_{I}\subseteq A_{J}$.
\item 
(Locality) If $I\cap J=\varnothing$, then $[A_{I},A_{J}]=0$.
\item 
The subalgebra $\vee_{F} A_{I}$ is dense in $A$.
\item 
(weak Haag duality) There exists some $R\ge 0$ such that for every interval $I$, $\{a\in A\ :\ [a,b]=0\ \text{for all b}\ \in A_{J},\ J\cap I=\varnothing\}\subseteq A_{I^{+R}}$.
\end{enumerate}

\end{defn}

\noindent In the above $I^{+R}=\{x\in \mathbbm{Z}\ : d(x,I)\le R\}$, and we interpret $\varnothing^{+R}=\varnothing$.

We can state the weak Haag duality condition a bit more smoothly If $B\subseteq A$ is a subalgebra, we define the centralizer of $B$ in $A$ as

$$Z_{A}(B):=\{a\in A\ :\ [a,B]=0\}.$$

When we have a quasi-local algebra $A$, then for any subset (possibly infinite) $F\subseteq \mathbbm{Z}$, we can define $A_{F}:=\left(\vee_{I\subseteq F} A_{I}\right)^{\|\cdot\|}$ to be the C*-subalgebra generated by the interval subalgebras contained in $F$. Then weak Haag duality can be stated as follows:

For all intervals $I\subseteq \mathbbm{Z}$,

$$Z_{A}(A_{I^{c}})\subseteq A_{I^{+R}}.$$

The physical interpretation of a quasi-local algebra is that the self-adjoint elements of the C*-algebra $A$ are norm limits of bounded local observables of the system. The self-adjoint elements of $A_{I}$ are interpreted as the observables of the system localized in the region $I$.  We note that the failure of more general types of Haag duality is intimately related to the non-triviality of DHR bimodules and realization of the theory at hand on the boundary of a TQFT \cite{CaMa,shao2025additivityhaagdualitynoninvertible}.

We now introduce the type of equivalence relation on quasi-local algebras that we consider.

\begin{defn}
Let $A,A^{\prime}$ be quasi-local algebras over $\mathbbm{Z}$. A \textit{bounded spread isomorphism} is an isomorphism $\alpha:A\rightarrow A^{\prime}$ of C*-algebras such that there exists an $R\ge 0$ with $\alpha(A_{I})\subseteq A^{\prime}_{I^{+R}}$.
\end{defn}

We note that weak Haag duality implies that if $\alpha$ has bounded spread, $\alpha^{-1}$ also has bounded spread. For a tensor product quasi-local algebra (i.e. a ``concrete spin chain"), bounded spread automorphisms are called \textit{quantum cellular automata} \cite{Farrelly2020reviewofquantum}. The idea is that bounded spread isomorphisms uniformly preserve the notion of geometric locality, which is why such isomorphisms are often called \textit{locality preserving}. Such an equivalence will preserve aspects of physics that are locally defined. In addition, as we are mainly interested in the theory in the IR, we expect locality preserving isomorphisms to flow to isomorphisms that are \textit{strictly locality preserving} (i.e. spread 0) in the continuum.

When considering dynamics of quasi-local algebras, in this paper, we are mainly interested in \textit{local dynamics}, which are specified by a Hamiltonian formally written as

$$H=\sum_{I} H_{I}$$

\noindent where 

\begin{enumerate}
    \item 
    the sum is taken over all finite $F\subseteq \Lambda$.
    \item 
    $H_{I}\in A_{I}$.
    \item 
    there is an $R\ge 0$ such that $H_{I}=0$ for $\text{diam}(I)\ge R$.
\end{enumerate}

\noindent Obviously $H$ can not be interpreted as a bounded observable since we essentially never expect this sum to converge, but rather as an (unbounded) derivation on the quasi-local algebra given on the local observable $a$ by

$$\delta_{H}(a):=\sum_{I} [H_{I},a]$$

\noindent Then time evolution of observables is given on local operators $a$ by $*$-automorphisms

$$\sigma^{H}_{t}(a):=\sum^{\infty}_{n=0}\frac{1}{n!}(it\delta_{H})^{n}(a)\approx e^{it H} a\ e^{-itH}\,.$$

\noindent We note the first sum above is guaranteed to converge by strict locality of the Hamiltonian and the Lieb-Robinson bound \cite{MR1441540}. 

A full $1+1$D quantum field theory requires us to specify a dynamics. However, in this paper we will not be focused on dynamics per se, but will instead study states, which should nominally arise as the \textit{ground states} and/or \textit{KMS} (i.e. thermal equilibrium) states of a Hamiltonian. This is motivated by the recent emergence of quantum information, which has emphasized the study of states (and their universality classes) independently of any parent Hamiltonian.

\bigskip

In this paper, we will be focused on 1+1D spin systems, hence we restrict our attention to quasi-local algebras satisfying the following conditions.

\begin{defn}
A quasi-local algebra over $\mathbbm{Z}$ is called an \textit{abstract spin chain} if $A_{I}$ is finite dimensional for each interval $I\subseteq \mathbbm{Z}$.
\end{defn}

Abstract spin chains will be the main example of abstract quasi-local algebras in this paper. We will now review the standard examples that will be relevant to our story.

\noindent \begin{ex}{\textbf{Spin systems}.}\label{ex:ConcreteEx} \noindent The standard examples of abstract spin chains are \textit{concrete spin chains}, which is the setting most directly used by physicists to describe quantum spin chain in the thermodynamic limit \cite{MR1441540}. The quasi-local algebra of concrete spin systems is the infinite tensor product $A:=\otimes_{\mathbbm{Z}} M_{d}(\mathbbm{C})$, with subalgebras $A_{I}:=\otimes_{I} M_{d}(\mathbbm{C})$. It is straightforward to verify the other axioms. 

In the physical version of this story, we typically start with an on-site Hilbert space $V$ of dimension $d$. Then we have a local Hilbert space associated to an interval $I$, defined as $V^{\otimes I}$. The local algebra $A_{I}\cong \otimes_{I} M_{d}(\mathbb{C})$ is naturally identified with all the bounded linear operators on $V^{\otimes I}$. Thus locally, we have a full quantum mechanical picture, wherein our observables are \textit{all} the operators on a local Hilbert space. Sectorization (or failure of the algebra of observables to be \textit{all} of the operators on a Hilbert space) only occurs in the infinite volume limit.  
\end{ex}

\bigskip

\noindent \begin{ex}{\textbf{Symmetric subalgebras from on-site group symmetry}.}\label{ex:SymSubalg} Start with a concrete spin system with a local $d$-dimensional Hilbert space $V$ as discussed in the previous example. A subalgebra of the full quasi-local algebra can be constructed by choosing a finite group $G$ and a faithful unitary representation of $G$ on $V$, which we denote by $\pi:G\rightarrow \text{U}(V)$ (where the latter denotes the unitary group of $V$). Then consider the usual quasi-local algebra of the concrete system $A:=\otimes_{\mathbbm{Z}} M_{d}(\mathbbm{C})$, and letting $G$ act diagonally on the tensor product Hilbert spaces and conjugating by the resulting unitary gives rise to a homomorphism $G\rightarrow \text{Aut}(A)$, which we denote $g\mapsto \alpha_{g}\in \text{Aut}(A)$. This action is on-site, which manifests in the Heisenberg picture as the property that $\alpha_{g}(A_{I})=A_{I}$ for all intervals $I\subseteq \mathbbm{Z}$. Define $B:=A^{G}=\{a\in A\ : \alpha_{g}(a)=a\ \text{for all}\ a\in A\}$.

For each interval $I\subseteq \mathbbm{Z}$, set $B_{I}:=A^{G}_{I}=\{a\in A_{I}\ :\ \alpha_{g}(a)=a\ \text{for all}\ g\in G\}=B\cap A_{I}$. Then $B$ obtains the structure of a quasi-local algebra, which satisfies Haag duality on the nose.

\end{ex}

This example is very important, since it clearly demonstrates how abstract spin systems can vary in significant ways from concrete spin systems before we even pick states or dynamics. Indeed, let $I$ be an interval with $J,K$ disjoint sub-intervals with the property that $J\cup K=I$. Then by construction, it is easy to see that for a concrete spin system $A_{I}\cong A_{J}\otimes A_{K}$. However, for the systems described above, we have $A_{J}\otimes A_{K}\subseteq A_{I}$ but we do not, in general, have equality. This is a basic fact from the representation theory of finite groups that an intertwiner between tensor product representations does not, in general, factorize as intertwiners between tensor factors. This can be thought of as a form of ``algebraic entanglement". This algebraic entanglement in-fact encodes intricate mathematical structure, and is ultimately responsible for the possibility of describing SymTFT decompositions algebraically.

We now move on to discuss a more general class of examples that can be utilized to formulate the standard examples of MPO symmetries. They also arise naturally (and have been studied for nearly 40 years) in the context of subfactor theory \cite{MR996454,MR1642584, MR1424954,MR1334479}.

\bigskip

\noindent \begin{ex}{\textbf{Multi-fusion spin chains}.}\label{ex:FusionSpin} Recall a unitary multi-fusion category is a rigid C*-tensor category with finitely many isomorphism classes of simple objects (but not necessarily a simple unit, see \cite[Section 2.2]{HP23} for a comprehensive overview of unitary (multi)-fusion categories). Let $\mathcal{C}$ be a unitary multi-fusion category, and let $X\in \mathcal{C}$. We typically assume $X$ is self-conjugate (for convenient access to subfactor theory) and strongly tensor generating, meaning there exists an $n>0$ such that every isomorphism class of simple object $Y\in \mathcal{C}$ occurs as a summand of $X^{\otimes n}$. Then we can build an abstract spin chain (over $\mathbbm{Z}$) as follows (for more details, see \cite{jones2024dhr,JoLi24,jones2025quantumcellularautomatacategorical}. For any interval $I$, define the algebra

$$A_{I}:=\text{End}_{\mathcal{C}}(X^{\otimes I}),$$

\noindent so that an element $f\in A_{I} $ is represented in the standard monoidal graphical calculus by 

$$\begin{tikzpicture}
\draw  (.1,1)--(0.1,1.3);
\draw  (.2,1)--(0.2,1.3);
\draw node at (0.51,1.2) {$\cdots$};
\draw  (.8,1)--(0.8,1.3);
\draw  (.9,1)--(0.9,1.3);
\draw  (0,0) rectangle (1,1);
\draw node at (0.5,0.5) {$f$};
\draw  (.1,-0.3)--(0.1,0);
\draw  (.2,-0.3)--(0.2,0);
\draw node at (0.51,-0.2) {$\cdots$};
\draw  (.8,-0.3)--(0.8,0);
\draw  (.9,-0.3)--(0.9,0);
\draw node at (0.5,-0.7) {$X^{\otimes I}$};
\draw node at (1.5,0.5) {$\in A_{I}$};
\end{tikzpicture}$$

\noindent where $X^{\otimes I}=X^{\otimes |I|}$, but we identify the tensor factors of $X$ with the points in the interval $I$. Then for $I\subseteq J$, we can define the natural inclusion 

$$A_{I}\hookrightarrow A_{J}$$

$$f\mapsto 1^{\otimes J<I}_{X}\otimes f\otimes \dots \otimes 1^{J>I}_{X}. $$

\noindent In graphical calculus this is expressed via

$$\begin{tikzpicture}
\draw  (.1,1)--(0.1,1.3);
\draw  (.2,1)--(0.2,1.3);
\draw node at (0.51,1.2) {$\cdots$};
\draw  (.8,1)--(0.8,1.3);
\draw  (.9,1)--(0.9,1.3);
\draw  (0,0) rectangle (1,1);
\draw node at (0.5,0.5) {$f$};
\draw  (.1,-0.3)--(0.1,0);
\draw  (.2,-0.3)--(0.2,0);
\draw node at (0.51,-0.2) {$\cdots$};
\draw  (.8,-0.3)--(0.8,0);
\draw  (.9,-0.3)--(0.9,0);
\draw node at (0.5,-0.7) {$X^{\otimes I}$};
\draw node at (1.4,0.5) {$\mapsto$};

\draw  (1.76,-0.3)--(1.76,1.3);
\draw  (1.86,-0.3)--(1.86,1.3);
\draw  (2.1,1)--(2.1,1.3);
\draw  (2.2,1)--(2.2,1.3);
\draw node at (2.51,1.2) {$\cdots$};
\draw  (2.8,1)--(2.8,1.3);
\draw  (2.9,1)--(2.9,1.3);
\draw  (2,0) rectangle (3,1);
\draw node at (2.5,0.5) {$f$};
\draw  (2.1,-0.3)--(2.1,0);
\draw  (2.2,-0.3)--(2.2,0);
\draw node at (2.51,-0.2) {$\cdots$};
\draw  (2.8,-0.3)--(2.8,0);
\draw  (2.9,-0.3)--(2.9,0);
\draw  (3.1,-0.3)--(3.1,1.3);
\draw  (3.2,-0.3)--(3.2,1.3);
\draw node at (2.5,-0.7) {$X^{\otimes J}$};
\end{tikzpicture}$$

\noindent Then define $A:=\text{colim}_{I} A_{I}.$ in the category of unital C*-algebras. Identifying each $A_{I}$ with its image in $A$ yields an abstract spin system over $\mathbbm{Z}$. We denote the abstract spin chain constructed this way by $A(\mathcal{C},X)$, and call it a \textit{fusion spin chain}. In the previous example built from a unitary representation of a group $\pi:G\rightarrow \text{U}(V)$, the associated spin chain $B$ we discussed is isomorphic to the fusion spin chain $A(\text{Rep}(G), \pi)$.

\end{ex}
\bigskip

\bigskip
\subsection{DHR bimodules.} In this section, we will review the basic theory of \text{DHR bimodules}. The main idea we are proposing in this paper is that DHR bimodules of a physical boundary algebra encode the bulk topological order, and hence the bulk TQFT, as explained in the introduction. First, we will give formal definitions of DHR bimodules and recall basic theorems and results. Then we will sketch a physical argument demonstrating how the bulk topological order actually corresponds to the category of DHR bimodules, which is an expansion of the ideas in \cite{jones2025localtopologicalorderboundary}.

Associated to an abstract quasi-local algebra $B$ is a C*-braided tensor category $\text{DHR}(B)$ of DHR bimodules \cite{jones2024dhr}. Objects in $\text{DHR}(B)$ consist of $B$-$B$ correspondences, which have projective bases localized in any sufficiently large interval $I$.

To be more precise, a DHR bimodule is a type of $B$-$B$ \textit{correspondence}, which is an algebraic $B$-$B$ bimodule with a right $B$-valued inner product (denoted $\langle x\ |\ y\rangle)$, satisfying a collection of axioms similar in spirit to a Hilbert space, but with the role of scalars played by the (right) action of the C*-algebra $B$ (see \cite{MR1325694} or \cite{moore2017quantummechanicsnoncommutativeamplitudes} for a more physically oriented interpretation of Hilbert modules, and \cite[Section 2.2]{MR4419534} for formal definitions of correspondences). We can think of a correspondence as a sector of quantum channels (or more generally, quantum operations) from $B$ to $B$, with a vector $x\in X$ giving rise to the completely positive map $\psi_{x}(b):=\langle x\ | bx\rangle\in B$. The left and right actions of $B$ correspond to pre and post local perturbations of the channel (for a more in-depth description of this picture, see \cite[Section 3]{jones2024dhr}).

A projective basis for a correspondence is a finite collection $\{b_{i}\}\subseteq X$ such that for any $x\in X$, $$x=\sum b_{i} \langle b_{i}\ | x\rangle .$$
A projective basis is localized in an interval $I$ if for all local operators $a\in B_{I^{c}}$ localized in the complement of the interval $I$, 
$$ab_{i}=b_{i}a.$$

\begin{defn} Let $B$ be a quasi-local algebra over $\mathbbm{Z}$. A $B$-$B$ correspondence is a \textit{DHR bimodule} if there exists an $R\ge 0$ such that for any interval $I$ with $|I|\ge R$, there is a projective basis localized in $I$.
\end{defn}

The collection of all DHR bimodules, $\text{DHR}(B)$, naturally forms a C*-tensor category, with morphisms adjointable bimodule intertwiners and the tensor product given by the relative tensor product of correspondences. This tensor category has a natural unitary braiding. We describe this result with the following theorem.

\begin{thm}\cite{jones2024dhr} If $B$ is a quasi-local algebra over $\mathbbm{Z}$, $\text{DHR}(B)$ naturally forms a unitary braided tensor category 
\end{thm}

The unitary braiding $\sigma_{X,Y}: X\boxtimes_{B} Y\cong Y\boxtimes_{B} X$ is defined by picking two disjoint intervals $I<<J$, and choosing projective bases $\{b_{i}\}$ localized in $I$ and $\{c_{j}\}$ localized in $J$ respectively. Then define
$$\sigma_{X,Y}(b_{i}\boxtimes c_{j}):=c_{j}\boxtimes b_{i}.$$

Since the set $\{b_{i}\boxtimes c_{j}\}$ is projective basis for $X\boxtimes_{B} Y$ this uniquely extends to a right $B$-module intertwiner. It is not immediately obvious, but $\sigma_{X,Y}$ actually extends to a unitary intertwiner of correspondences which does not depend on the choice of intervals $I<<J$ or on basis elements. It is also shown that this gives a braiding on the category (i.e. satisfies the hexagon axioms \cite{MR3242743}). Despite the simple form given here, this braiding can be highly non-trivial. Indeed, fairly generically in 1+1D, the $\text{DHR}$ category is non-degenerate, and is even a Drinfeld center. Indeed, we have the following result, identifying the DHR bimodule category of fusion spin chains.

\begin{thm}{\cite[Theorem C]{jones2024dhr}}.\label{thm:Drinfeldcenter-fusionspin}
    If $A(\mathcal{C},X)$ is a fusion spin chain associated to the fusion category $\mathcal{C}$ and self-dual strong tensor generator $X$ (see Example \ref{ex:FusionSpin}), then $\text{DHR}(A(\mathcal{C},X))\cong \mathcal{Z}(\mathcal{C}))$ as unitary braided tensor categories.
\end{thm}

Note that in particular, for the symmetric subalgebras $B=A^{G}$ described in Example \ref{ex:SymSubalg}, $\text{DHR}(B)\cong \mathcal{Z}(\text{Rep}(G))$.

\subsubsection{Physical picture for DHR bimodules}\label{subsec:physpic} We will now provide a sketch that connects the DHR bimodules of the physical boundary algebra to the bulk topological order. This section is not intended to prove  theorems, but give a physical story of why you might expect such an abstract structure as a DHR bimodule of the boundary algebra to encode the topological order of the bulk. Let us assume we are working with a $1+1D$ theory $\mathcal{F}$, so that the bulk TQFT is a 2+1D TQFT, determined by its unitary modular category of point defects $\mathcal{A}$. 

If $A$ denotes the quasi-local algebra for the whole theory, then the quasi-local subalgebra $B$ associated to the physical boundary is built from operators localized near the boundary. In particular, for an interval $I$ along the boundary, $B_{I}$ consists of operators localized in a region intersecting the boundary at $I$, so that $a\in B_{I}$ is represented graphically by

$$
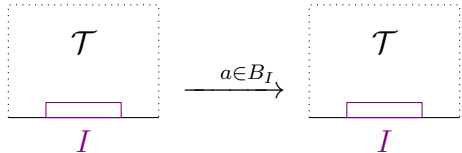
\begin{tikzpicture}
\draw  (4,0.5)--(6,0.5);
\draw  node at (5,1.5){$\mathcal{T}$};
\draw node at (5,0.18) {$\color{violet}I$};
\draw [violet]  (4.5,0.5) rectangle (5.5,0.7);
\draw [dotted] (4,0.5)--(4,2)--(6,2)--(6,0.5);
 \draw node at (7, 1){$\xrightarrow{\ \ \ a\in B_{I} }$};
 \draw  (8,0.5)--(10,0.5);
\draw  node at (9,1.5){$\mathcal{T}$};
\draw node at (9,0.18) {$\color{violet}I$};
\draw [violet]  (8.5,0.5) rectangle (9.5,0.7);
\draw [dotted] (8,0.5)--(8,2)--(10,2)--(10,0.5);
\end{tikzpicture}$$

Now, suppose we have a topological point defect, labelled by the simple object $x\in \mathcal{A}$. We define the vector space $X=\bigcup_{I} X_{I}$, where the $X_{I}$ operators localized near the (physical boundary) as before, but mapping from the bulk with vacuum to the bulk with an $X$ defect. An operator $\xi\in X_{I}$ can be visualized with the picture

$$
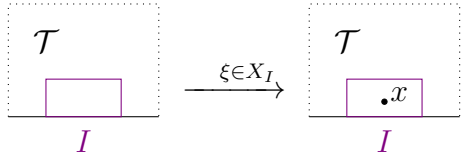
\begin{tikzpicture}
\draw node at (4.5,1.5) {$\mathcal{T}$};
\draw node at (8.5,1.5) {$\mathcal{T}$};
\draw  (4,0.5)--(6,0.5);
\draw node at (5,0.18) {$\color{violet}I$};
\draw [violet]  (4.5,0.5) rectangle (5.5,1);
\draw [dotted] (4,0.5)--(4,2)--(6,2)--(6,0.5);
 \draw node at (7, 1){$\xrightarrow{\ \ \ \xi\in X_{I} }$};
 \draw  (8,0.5)--(10,0.5);
\filldraw [black] (9,0.7) circle (1pt);
\draw node at (9.2,0.8){$x$};
\draw node at (9,0.18) {$\color{violet}I$};
\draw [violet]  (8.5,0.5) rectangle (9.5,1);
\draw [dotted] (8,0.5)--(8,2)--(10,2)--(10,0.5);
\end{tikzpicture}$$

In general, the defect has a fixed position, and the region of support for localized operators may or may not contain it. In the picture above, we have featured the case where regions intersecting the interval $I$ along the boundary \textit{do} contain the defect. 

\begin{remark} Our convention below is to draw a violet rectangle to indicate the support of a local operator $a\in B$ since we have to take this into account when composing things, but \textit{not} continuing to draw the support region for a vector in $X$
\end{remark}

Now we claim that $X$ has the structure of a right $A$ module. Indeed for any $a\in B$, we define the right $A$ action

$$\xi\triangleleft a:= \xi \circ a
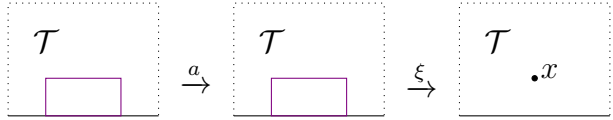
$$

$$\begin{tikzpicture}
\draw node at (4.5,1.5) {$\mathcal{T}$};
\draw node at (7.5,1.5) {$\mathcal{T}$};
\draw node at (10.5,1.5) {$\mathcal{T}$};
\draw  (4,0.5)--(6,0.5);
\draw [violet]  (4.5,0.5) rectangle (5.5,1);
\draw [dotted] (4,0.5)--(4,2)--(6,2)--(6,0.5);
\draw node at (6.5, 1){$\xrightarrow{a}$};
\draw  (7,0.5)--(9,0.5);
\draw [violet]  (7.5,0.5) rectangle (8.5,1);
\draw [dotted] (7,0.5)--(7,2)--(9,2)--(9,0.5);
 \draw node at (9.5, 1){$\xrightarrow{\xi}$};
 \draw  (10,0.5)--(12,0.5);
\filldraw [black] (11,1) circle (1pt);
\draw node at (11.2,1.1){$x$};
\draw [dotted] (10,0.5)--(10,2)--(12,2)--(12,0.5);
\end{tikzpicture}$$

\noindent $X$ also has the structure of a $B$-valued inner product

$$\langle \xi | \eta\rangle := \xi^{\dagger}\circ\eta.$$

To define the left action, we need to consider the location of the defect. If $a$ is localized in a rectangle away from $x$, the post composition makes sense, and we could define $a\triangleright \xi=a\circ \xi$. In general, however, we have to ``move the defect out of the way". We consider a localized unitary $U_{+}$ that maps the Hilbert space with the defect localized near the boundary to the Hilbert space with the $x$-defect localized far from the boundary, and define

$$a\triangleright \xi:= U^{\dagger}_{+}\circ a\circ U_{+}(\xi).
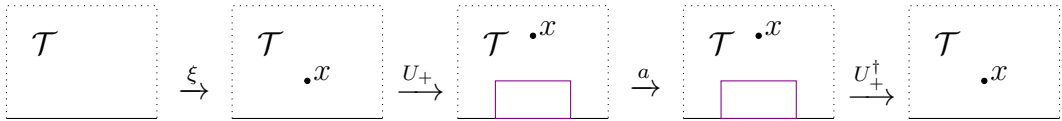
$$

$$\begin{tikzpicture}
\draw node at (4.5,1.5) {$\mathcal{T}$};
\draw node at (7.5,1.5) {$\mathcal{T}$};
\draw node at (10.5,1.5) {$\mathcal{T}$};
\draw node at (13.5,1.5) {$\mathcal{T}$};
\draw node at (16.5,1.5) {$\mathcal{T}$};
\draw  (4,0.5)--(6,0.5);
\draw [dotted] (4,0.5)--(4,2)--(6,2)--(6,0.5);
\draw node at (6.5, 1){$\xrightarrow{\xi}$};
\draw  (7,0.5)--(9,0.5);
\filldraw [black] (8,1) circle (1pt);
\draw node at (8.2,1.1){$x$};
\draw [dotted] (7,0.5)--(7,2)--(9,2)--(9,0.5);
 \draw node at (9.5, 1){$\xrightarrow{U_{+}}$};
 \draw  (10,0.5)--(12,0.5);
\filldraw [black] (11,1.6) circle (1pt);
\draw node at (11.2,1.7){$x$};
\draw [dotted] (10,0.5)--(10,2)--(12,2)--(12,0.5);
\draw  (13,0.5)--(15,0.5);
\filldraw [black] (14,1.6) circle (1pt);
\draw node at (14.2,1.7){$x$};
\draw [dotted] (13,0.5)--(13,2)--(15,2)--(15,0.5);
 \draw node at (12.5, 1){$\xrightarrow{a}$};
\draw [violet]  (13.5,0.5) rectangle (14.5,1);
\draw [violet]  (10.5,0.5) rectangle (11.5,1);
 \draw  (16,0.5)--(18,0.5);
\filldraw [black] (17,1) circle (1pt);
\draw node at (17.2,1.1){$x$};
 \draw node at (15.5, 1){$\xrightarrow{U^{\dagger}_{+}}$};
\draw [dotted] (16,0.5)--(16,2)--(18,2)--(18,0.5);
\end{tikzpicture}$$

Note that since the defect is topological, this won't depend on the choice of transport unitary $U_{+}$. Clearly the left action of $B$ commutes with the right action $B$ giving us a $B$-$B$ correspondence (after completing). 

It remains to make sense of the localized basis condition, which is crucial to constructing the braiding on $\text{DHR}(B)$. For this, we assume that the physical boundary absorbs bulk defects. In other words, we'd expect the defect-near-the-boundary Hilbert space $H_{x}$ to be unitarily equivalent to a finite direct sum of the physical boundary Hilbert space with no bulk defects ($H_{\mathbbm{1}}$) locally. This implies in particular that we would have operators $v_{i}: H_{\mathbbm{1}}\rightarrow H_{x}$ such that $\sum_{i} v_{i}\circ v^{\dagger}_{i}=\text{Id}_{H_{x}}$. But viewing $v_{i}\in X$, this is precisely the statement 

$$\text{Id}_{X}=\sum_{i} v_{i}\langle v_{i}\ |\ $$

\noindent which is the projective basis condition. But the $v_{i}$ can be localized anywhere, since they can be transported under the defect transport unitaries.  This yields a DHR bimodule.

\section{Local subalgebras}\label{sec:localsubalg} In this section, we will introduce the notion of a local subalgebra $B$ of an abstract spin system $A$. This is motivated by the SymTFT picture, and partly captures the topological nature of the decomposition. 

Let $A$ be the quasi-local algebra for our original theory, and $B$ the subalgebra of quasi-local observables localized near the physical boundary. The first thing to note is that this subalgebra itself inherits a quasi-local structure. We formalize this with the following definition.

\begin{defn}{Inclusions of quasi-local algebras}\label{def:inclusions} If $A$ is a quasi-local algebra, a unital C*-subalgebra $B\subseteq A$ is called a \textit{quasi-local subalgebra} if the union of the algebras $B_{I}:=B\ \cap A_{I}$ is norm dense in $B$.
\end{defn}

\noindent In this case, the assignment $I\mapsto B_{I}$ equips $B$ with the structure of a quasi-local algebra. We also say $B\subseteq A$ is an inclusion of quasi-local algebras. 

Returning to the SymTFT picture, take an interval $I$, and consider operators localized near the physical boundary in the complement of $I$, $B_{I^{c}}$. The operators $A_{I}$ localized in $I$ will necessarily commute with $B_{I^{c}}$. However, since the bulk and gapped boundary are topological, there are no non-trivial operators localized away from the physical boundary, and the support of operators in the topological region is contractible. This suggests that \text{all operators} in $A$ that commute with $B_{I^{c}}$ should actually reside in $A_{I}$. We formalize this below, and weaken it up to bounded spread, to take into account the discrete nature of the lattice.

\begin{defn}{\textbf{Weak relative Haag duality}}\label{def:Haag} A quasi-local subalgebra $B\subseteq A$ satisfies \textit{weak relative Haag duality} if there exists some $R\ge 0$ such that $Z_{A}(B_{I^{c}})\subseteq A_{I^{+R}}$. If we can choose $R=0$, we say $B\subseteq A$ satisfies strong relative Haag duality.
\end{defn}

Note that if $B\subseteq A$ satisfies strong relative Haag duality, then $A$ satisfies strong relative Haag duality. This also implies $B$ satisfies strong Haag duality itself, since $B_{I}=A_{I}\cap B$. We will not typically assume relative strong Haag duality because it is not invariant under bounded spread equivalence, though in some situations we can say more under this hypothesis.

We also note that using our conventions that $A_{\varnothing}=\mathbbm{C}$ and $\varnothing^{+R}=\varnothing$, if $B\subseteq A$ satisfies weak relative Haag duality, then $Z_{A}(B)=\mathbbm{C}1$, which means the inclusion $B\subseteq A$ is irreducible.

Another feature of the physical boundary subalgebra $B\subseteq A$ in a SymTFT decomposition is the existence of a canonical \textit{conditional expectation} of $A$ onto $B$ (see, for example, \cite{MR2391387}). Recall that if $B\subseteq A$, a unital, completely positive map of C*-algebras $E:A\rightarrow B$ is called a conditional expectation if $E=E\circ E$ and $E(b_{1}ab_{2})=b_{1}E(a)b_{2}$. The conditional expectation is called \textit{faithful} if $E(x^{*}x)=0$ implies $x=0$.

In the SymTFT picture, the conditional expectation is obtained by projecting the gapped boundary into the bulk (or, if you prefer, averaging over the fusion category symmetry on the gapped boundary). Note that since $B\subseteq A$ satisfies weak relative Haag duality, then $E(A_{I})\subseteq B_{I^{+R}}$, so any conditional expectation is always bounded spread. 

The final technical assumption we need is that for sufficiently large intervals $I$, we assume that $A_{I}$  generates $A$ as a (right) $B$-module in a particularly nice way. In particular, we require the existence of localized projective bases with respect to the local conditional expectation. Essentially, this ensures that $A$, equipped with the structure of a $B$-$B$ correspondence via $E$, is a DHR bimodule. We encode this in the following definition.

\begin{defn} A conditional expectation $E:A\rightarrow B$ is called \textit{local} if for every sufficiently large interval $I$, there exists a finite subset $\{b_{i}\}\subseteq A_{I}$ such that for all $a\in A$, $$a=\sum_{i} b_{i}E(b^{*}_{i}a).$$
\end{defn}

\begin{defn} A \textit{local subalgebra of A} is a subalgebra $B\subseteq A$ satisfying weak relative Haag duality such that there exists a faithful local conditional expectation $E:A\rightarrow B$.
\end{defn}

We will sometimes say that $B\subseteq A$ is a local inclusion or that $A$ is a local extension of $B$. As mentioned above, the inclusion $B\subseteq A$ is irreducible and by the definition of locality, the inclusion $(B\subseteq A, E)$ has finite Watatani index \cite{MR996807}, which together imply that there is a \textit{unique} conditional expectation \cite{MR4079745}, and thus we can take the existence of $E$ as a property, rather than an actual choice (i.e. structure).

\bigskip

\begin{ex}\textbf{On-site group symmetry}\label{ex:SymsubAlg2}. Recall the case of an on-site unitary representation of the finite group $G$ on a finite dimensional Hilbert space, considered in example \ref{ex:SymSubalg}. We assume the representation of $G$ is faithful (otherwise, replace $G$ with $G$ modulo its kernel). Define the abstract quasi-local algebra $B:=A^{G}=\{x\in A\ :\ \alpha_{g}(x)=x\}$, with local subalgebras $B_{I}:=A^{G}_{I}$.
By construction, $B$ is naturally contained in $A$, and this inclusion restricts to local inclusions $B_{I}\subseteq A_{I}$, yielding an inclusion of quasi-local algebras. That this satisfies relative Haag duality follows, for example, from \cite{jones2025quantumcellularautomatacategorical}. The conditional expectation is given by

$$E(a):=\frac{1}{|G|}\sum_{g\in G} \alpha_{g}(a).$$

\noindent The inclusion $(A^{G}\subseteq A, E)$ is the canonical example of a local inclusion.

\end{ex}

\begin{ex}{\textbf{Categorical inclusions of fusion spin chains}.}\label{ex:fusionspinsubalg} We consider the following example introduced in \cite{jones2025quantumcellularautomatacategorical}. Let $\mathcal{C}$ be a unitary fusion category and $\mathcal{D}$ a unitary, indecomposable multi-fusion category. For convenience, assume both categories are strict. Let $F$ be a dominant tensor functor $F:\mathcal{C}\rightarrow \mathcal{D}$. Let $X\in \mathcal{C}$ be a strong tensor generator and $Y:=F(X)\in \mathcal{D}$. Then we have a natural local inclusion of spin chains

$$A(\mathcal{C},X)\hookrightarrow A(\mathcal{D},Y)$$
defined by sending 
$$f\in \mathcal{C}(X^{\otimes n}, X^{\otimes n})\mapsto c_{X^{n}}\circ F(f)\circ c^{-1}_{X^{n}}\in \mathcal{D}(Y^{\otimes n}, Y^{\otimes n}),$$

\noindent
where $c_{X^{n}}:F(X^{\otimes n})\rightarrow F(X)^{\otimes n}$ is the unitary coherence isomorphism. Dominant tensor functors $F:\mathcal{C}\rightarrow \mathcal{D}$ correspond to connected, commutative $Q$-systems (see Section \ref{sec:Q-sys} for definitions) $A\in \mathcal{Z}(\mathcal{C})\cong \text{DHR}(A(\mathcal{C},X))$. $\mathcal{D}$ can be identified with $\mathcal{C}_{A}$, where we consider $A$ as an algebra in $\mathcal{C}$, take right $A$ modules, and use the commutative half-braiding to equip this category with a unitary tensor structure. Then $F:\mathcal{C}\rightarrow \mathcal{D}$ can be identified with the free module functor $x\mapsto x\otimes A$.

We will see in the next section a generalization of this story: essentially arbitrary local extensions of a quasi-local algebra $B$ correspond to connected, commutative Q-systems in the unitary braided category $\text{DHR}(B)$.
\end{ex}

The following is a natural notion of equivalence between the local subalgebras.
    
\begin{defn}\label{def:bddspreadequiv} Given quasi-local algebras $A,A^{\prime}$ with local subalgebras $B\subseteq A,\ B^{\prime}\subseteq A^{\prime}$ and a \textit{bounded spread equivalence} is a bounded spread isomorphism $\alpha:A\rightarrow A^{\prime}$ such that $\alpha(B)=B^{\prime}$.
\end{defn}

\begin{prop} If $\alpha$ is a bounded spread equivalence between local inclusions $B\subseteq A$ and $B^{\prime}\subseteq A^{\prime}$, then $\alpha|_{B}:B\rightarrow B^{\prime}$ is a bounded spread equivalence between the induced quasi-local structures.
\end{prop}

\begin{proof}
Suppose $\alpha:A\rightarrow A^{\prime}$ has spread at most $R$. Then $\alpha(B_{I})=\alpha(B\cap A_{I})=\alpha(B)\cap \alpha(A_{I})\subseteq B^{\prime}\cap A^{\prime}_{I^{+R}}=B^{\prime}_{I^{+R}}$.
\end{proof}

\begin{thm}\label{thm:transportsubalg}
  Suppose $B\subseteq A$ is a local inclusion, and $\alpha:A\rightarrow A^{\prime}$ is a bounded spread isomorphism of quasi-local algebras, with spread at most $R$. Then $B^{\prime}:=\alpha(B)\subseteq A^{\prime}$ is a local subalgebra. 
\end{thm}

\begin{proof}
    
Pick any interval $I$. Then
  
$$B^{\prime}_{I^{+R}}=B^{\prime}\cap A^{\prime}_{I^{+R}}=\alpha(B)\cap A^{\prime}_{I^{+R}}\supseteq \alpha(B\cap A_{I})=\alpha(B_{I}).$$ 

\noindent Thus, since the collection $\{B_{I}\}$ is dense in $B$, the collection $\{B^{\prime}_{I}\}$ is dense in $B^{\prime}$, so $B^{\prime}\subseteq A^{\prime}$ is quasi-local. 

Now we will see that $B^{\prime}\subseteq A^{\prime}$ satisfies weak relative Haag duality. For any interval $I$, let $a\in Z_{A^{\prime}}(B^{\prime}_{I^{c}})$. Then, since $B^{\prime}_{I^{c}}=\alpha(B)\cap A_{I^{c}}$, $\alpha^{-1}(a)$ centralizes $B\cap \alpha^{-1}(A^{\prime}_{I^{c}})$. But $\alpha^{-1}(A^{\prime}_{I^{c}})\supseteq A_{(I^{+R})^{c}}$, thus

$$\alpha^{-1}(a)\in Z_{A}(B\cap A_{(I^{+R})^{c}})=Z_{A}(B_{(I^{+R})^{c}})\subseteq A_{I^{+R+S}},$$

\noindent where $S$ is the weak Haag duality constant for $A$. Then 

$$a=\alpha(\alpha^{-1}(a))\in A^{\prime}_{I^{+R+S}}.$$

\noindent Thus $B^{\prime}\subseteq A^{\prime}$ satisfies weak relative Haag duality. Finally, it is easy to see that if $E:A\rightarrow B$ is a local conditional expectation, so is $\alpha\circ E\circ \alpha^{-1}: A^{\prime}\rightarrow B^{\prime}$.  

\end{proof}

\begin{remark}
If $B\subseteq A$ and $B\subseteq A^{\prime}$ are both local extensions of the same algebra, then we could require bounded spread isomorphisms such that $\alpha|_{B}=\text{Id}_{B}$. In this case, we do not even need to require $\alpha$ to be bounded spread from $A\rightarrow A^{\prime}$, since this follows automatically from weak Haag duality. We will call such an equivalence of $A$ and $A^{\prime}$ an \textit{equivalence of extensions}.
\end{remark}

\subsection{Q-systems and local extensions}\label{sec:Q-sys}

We will show that there is a correspondence between local extensions of a quasi-local algebra $B$ and connected, commutative (dual) Q-systems in $\text{DHR}(B)$. This will be a crucial component in our definition of physical boundary subalgebra. For general references on Q-systems and related topics, we refer the reader to \cite{MR3308880, 10.1063/5.0071215, MR4419534} with the caveat that different terminology may be used in these references.

\begin{defn} Let $\mathcal{B}$ be a unitary tensor category. Recall a \textit{C*-Frobenius algebra} is an object $A\in \mathcal{C}$ such that

\begin{enumerate}
    \item 
    (Associativity)\ $m:A\otimes A\rightarrow A $ satisfies  $m\circ (m\otimes 1_{A})=m\circ (1_{A}\otimes m)$
    \item 
    (Unital)\ $m\circ (\iota\otimes 1_{A})=m\circ(1_{A}\otimes i)=1_{A}$.
    \item 
    (Frobenius condition)\ $(m \otimes1_{A})\circ (1_{A}\otimes m^{\dagger})=m^{\dagger}\circ m=(1_{A}\otimes m)\circ(m^{\dagger}\otimes 1_{A})$.

\end{enumerate}

\end{defn}

This is simply the technical, unitary version of semi-simple algebra object in $\mathcal{C}$.  A further assumption is \textit{specialness}, i.e.    
$$m\circ m^{\dagger}=\lambda 1_{A}$$ for some scalar $\lambda$.

Note that $i^{\dagger}\circ i=\beta 1_{\mathbbm{1}}$, and thus a special C*-Frobenius algebra has two positive numerical parameters $\lambda$ and $\beta$. It turns out only the product $\lambda \beta$ is an invariant of the algebra up to C*-algebra isomorphism. Indeed, we can always simultaneously rescale the multiplication to normalize the scalar $\lambda$ to be any non-zero scalar we want, but this changes $\beta$ by the inverse.

A special C*-Frobenius algebra is a $\textit{Q}$-system if $\lambda=1$ and is a \textit{dual} $Q$-system if $\beta=1$. There  is a clear bijective correspondence between Q-systems and dual Q-systems, and these are isomorphic as C*-Frobenius algebras. It is shown in \cite{MR4419534} that if $B$ is a C*-algebra with trivial center, inclusions $(B\subseteq A, E)$ where $E$ is a conditional expectation are precisely the dual Q-systems in the unitary tensor category of dualizable $B$-$B$ correspondences.

Now suppose that $\mathcal{B}$ also has a unitary braiding $\{\sigma_{X,Y}:X\otimes Y\cong Y\otimes X\}$ for all $X,Y\in \mathcal{B}$. Then we say a C*-Frobenius algebra $A$ with multiplication $m$ is \textit{commutative} if 

$$ m\circ \sigma_{A,A}=m.$$

\begin{thm} Suppose $B$ has the structure of a quasi-local algebra over $\mathbbm{Z}$ with weak Haag duality. Then any local extension $B\subseteq A$ is a connected commutative dual Q-system in $\text{DHR}(B)$. If $B$ satisfies strong relative Haag duality, every connected commutative dual Q-system is a local extension.
\end{thm}

\begin{proof}
Suppose we have a local extension $B\subseteq A$. Let $E:A\rightarrow B$ be the unique (local) conditional expectation. Then by \cite{MR4419534}, $(B\subseteq A, E)$ is a dual Q-system. By definition $A$ is a DHR correspondences. It remains to check that $A$ is commutative. Let $I<<J$ and choose bases $\{b_{i}\}\subseteq A_{I}$ and $\{c_{j}\}\subseteq A_{J}$ (which exist by our locality hypothesis). Since these are bases localized in $I$ and $J$ respectively, then

$$m\circ \sigma_{A,A}(b_{i}\boxtimes c_{j})=c_{j}b_{i}=b_{i}c_{j}=m(b_{i}\boxtimes c_{j}),$$

\noindent since $b_{i}$ and $c_{j}$ commute in $A$. As $m\circ \sigma_{A,A}=m$ on a projective basis, they agree as morphisms between correspondences.

To show the second part of the theorem, assume that $B$ satisfies strong Haag duality on the nose. Let $Q$ be a commutative dual Q-system in $\text{DHR}(B)$. Then by again by \cite{MR4419534}, $Q$ is a C*-algebra containing $B$ irreducibly with a canonical conditional expectation $E:Q\rightarrow B$.

We need to define a local structure on $Q$. For each interval $I$, set $Q_{I}:=B^{\prime}_{I^{c}}\cap B$. This is clearly a net of algebras (satisfying isotonoy, etc.) but it is not obvious at this stage that locality or (weak) Haag duality is actually satisfied.
Let $I$ and $J$ be disjoint intervals, and consider large intervals $I\subseteq I^{\prime},\ J\subseteq J^{\prime}$ with $I^{\prime}\cap J^{\prime}=\varnothing$ such that the inclusion $(B\subseteq A,E)$ has a projective bases $\{b_{i}\}$ and $\{c_{j}\}$ localized in $I^{\prime}$ and $J^{\prime}$ respectively. Since $B$ satisfies algebraic Haag duality, we see that $Q_{I^{\prime}}=\text{span}\{b_{i}B_{I^{\prime}}\}$ and $Q_{J^{\prime}}=\text{span}\{c_{j}B_{J^{\prime}}\}$. Then for $x\in Q_{I^{\prime}}, y\in Q_{J^{\prime}}$ we have $\langle b_i\ |\ x\rangle_{B}\in B_{I^{\prime}}$ and $\langle c_{j}\ |\ y\rangle_{B}\in B_{J^{\prime}}$, hence

\begin{align*}
xy&=m\left(\sum_{i} b_{i}\langle b_{i}\ |\ x\rangle_{B} \boxtimes \sum_{j} c_{j} \langle c_{j}\ |\ y\rangle_{B}\right)\\
&=m\left(\sum_{i,j} b_{i} \boxtimes c_{j} \langle b_{i}\ |\ x\rangle_{B} \langle c_{j}\ |\ y\rangle_{B})\right)\\
&=m\left(\sum_{i,j}b_{i}\boxtimes c_{j}\right)\langle b_{i}\ |\ x\rangle_{B}\ \langle c_{j}\ |\ y\rangle_{B}\\
&=m\circ \sigma_{Q,Q} \left(\sum_{i,j}b_{i}\boxtimes c_{j}\right)\langle b_{i}\ |\ x\rangle_{B}\ \langle c_{j}\ |\ y\rangle_{B}\\
&=m\circ\left(\sum_{i,j}c_{j}\boxtimes b_{i}\right)\langle b_{i}\ |\ x\rangle_{B}\ \langle c_{j}\ |\ y\rangle_{B}\\
&=m\left(\sum_{j} c_{j} \langle c_{j}\ |\ y\rangle_{B}\boxtimes\sum_{i} b_{i}\langle b_{i}\ |\ x\rangle_{B} \right)\\
&=yx.
\end{align*}

Thus $Q$ satisfies the locality axiom. Finally, by definition $B^{\prime}_{I^{c}}\cap Q=Q_{I}$, hence by locality 

$$Q_{I}\subseteq Q^{\prime}_{I^{c}}\cap Q\subseteq B^{\prime}_{I^{c}}\cap Q=Q_{I},$$

hence $Q_{I}=Q^{\prime}_{I^{c}}\cap Q$. Thus $Q$ in fact satisfies \textit{strong} Haag duality.

\end{proof}

\begin{remark}
The statement of the above theorem seems a bit awkward, since the second part of the theorem requires strong Haag duality. Following through the construction given above by applied to the case that $B$ satisfies weak algebraic Haag duality, we see that a Q-system gives rise to a net of algebras over $\mathbbm{Z}$ that satisfies the locality axiom $[A_{I},B_{J}]=0$ not for arbitrary disjoint intervals, but for intervals $I$ and $J$ that are sufficiently far apart. In some sense, this suggests that these objects are more natural to work with (rather than our strictly local version of quasi-local algebra). However, based on our concrete examples of quasi-local algebras, which in practice seem to be always strictly local, we will stick to this set-up.
\end{remark}

\subsection{Categories of bimodules from local extensions}
 
Recall that if $A$ is a (dual) Q-system in a unitary tensor category $\mathcal{B}$, we can consider the category of right $A$-modules $\mathcal{B}_{A}$. Objects are pairs $(X,m_{X})$, where $m_{X}:X\otimes A\rightarrow X$ is a morphism satisfying ``one-sided" versions of the C*-Frobenius axioms \cite{MR3308880}, i.e. unital, associative, and Frobenius. We also require ``specialness", i.e. 
$$m_{X}\circ m^{\dagger}_{X}=\nu 1_{X}$$

\noindent for some non-zero scalar $\nu$. Morphisms $f:(X,m_{X})\rightarrow (Y,m_{Y})$ are simply intertwining morphisms, i.e. $f\in \mathcal{B}(X,Y)$ such that 
$$f\circ m_{x}=n_{x}\circ (f\otimes 1_{A}).$$ 

\noindent The category of $A$-modules $_{A}\mathcal{B}$ is defined similarly. At this level of generality, $\mathcal{B}_{A}$ is only a left $\mathcal{B}$-module category, and $_{A} \mathcal{B}$ is a right $\mathcal{B}$-module category.

However, the category $_{A}\mathcal{B}_{A}$ of $A$-$A$ bimodules is again a unitary tensor category. Recall a bimodule is a triple $(X,n_{X},m_{X})$ where $n_{X}:A\otimes X\rightarrow X$ is a left $A$-module structure, $m_{X}:X\otimes A\rightarrow X$ is a right $A$-module structure, and these morphisms commute with each other in the sense that 

$$m_{X}\circ (n_{X}\otimes 1_{A})=n_{X}\circ(1_{A}\otimes m_{X}).$$

Morphisms are simultaneous left and right module intertwiners. The relative tensor product of $(X,n_{X},m_{X})$ and $(Y, n_{Y}, m_{Y})$ can be defined as the image of the idempotent 

$$p_{X\otimes_{A}Y}:=\frac{1}{\lambda} (1_{X}\otimes n_{Y})\circ (m^{*}_{X}\otimes 1_{Y})\in \mathcal{B}(X\otimes Y, X\otimes Y),$$ 

\noindent with left and right actions induced from the left and right actions of $A$ $X$ and $Y$ respectively. For more details, see also \cite{MR4419534}.

If $A\in \text{Corr}(B)$ is a (dual) Q-system in the category of $B$-$B$ correspondences, then as above $A$ is a C*-algebra containing $B$, equipped with a conditional expectation $E:A\rightarrow B$. The philosophy is that (dual) Q-systems internal to the category $\text{Cor}(B)$ correspond to \textit{actual} finite index C*-algebra extensions of $B$, and that internal bimodules between the (dual) Q-systems correspond to actual correspondences of the associated C*-algebras.

In particular, the unitary tensor category $_{A} \mathcal{B} _{A}$ embeds fully faithfully into $\text{Corr}(A)$. If $\mathcal{B}\subseteq \text{Corr}(B)$ is a full tensor subcategory of $B$-$B$ correspondences containing $A$, then we have a fully faithful embedding of unitary tensor categories $_{A}\mathcal{B}_{A}\hookrightarrow \text{Corr}(A)$, which is an immediate consequence of the statement that the C* 2-category of C*-algebras and correspondences is Q-system complete \cite{MR4419534}.

Now let $A$ be a commutative, connected (dual) Q-system in a unitary braided tensor category $\mathcal{B}$. Let $\mathcal{B}_{A}$ denote the category of right $A$-modules. Then there are two ways to embed $\mathcal{B}_{A}$ into $_{A}\mathcal{B}_{A}$ using the braiding. More explicitly, for a right $A$-module $X$ with right multiplication denoted $m_{X}:X\otimes A\rightarrow A$ we define the $A$-$A$ bimodules

$$\beta_{+}(X,m_{X})=(X, m_{X}\circ\sigma_{A,X},m_{X})$$
$$\beta_{-}(X,m_{X})=(X, m_{X}\circ \sigma^{-1}_{X,A}, m_{X}).$$

Both $\beta_{+}$ and $\beta_{-}$ are fully faithful embeddings $\mathcal{B}_{A}\hookrightarrow _{A}\mathcal{B}_{A}$ whose images are closed under $\otimes_{A}$ \cite{MR3242743}. We can transport structure to obtain two unitary tensor structures on $\mathcal{B}_{A}$. As it turns out, these are monoidally opposite equivalence, or in other words, as unitary tensor categories

$$\beta_{+}(\mathcal{B}_{A})\cong \beta_{-}(\mathcal{B}_{A})^{mp},$$

\noindent where the superscript $\textit{mp}$ denotes the monoidal opposite. As a convention, we will consider $\mathcal{B}_{A}$ as a unitary tensor category with the $\beta_{+}$ monoidal structure. 

Now assume that $B\subseteq A$ is a local extension, hence is a (dual) connected Q-system in $\text{DHR}(B)$. 
Define 

\begin{enumerate}
\item 
$\mathcal{C}_{+}:=\beta_{+}(\text{DHR}(B)_{A})$
\item 
$\mathcal{C}_{-}:=\beta_{-}(\text{DHR}(B)_{A})$
\item 
$\mathcal{C}_{0}=\mathcal{C}_{+}\cap \mathcal{C}_{-}=\beta_{\pm}(\text{DHR}(B)^{loc}_{A})$. 
\end{enumerate}

Then as described above we have fully faithful embeddings $\mathcal{C}_{\pm}\hookrightarrow \text{Cor}(A)$. We will identify $\mathcal{C}_{\pm}$ with their (replete) images. Our next goal is to characterize these with a DHR-like property.

\begin{defn} Let $B\subseteq A$ be a local inclusion.

\begin{enumerate} 

\item 
$\text{DHR}_{+}(A|B)$ is the full tensor subcategory of $A$-$A$ correspondences with the property that for all sufficiently large intervals, there exists a finite projective basis $\{c_{j}\}$ such that

\begin{enumerate}
\item 
$b c_{j}=c_{j}b$ for all $b\in B_{I^{c}}$
\item 
$ac_{j}=c_{j} a$ for all $a\in A_{J}$ with $J<I$.
\end{enumerate}

\item 
$\text{DHR}_{-}(A|B)$ is the full tensor subcategory of $A$-$A$ correspondences with the property that for all sufficiently large intervals, there exists a finite projective basis $\{c_{j}\}$ such that

\begin{enumerate}
\item 
$b c_{j}=c_{j}b$ for all $b\in B_{I^{c}}$
\item 
$ac_{j}=c_{j} a$ for all $a\in A_{J}$ with $J>I$.
\end{enumerate}

\end{enumerate}

\end{defn}

\begin{thm}
For any local inclusion $A\subseteq B$

\begin{enumerate}
\item 
$\mathcal{C}_{+}=\text{DHR}_{+}(A|B)$.
\item 
$\mathcal{C}_{-}=\text{DHR}_{-}(A|B)$.
\item 
$\mathcal{C}_{0}=\text{DHR}(A)$.
\end{enumerate}

\end{thm}

\begin{proof}
First, we will show $C_{+}\subseteq \text{DHR}_{+}(A|B)$. If $(X,m_{X})\in \mathcal{C}_{+}=\text{DHR}(B)_{A}$, then $\frac{1}{\sqrt{\nu}} m^{\dagger}_{X}:X\rightarrow X\boxtimes_{B} A$ is an embedding of right $A$ modules into the free $A$-module object $X\boxtimes_{B} A$. This module embedding is automatically an embedding of correspondences. Since the desired properties of the basis $\{c_j\}$ passes to summands, it suffices to show the existence of the appropriate $\{c_{j}\}$ for (categorical) free right $A$-modules and then project these basis elements onto the image of $X$.

For the right $A$ module $X\boxtimes_{B} A$ (viewed as a B-A bimodule), for sufficiently large intervals $I$, we can find $B$-$B$-bimodule bases $\{e_{j}\}\subseteq X$ localized in $I$. Since $1$ is a projective basis of $A$ as a right $A$ module, the elements $\{e_{i}\boxtimes 1\}$ form a projective basis for $X\boxtimes_{B}A$ as a $B$-$A$ bimodule. This basis will satisfy the first condition by construction.

For the second condition, we note that for any $a\in A$, the left action of $a$ on $e_{j}\boxtimes 1$ is computed

$$a(e_{j}\boxtimes 1)=\sigma_{A,X}(a\boxtimes e_{j})\in X\boxtimes_{B} A$$

\noindent But if $a\in A_{J}$ with $J<I$, we have $\sigma_{A,X}(a\boxtimes e_{j})=e_{j}\boxtimes a$, thus

\begin{align*}
a(e_{j}\boxtimes 1)&=e_{j}\boxtimes a\\
&=(e_{j}\boxtimes 1)a\\
\end{align*}

\noindent Thus the set $\{c_{j}:=e_{j}\boxtimes 1\}$ satisfy the desired criteria.

Now, for any $X\in \text{DHR}_{+}(A|B)$, then $X$ viewed as a $B$-$B$ correspondence (using the conditional expectation for the right inner product) is an element of $\text{DHR}(B)$. Thus by \cite{MR4419534} $X\in _{A}\text{DHR}(B)_{A}$, the category of $A$-$A$ bimodules internal to $\text{DHR}(B)$. It remains to show that

$$ax=m_{X}\circ \sigma_{A,X}(a\otimes x).$$

\noindent Since both sides are right $A$-modular, it suffices to further assume $x$ is an element of some projective basis set and $a\in A_{I}$ for some finite interval $I$. Pick a basis $\{c_{i}\}$ satisfying the defining property of $\text{DHR}_{+}(A|B)$ for some interval $J>I$. Then 

$$m_{X}\circ \sigma_{A,X}(a\otimes c_{j})=c_{j}a=ac_{j}.$$

\noindent Thus $\text{DHR}_{+}(A|B)\subseteq \mathcal{C}_{+}$.

The proof that $\text{DHR}_{-}(A|B)= \mathcal{C}_{-}$ is directly analogous. It remains to show $\text{DHR}(A)=\mathcal{C}_{0}$. Clearly $\text{DHR}(A)\subseteq \text{DHR}_{+}(A|B)\cap \text{DHR}_{-}(A|B)=\mathcal{C}_{0}$. Now, suppose $X\in \mathcal{C}_{0}$. For sufficiently large intervals $I$, we can find projective bases $\{c_{j}\}$ and $\{c^{\prime}_{j}\}$ both of which centralize $B_{I^{c}}$, with $c_{j}$ centralizing $A_{<I}$ and $c^{\prime}_{i}$ centralizing $A_{>I}$. We claim that in fact $c_{j}$ also commutes with $A_{>I}$. Indeed, we have

$$c_{j}=\sum_{i}c^{\prime}_{i}\langle c^{\prime}_i\ |\ c_{j}\rangle_{A}.$$

\noindent But since both $c_{j}$ and the $c^{\prime}_{i}$ are $B_{I^{c}}$-central, so is $\langle c^{\prime}_i\ |\ c_{j}\rangle_{A}$. But by weak relative Haag duality, this implies $\langle c^{\prime}_i\ |\ c_{j}\rangle_{A}\in A_{I^{+R}}$. In particular, each term $\langle c^{\prime}_i\ |\ c_{j}\rangle_{A}$ is $A_{<I^{+R}}$ central, yielding the desired result. In particular $X\in \text{DHR}(A)$.

\end{proof}

\begin{remark}\label{rem:alphainduction}

As alluded to in the introduction, taking $B$-$B$ bimodules to $A$-$A$ bimodules is related to the corresponding notions in subfactor theory and conformal nets, namely $\alpha$-induction of sectors for inclusions $N \subset M$. Here the type III$_1$ factor $N$ carries a braided system $\chi$ of $N$-$N$ sectors (bimodules or endomorphisms) and $M$ as 
$N$-$N$ sector is in $\chi$ and described by a $Q$-system for $\chi$. In applications to conformal nets, $N$ is a factor obtained from a local factor $N = N(I_0)$ of a  conformally covariant quantum field theoretic net of factors $\{N(I)\}$ indexed by proper intervals $I \subset \mathbb R$ of the real line. The $N$-$N$ system $\chi$ is obtained as restrictions of Doplicher-Haag-Roberts to $N$. The braiding arises as the monodromy of moving the interval to disjoint one where there is relative commutativity and back again. The $Q$ system is commutative when the extended net $\{M(I)\}$ is local. 
In particular $\beta_\pm$ are essentially $\alpha$-induction, corresponding to taking a sector $\lambda$ as a $N$-$N$ bimodule to the $M$-$M$ bimodule
$\lambda \otimes_N M$, where the braiding is employed in two different $\pm$ ways for $M$ to act on the left. Then $\mathcal{C}_{\pm}$ correspond to the induced sectors of $M$ and $\mathcal{C}_{0}$ their intersection are the neutral, ambichiral sectors representing the braided sectors and representation theory of the DHR sectors of the extended system \cite{MR1332979, MR1617550, MR1729094, MR1777347,
MR1785458, MR4642115}.

\end{remark}

\subsection{Physical boundary subalgebras}

In general, if $\mathcal{B}$ is a unitary braided tensor category and $A$ is a commutative (dual) Q-system, there is a tensor functor 

$$\alpha:\mathcal{B}\rightarrow \beta_{-}(\mathcal{B}_{A}),$$
$$\alpha(X):=(X\otimes A, 1_{X}\otimes m).$$

This functor has a canonical \textit{central structure}, meaning there is a \textit{braided tensor functor} to the Drinfeld center 
$$\widetilde{\alpha}: \mathcal{B}\rightarrow \mathcal{Z}(\beta_{-}(\mathcal{B}_{A}))$$

\noindent such that $F\circ \widetilde{\alpha}\cong \alpha$, where $F:\mathcal{Z}(\beta_{-}(\mathcal{B}_{A}))\rightarrow \beta_{-}(\mathcal{B}_{A})$ is the forgetful functor. In the case of a local inclusion $B\subseteq A$, then with  $\mathcal{B}=\text{DHR}(B)$, we obtain the braided monoidal functor

$$\widetilde{\alpha}: \text{DHR}(B)\rightarrow \mathcal{Z}(\text{DHR}_{-}(A|B)).$$

This leads us to our definition of physical boundary subalgebra of a quasi-local algebra $A$.

\begin{defn}
Let $A$ be a quasi-local algebra over $\mathbbm{Z}$. A local subalgebra $B\subseteq A$ is called a \textit{physical boundary algebra} if 

\begin{enumerate}
    \item 
$\text{DHR}(B)$ is fusion
\item The braided tensor functor $\widetilde{\alpha}:\text{DHR}(B)\rightarrow \mathcal{Z}(\text{DHR}_{-}(A|B))$ is an equivalence. 
\end{enumerate}

\noindent In this case, we define the symmetry category $\mathcal{C}:=\text{DHR}_{-}(A|B)=\beta_{-}(\text{DHR}(B)_{A})$.
\end{defn}

Physical boundary algebras allow us to complete the SymTFT picture. $B$ describes the physical boundary while the 2+1D bulk TQFT is defined by the unitary modular tensor category $\text{DHR}(B)=\mathcal{Z}(\mathcal{C})$, which can be realized (for example)  as a Turaev-Viro theory \cite{MR1292673}. The gapped boundary is the one associated to the fusion category $\mathcal{C}$, arising from the Lagrangian algebra $A$. The fusion category $\mathcal{C}$ realizes the category of topological point defects on the gapped boundary, and implements the categorical symmetry.

\begin{ex}{\textbf{Spatial realizations, or: MPO symmetries without the symmtries}.}\label{ex:MPO} Let $\mathcal{C}$ be a unitary fusion category and let $\mathcal{M}$ be an indecomposable (right) semi-simple $\mathcal{C}$-module category. Then consider the dual fusion category $\mathcal{D}:=\mathcal{C}^{*}_{\mathcal{M}}$ consisting of $\mathcal{C}$-module endofunctors. Then by construction $\mathcal{M}$ is a left $\mathcal{D}=\mathcal{C}^{*}_{\mathcal{M}}$-module category, and thus acts on $\mathcal{M}$. This gives us a dominant unitary tensor functor $F: \mathcal{D}\rightarrow \text{End}(\mathcal{M})$, where $\text{End}(\mathcal{M})$ is the indecomposable, Morita trivial unitary multi-fusion category of all $\dagger$-endofunctor on $\mathcal{M}$. If $n=\text{rank}(\mathcal{M})$, then as multi-fusion categories $\text{End}(\mathcal{M})\cong \text{Mat}_{n}(\text{Hilb}_{f.d.})$.

Picking a strong tensor generator $X\in \mathcal{D}=\mathcal{C}^{*}_{\mathcal{M}}$, we obtain a categorical inclusion $A(\mathcal{D}, X)\hookrightarrow A(\text{Mat}_{n}(\text{Hilb}_{fd}), H)$, where $H=F(X)$. We call these categorical inclusions \textit{spatial realizations} \cite{jones2025quantumcellularautomatacategorical}, since the local algebras in $A(\text{Mat}_{n}(\text{Hilb}_{fd}), H)$ have an interpretation as all the operators on constrained subspaces of a tensor product Hilbert space. Indeed, consider the case $n=1$, which corresponds to the case when $\mathcal{M}\cong \text{Hilb}_{f.d.}$, i.e. the module category arises from a fiber functor on $\mathcal{C}$ (and hence also on $\mathcal{D}$). In this case $H\in \text{Hilb}_{f.d.}$, and the target net is simply the standard quasi-local structure on $A=\otimes_{\mathbbm{Z}} M_{d}(\mathbbm{C})$, where $d=\text{dim}(H)$. In general $H\in \text{Mat}_{n}(\text{Hilb}_{f.d.})$ can be interpreted as a \textit{bigraded} Hilbert space, graded by the simple objects in $\mathcal{M}$ (see \cite{jones2025quantumcellularautomatacategorical}).

As described in Example \ref{ex:fusionspinsubalg}, these inclusions correspond to a Lagrangian algebra $A\in \mathcal{Z}(\mathcal{D})$. Thus 
$$A(\mathcal{C}, X)\subseteq A(\text{Mat}_{n}(\text{Hilb}_{f.d.},H))$$

\noindent is a physical boundary subalgebra.

If we start from the right $\mathcal{C}$-module category $\mathcal{M}$ and object $X\in \mathcal{C}^{*}_{\mathcal{M}}$, we can apply the usual framework to construct an MPO symmetry \cite{PRXQuantum.4.020357,MR3614057}, where the individual MPO operators indexed by simple objects in $\mathcal{C}$. If we consider the constrained subspace defined by the unit MPO, then the quasi-local algebra in this sector is given by the multi-fusion categorical spin chain $A(\text{Mat}_{n}(\text{Hilb}_{f.d.}), H)$ as described above. The quasi-local algebra of symmetric operators can be identified with the fusion spin chain $A(\mathcal{D},X)$ \cite{MR4109480}.

\end{ex}

We have the following useful corollary that allows us to construct many examples of physical boundary subalgebras.

\begin{remark}\label{rem:transportsymmetry}
Suppose $B\subseteq A$ is a physical boundary subalgebra, and $\alpha: A\rightarrow A^{\prime}$ is a bounded spread isomorphism. Then by Theorem \ref{thm:transportsubalg}, $B^{\prime}:=\alpha(B)\subseteq A^{\prime}$ is a physical boundary subalgebra of $A^{\prime}$, and $\alpha$ yields a bounded spread isomorphism between the inclusions $B\subseteq A$ and $B^{\prime}\subseteq A^{\prime}$. Furthermore, $\text{DHR}(\alpha): \text{DHR}(B)\cong \text{DHR}(B^{\prime})$ is a braided equivalence, and clearly $\text{DHR}(\alpha)(A)\cong A^{\prime}$ as algebra objects in $\text{DHR}(B^{\prime})$, in particular this gives us an identification between the symmetry categories in both cases.
\end{remark}

\begin{remark}{\textbf{Left and right movers?}} We actually have \textit{two} different copies of the symmetry category $\mathcal{C}$ (or rather $\mathcal{C}$ and $\mathcal{C}^{mp})$ that naturally arise in our setting. Indeed, we have

$$_{A} \text{DHR}(B)_{A}\cong \text{DHR}_{-}(A|B)\boxtimes \text{DHR}_{+}(A|B)\cong \mathcal{C}\boxtimes \mathcal{C}^{mp}\subseteq \text{Corr}(A).$$

\noindent The reason for this can be seen in our physical picture for this bimodules as in Section \ref{subsec:physpic}. In this case, we incorporate our topological boundary at the top of the diagram, and the objects correspond to topological defects localized on the topological boundary. To turn this into a bimodule over the algebra $A$, we have to `move the defect out of the way, as in the definition of the left action for a DHR bimodule. However, we have \textit{two inequivalent ways} of moving the defect out of the way, namely to the left and to the right. Thus, the single defect $X$ corresponds to two different bimodules over $A$, depending on our choice of moving the defect. See also \cite{MR1332979, MR1617550, MR1729094, MR1777347,
MR1785458, MR4642115}.

\end{remark}

\subsection{Symmetries as channels}\label{sebsec:SymasChan}

In the previous section, we showed how to reconstruct fusion category defect symmetry from a local subalgebra $B\subseteq A$. However, to make contact with the standard literature on categorical symmetries of spin chains, we need to find something closer to realizing symmetries as operators. We will realize this by implementing the fusion ring as \textit{quantum channels}, manifested in our infinite volume limit via unital completely positive (u.c.p) maps on the quasi-local algebra $A$.

Suppose we have a physical boundary subalgebra $B\subseteq A$. Consider the set 

$$\text{Ch}(A\ |\ B):=\{\Psi: A\rightarrow A\ :\ \Psi\ \text{is ucp and}\ \Psi(b_{1}ab_{2})=b_{1}\Psi(a)b_{2} \}. $$

\noindent $\text{Ch}(A\ |\ B)$ is a monoid under composition. It is also a convex space in the obvious way, and composition is ``bilinear" or ``affine" with respect to composition.

The idea is that this convex monoid captures the symmetries of the extension $B\subseteq A$ in a manner directly analogous to the Galois group of a field extension. This is important in the context of categorical symmetries, since unlike bimodules, quantum channels actually act on operators and states, allowing for a natural notion of symmetric states.

\begin{remark}If $\Psi\in \text{Ch}(A\ |\ B) $, then $\Psi$ has \textit{bounded spread}, i.e. $\Psi(A_{I})\subseteq A_{I^{+R}}$ for all intervals $I$. Indeed, since $B\subseteq A$ is local, $B^{\prime}_{I^{c}}\cap A\subseteq A_{I^{+R}}$. Thus if $a\in A_{I}$ and $b\in B_{I^{c}}$,

$$b\Psi(a)=\Psi(ba)=\Psi(ab)=\Psi(a)b,$$

\noindent hence $\Psi(a)\in A_{I^{+R}}$. If $B\subseteq A$ satisfies strong relative Haag duality so that $R=0$, we see that $\Psi(a)\in A_{I}$, in which case $\Psi$ has \textit{spread 0}.

\end{remark}

A class of examples of abstract convex monoids arise from the fusion rings of fusion categories. Let $\mathcal{C}$ be a fusion category, and consider the simplex $S(\mathcal{C})$ spanned by symbols $\{\lambda_{X}\}_{X\in \text{Irr}(\mathcal{C})}$. Define the multiplication on simplex vertices 

$$\lambda_{X}\lambda_{Y}:=\sum_{Z\in \text{Irr}(\mathcal{C})} \frac{d_{Z}}{d_{X}d_{Y}}N^{Z}_{XY}\lambda_{Z},$$

\noindent where $N^{Z}_{XY}$ are the fusion rules and $d_{X}$ denotes the quantum dimension of the object $X$. The affine extension to $S(\mathcal{C})$ is a convex monoid. We have the following theorem

\begin{thm}
Let $B\subseteq A$ be a physical boundary subalgebra with categorical symmetry $\mathcal{C}$. Then  $\text{Ch}(A\ |\ B)\cong S(\mathcal{C})$ as convex monoids. 
\end{thm}

\begin{proof}
Given a commutative, connected $Q$-system $A$ in a braided fusion category $\mathcal{B}$, there is a commutative C*-algebra $(\text{End}(A), \ast)$ \cite{MR4357481, schatz2024boundarysymmetries21dtopological, huang2024phasegroupcategorybimodule}. The elements are simply endomorphisms of $A$ in $\mathcal{B}$, and 

$$f*g:= m\circ (f\otimes g)\circ m^{\dagger}.$$

From a subfactor perspective, this is simply multiplication of the \textit{Fourier transform} of the 2-box space in the associated subfactor planar algebra \cite{MR4374438}. If $B\subseteq A$ is a physical boundary subalgebra, then viewing $A\in \text{DHR}(B)$ as a Q-system, the positive elements of $(\text{End}(A), \ast)$ are precisely the $B$-$B$ bimodular cp-multipliers on the C*-algebra $A$ (see \cite{huang2024phasegroupcategorybimodule,HP23,MR3687214}

Since the algebra is commutative, there exists a basis of minimal orthogonal projections, and the positive elements are precisely the convex span of these projections. Since $A$ is a Lagrangian algebra with $\beta_{-}(\text{DHR}(B)_{A})\cong \mathcal{C}$, then there is an identification of these minimal projections with the simple objects of $\mathcal{C}$, which we call $e_{x}$ (see \cite{MR4357481}). If we set $\lambda_{X}:=\frac{1}{d^{2}_{X}} e_{x}$, these will satisfy $\lambda_{X}\circ \lambda_{Y}=\sum_{Z} \frac{d_{Z}}{d_{X}d_{Y}} N^{Z}_{XY} \lambda_{Z}$. By construction, these are unital, hence $\lambda_{X}\in \text{Ch}(A\ |\ B)$ for all $X\in \text{Irr}(\mathcal{C})$. Since arbitrary cp-multipliers are positive, linear sums of the $\lambda_{X}$ and the $\lambda_{X}$ are unital, then the elements of $\text{Ch}(A\ |\ B)$ are precisely the \textit{convex sums} of the $\lambda_{X}$. Thus $\text{Ch}(A\ |\ B)\cong S(\mathcal{C})$ as convex monoids. 
\end{proof}

We note that in particular, the conditional expectation $E:A\rightarrow B\hookrightarrow A$ is a $B$-$B$ bimodular quantum channel, which can be written $\iota\circ \iota^{\dagger}$. Thus $E$ can be written as a convex combination of the $\lambda_{X}$, and solving for these coefficients in terms of the formulas of \cite{MR4357481},

$$E=\sum_{X\in \text{Irr}(\mathcal{C})} \frac{d^{2}_{X}}{\text{Dim}(\mathcal{C})} \lambda_{X}.$$

\noindent where $\text{Dim}(\mathcal{C})=\sum_{X\in \text{Irr}(\mathcal{C})} d^{2}_{X}$. Indeed, it is easy to see that the right-hand side is the unique convex combination of the $\lambda_{X}$ that leaves composition by any individual $\lambda_{X}$ invariant, which is clearly satisfied by $E$.

\begin{cor} Let $B\subseteq A$ be a physical boundary subalgebra with categorical symmetry $\mathcal{C}$. Let $$A^{\mathcal{C}}:=\{a\in A\ : \Phi(a)=a\ \text{for all}\ \Phi\in \text{Ch}(A\ |\ B)\}$$ be the set of fixed point operators. Then $B=A^{\mathcal{C}}$.
    
\end{cor}

\begin{proof}
If $a\in B$, then clearly $\lambda_{X}(a)=a$ for each $X\in \text{Irr}(\mathcal{C})$. Conversely, if $a\in A^{\mathcal{C}}$, $E(a)=a$, which implies $a\in B$.
\end{proof}

\begin{remark}{\textbf{Recovering MPOs from the symmetric subalgebra}.}
Given a physical boundary subalgebra $B\subseteq A$ with categorical symmetry $\mathcal{C}$, we obtain for each simple object $X\in \mathcal{C}$ a $B$-$B$ bimodular u.c.p. map $\lambda_{X}:A\rightarrow A$, satisfying the normalized fusion rules. In the case of MPO symmetries constructed from the setup of Example \ref{ex:MPO}, these quantum channels correspond to the channels built from MPOs as in \cite{jones2025quantumcellularautomatacategorical}. Thus we recover the MPO action (although at the level of channels in the thermodynamic limit rather than directly as operators on the periodic spin chain.)

\end{remark}

\begin{defn} Let $B\subseteq A$ be a physical boundary algebra with symmetry category $\mathcal{C}$. We say $\mathcal{C}$ is \textit{on-site} if, for each $X\in \text{Irr}(\mathcal{C})$, the symmetry channel $\lambda_{X}$ satisfies $\lambda_{X}(A_{I})\subseteq A_{I}$.
\end{defn}

\begin{thm}
If $B\subseteq A$ is a physical boundary subalgebra and $A$ satisfies strong Haag duality, then the symmetry is on-site if and only if $B\subseteq A$ satisfies strong relative Haag duality.
\end{thm}

\begin{proof}
Clearly if $B\subseteq A$ satisfies strong relative Haag duality, then since each $\lambda_{X}$ is $B$-$B$ bimodular,  for any $b\in B_{I^{c}}$ and $a\in A_{I}$, $b\lambda_{X}(a)=\lambda_{X}(ba)=\lambda_{X}(ab)=\lambda_{X}(a)b$ so $\lambda_{X}(a)\in A_{I}$.

Conversely, suppose $A$ has strong Haag duality and the symmetry channels are on-site. This implies the unique local conditional expectation $E=\sum_{X\in \text{Irr}(\mathcal{C})} \frac{d^{2}_{X}}{\text{Dim}(\mathcal{C})} \lambda_{X}$ is on-site (i.e. $E(A_{I})=B_{I}$). Now, let $I$ be any interval and $J\subseteq I^{c}$. We claim that if $a\in Z_{A}(B_{I^{c}})$, then $a\in Z_{A}(A_{J})$. Since this is true for any interval, this will imply $a\in Z_{A}(A_{I^{c}})$, and by strong Haag duality, we can conclude $a\in A_{I}$, thus obtaining strong relative Haag duality.

To see our claim, note that by weak relative Haag duality, $a\in A_{I^{+R}}$ for some $R$. Choose a projective basis $\{e_{i}\}\subseteq A_{J^{\prime}}$ for some interval $J^{\prime}\subseteq I^{c}$ which is on the same side of $I$ as $J$, but whose distance from $I$ is greater than $R$. Then for any $a^{\prime}\in A_{J}$, we have $a^{\prime}=\sum_{i} e_{i}E(e^{*}_{i}a^{\prime})$. By the on-site condition, we see that $E(e^{*}_{i}a)\in B_{J\vee J^{\prime}}\subseteq B_{I^{c}}$, and $e_{i}$ commutes with $a\in A_{I^{+R}}$. 
Thus $[a,a^{\prime}]=0$, so by strong Haag duality, $a\in A_{I}$, hence $Z_{A}(B_{I^{c}})\subseteq A_{I}$, and the other inclusion is trivial.
\end{proof}

This result implies that for categorical symmetries of tensor product spin chains, on-site is equivalent to strong relative Haag duality of the symmetric operators.

\section{Categorical symmetries of tensor product quasi-local algebras}\label{sec:tensorprod}

Fusion categories will generically admit actions on \textit{some} quasi-local algebra, e.g. anyon chains, which are captured in the typical MPO picture. In this section, we investigate a question that has naturally arisen from the MPO perspective: which fusion categories can act by symmetries on tensor product quasi-local algebras, i.e. on UHF algebras $A:=\otimes_{\mathbbm{Z}} M_{d}(\mathbbm{C})$ with the standard quasi-local structure? 

The first observation is that there is a $K$-theoretic obstruction which drastically restricts which fusion categories can act on concrete spin chains in principle.

\begin{prop}\label{prop:obstruction}
If $\mathcal{C}$ is realized by symmetries of a tensor product spin chain, then the dimension of every object of $\mathcal{C}$ is an integer.
\end{prop}

\begin{proof}
Since we have a tensor functor $\mathcal{C}\rightarrow \text{Bim}(A)$, the operator $K$-group $K_{0}(A)$ becomes an ordered module for the fusion ring $\mathcal{C}$. But $K_{0}(A)$, as an ordered abelian group, is canonically isomorphic to  the d-adic rationals \cite{MR1783408}. In fact, there exists a unique state $\phi$ on $K_{0}(A)$ which is precisely this identification. By \cite[Proposition 5.3]{MR4419534} for any object $X\in \mathcal{C}$ its action on $K_{0}$ must be multiplication by some positive real number yielding a normalized, positive dimension function. But the only algebraic integers that leave all d-adic rational invariant under multiplication are integers.
    \end{proof}

There are many integral fusion categories. The most basic (and ubiquitous) examples are the \textit{group theoretical} fusion categories, which defined to be unitary fusion categories Morita equivalent to $\text{Vec}(G,\omega)$ for some finite group $G$ and $\omega\in Z^{3}(G,\text{U}(1))$. Integral fusion can also arise as representation categories of finite dimensional C*-Hopf algebras, and this is precisely the class that admit fiber functors \cite{MR3204665}, which we discuss below. A natural question raised by the above result is, which integral fusion categories admit actions on tensor product spin chains? The next result answers this question.

\begin{thm}\label{thm:intrealiz}
    Let $\mathcal{C}$ be an integral fusion category, and let $d=\sum_{X\in \text{Irr}(\mathcal{C})} d^{2}_{X}$. Then $\mathcal{C}$ is realized as symmetries on the tensor product quasi-local algebra with local Hilbert space dimension $d$.
\end{thm}

\begin{proof}
Our construction first involves giving a standard MPO-type construction with a spatial realization, and then passing this through a bounded spread isomorphism to the standard spin system algebra $A$.

We will use the notation and definitions from Example \ref{ex:MPO}. Let $n=\text{rank}(\mathcal{C})$. First, consider $\mathcal{C}$ as a right $\mathcal{C}$-module category, so that $\mathcal{D}=\mathcal{C}^{*}_{\mathcal{C}}=\mathcal{C}$ (we will call it $\mathcal{D}$ to emphasize that it is playing the role of the ``charge category"). Then choose the regular object $R=\bigoplus_{X\in \text{Irr}(\mathcal{C})} X^{\oplus d_{X}}$. $R$ is clearly a strong $\otimes$-generator for $\mathcal{D}$. If we set $H:=\{_{X} H_{Y}\}_{X,Y\in \text{Irr}(\mathcal{C})}$ to be the bigraded Hilbert space in $\text{Mat}_{n}(\text{Hilb}_{f.d.})$ corresponding to $R$ under the functor $\mathcal{D}\rightarrow \text{Mat}_{n}(\text{Hilb}_{f.d.})$, then we have an ``anyon chain" quasi-local algebra $A:=A(\text{Mat}_{n}(\text{Hilb}_{f.d.}),H)$ and a physical boundary subalgebra $B:=A(\mathcal{D}, R)\subseteq A$ realizing the categorical symmetry $\mathcal{C}$.

Now our goal is to construct a bounded spread isomorphism $\alpha:A\rightarrow A^{\prime}:=\otimes_{\mathbbm{Z}} M_{d}(\mathbbm{C})$. By Remark \ref{rem:transportsymmetry}, this gives a physical boundary subalgebra $B^{\prime}\subseteq A^{\prime}$ bounded spread isomorphic to the local inclusion $B\subseteq A$. In particular, $\mathcal{C}$ is realized by an action on $A^{\prime}$.

To do this, consider $X,Y\in \text{Irr}(\mathcal{C})$, and consider the component $_{X}H_{Y}$ of the bi-graded Hilbert space $H$. Then we claim $\text{dim}( _{X}H_{Y})=d_{X}d_{Y}$. To see this, first note that $_{X}H_{Y}\cong \mathcal{C}(X, R\otimes Y)$. But $[R]$ is the regular element in the fusion ring \cite[Definition 3.3.8]{MR3242743} so $[R]\cdot [Y]=d_{Y} [R]$. Thus 

$$R\otimes Y\cong \bigoplus_{X\in \text{Irr}(\mathcal{C})} X^{\oplus d_{Y}d_{X}},$$

\noindent hence $\text{dim}(_{X} H_{Y})=\mathcal{C}(X, R\otimes Y)=d_{X}d_{Y}$ as desired. For each $X\in \text{Irr}(\mathcal{C})$, set $V_{X}:=\mathbbm{C}^{d_{X}}$. Then for each pair $X,Y\in \text{Irr}(\mathcal{C})$, we can arbitrarily choose a decomposition of Hilbert spaces $_{X} H_{Y}\cong V_{X}\otimes V_{Y}$.

If $I$ is an interval with length $k$, then

$$A_{I}:=\text{End}_{\text{Mat}_{n}}(H^{\otimes k}).$$

But

$$_{X}(H^{\otimes k})_{Y}\cong \bigoplus_{Z_{1},\dots Z_{k-1}\in \text{Irr}(\mathcal{C})}\ (_{X} H_{Z_{1}})\ \otimes\ (_{Z_{1}} H_{Z_{2}})\ \otimes\ \dots\ \otimes (_{Z_{k-1}} H_{Y}).$$
$$\cong \bigoplus_{Z_{1}, \dots, Z_{k-1}} V_{X}\otimes V_{Z_{1}}\otimes V_{Z_{1}}\otimes V_{Z_{2}}\otimes V_{Z_{2}}\otimes \dots \otimes V_{Z_{k-1}}\otimes V_{Z_{k-1}}\otimes V_{Y}$$

If we set $U:=\bigoplus_{X\in \text{Irr}(\mathcal{C})} V_{X}\otimes V_{X}$, we have

$$_{X}(H^{\otimes k})_{Y}\cong V_{X}\otimes U^{\otimes k-1}\otimes V_{Y}$$

Define the tensor product quasi-local algebra $A^{\prime}:=A(\text{Hilb}_{f.d.}, U)$. We will define a natural map $\alpha: A_{[a,b]}\rightarrow A^{\prime}_{[a,b+1]}$. For any $f\in A_{[a,b]}$ with $|b-a|+1=k$, denote the components $f:=\{_{X} f_{Y}\ :\ _{X} f_{Y}: _{X}\left(H^{\otimes k}\right)_{Y}\rightarrow _{X}\left(H^{\otimes k}\right)_{Y}\}$. Then define

$$\alpha(f):=\bigoplus _{X,Y\in \text{Irr}(\mathcal{C})} 1_{V_{X}}\otimes (_{X}f_{Y})\otimes 1_{V_{Y}}\in \text{End}_{\text{Hilb}_{f.d.}}(U^{\otimes k+1})=A^{\prime}_{[a,b+1]}.$$

This is an injective $*$-homomorphism, that is compatible with inclusions, and thus extends to a unital, bounded spread $*$-homomorphim $\alpha: A\rightarrow A^{\prime}$. To see that it is an isomorphism, we will construct its inverse.

For $g\in A^{\prime}_{[a,b]}=\text{End}_{\text{Hilb}_{f.d.}}(U^{k})$, we define components

$$_{Z}(\alpha^{-1}(g))_{W}:=1_{Z}\otimes g\otimes 1_{W}: _{Z} H^{\otimes k+1}_{W}\rightarrow _{Z} H^{\otimes k+1}_{W},$$ 

\noindent which assemble into a morphism

$$\alpha^{-1}(g)\in \text{End}_{\text{Mat}_{n}(\text{Hilb}_{f.d.})}(H^{\otimes k+1})=A_{[a-1,b]}.$$

Again, this extends to a $*$-homomorphism with spread $1$ from $A^{\prime}$ to $A$. We check

\begin{align*}
\alpha^{-1}\circ \alpha(f)&= \bigoplus_{X,Y,Z,W} 1_{Z}\otimes 1_{X}\otimes ( _{X}f_{Y}) \otimes 1_{Y}\otimes 1_{W}\\
&=1_{H}\otimes f\otimes 1_{H}\\
&=f\in A
\end{align*}

and 

\begin{align*}
\alpha\circ \alpha^{-1}(g)&=\bigoplus_{Z,W} 1_{Z}\otimes 1_{Z}\otimes g\otimes 1_{W} \otimes 1_{W}\\
&=1_{U}\otimes g\otimes 1_{U}\\
&=g\in A^{\prime}.
\end{align*} \end{proof}

We will now also address the relationship between tensor product realization and \textit{anomalies} of fusion categories \cite{TW24,PhysRevB.110.035155}. A unitary fusion category $\mathcal{C}$ is said to be \textit{anomaly free} if there exists a unitary fiber functor $\mathcal{C}\rightarrow \text{Hilb}_{f.d}$. An anomaly is commonly viewed an obstruction to the existence of a symmetric gapped ground state. Every anomaly free fusion category has integral dimensions, but the converse is not true. $\text{Hilb}_{f.d.}(G,\omega)$ are the standard counter-examples when $\omega$ is non-trivial. While we have just shown any integral fusion category admits an action on some concrete spin chain, we will now show that the anomaly is an obstruction to the existence of an on-site action. In other words, when we have an on-site action, we will show we can build a canonical tensor functor from $\mathcal{C}$ to a half-interval DHR category, which for tensor product spin chains will be just $\text{Hilb}_{f.d.}$, yielding a fiber functor.

Let $A$ be an abstract quasi-local algebra over $\mathbbm{Z}$. For convenience, we assume our theory is translation covariant, i.e. there is an automorphism $\alpha:A\rightarrow A$ such that $\alpha(A_{I})=A_{I+1}$. Then pick any point (for example, $0$). Then set $A_{-}$ to be the C*-algebra generated by $A_{I}$ with $I\le 0$. 

Now define the category $\text{DHR}_{L}(A_{-})$ to be the full C*-tensor category of correspondence such that for sufficiently large intervals $I$ of the form $[a,0]$, there exists a projective basis localized in $I$. Unlike the usual DHR category of full quasi-local algebras, this one is not braided, though it is monoidal.

\begin{thm}
   Suppose $A$ satisfies strong Haag duality, and $B\subseteq A$ is a physical boundary subalgebra with on-site categorical symmetry $\mathcal{C}$. Then there exists a unitary tensor functor $\mathcal{C}\rightarrow \text{DHR}_{L}(A_{-})$.
\end{thm}

\begin{proof}
Choose some interval $[a,0]$ large enough such that for all $X\in \mathcal{C}$, we can pick a projective basis $\{b_{i}\}$ for $X$ that has $ab_{i}=b_{i}a$ for all $a\in A_{(\infty,a)}$ and $ab_{i}=b_{i}a$ for all $a\in B_{(0,\infty)}$. By strong algebraic Haag duality, $\langle b_{i}\ |\ b_{j}\rangle\in A_{[a,0]}\subseteq A_{-}$. 

Thus the space 
$$\widetilde{X}:=\text{span}\{b_{i}A_{-}\}$$

\noindent is an $A_{-}$ bimodule with $A_{-}$-valued inner product. This bimodule does not depend on choice of interval $I=[a,0]$ or projective basis localized in $I$ (again by strong Haag duality).

If $f:X\rightarrow Y$ is an intertwiner of correspondences, then if $\{c_{j}\}$ denotes the choice of localized basis for $Y$, $\langle c_{j}\ |\ f(b_{i})\rangle$ is also localized in $A_{[a,0]}$ hence for $\sum_{i} b_{i}a_{i}\in \widetilde{X}$

$$f(\sum_{i}b_{i}a_{i})=\sum_{i}f(b_{i})a_{i}=\sum_{j,i}c_{j}\langle c_{j}\ |\ f(b_i)\rangle a_{i}\in \widetilde{Y}.$$

Thus $\widetilde{f}:\widetilde{X}\rightarrow \widetilde{Y}$ given by restriction is well-defined, and thus the assignment $X\mapsto \widetilde{X}$ extends to a $\dagger$ functor $L:\text{DHR}(A)\rightarrow \text{DHR}_{L}(A_{-})$. We will now check that this is monoidal.

If $\{b_{i}\}\subseteq X$ and $\{c_{j}\}\subseteq Y$ are $[a,0]$-localized bases, then $\{b_{i}\boxtimes c_{j}\}$ is an $[a,0]$ localized basis for $X\boxtimes_{A} Y$. Thus $$\widetilde{X\boxtimes_{A}Y}=\text{span}\{b_{i}\boxtimes c_{j} A_{-}\}$$

On the other hand, $\{b_{i}\}$ and $\{c_{j}\}$ are projective bases for $\widetilde{X}$ and $\widetilde{Y}$ respectively, hence $\{b_{i}\boxtimes c_{j}\}$ is a projective basis for $\widetilde{X}\boxtimes_{A_{-}} \widetilde{Y}$. Thus

$$\widetilde{X}\boxtimes_{A_{-}}\widetilde{Y}=\text{span}\{b_{i}\boxtimes c_{j} A_{-}\}.$$

By the definition of the inner product on both sides, we see that the natural identification $\widetilde{X}\boxtimes_{A_{-}}{Y}\rightarrow \widetilde{X\boxtimes_{A} Y}$ is a unitary. It is straightforward to see this is natural and satisfies the usual coherence.
\end{proof}

\begin{thm}
    Let $A=\otimes_{\mathbbm{Z}} M_{d}(\mathbbm{C})$ be a tensor product quasi-local algebra. Then $\text{DHR}_{L}(A_{-})=\text{Hilb}_{f.d.}$.
\end{thm}

\begin{proof}
Let $X\in \text{DHR}_{L}(A_{-})$, and let $\{b_{i}\}$ be a basis localized in some interval $I=[a,0]$. If we set $$X_{I}:=\text{span} \{b_{i}A_{I}\},$$
    
\noindent then by strong Haag duality for $A_{-}$ in the tensor product spin system, $X_{I}$ with its $A_{-}$-valued inner product is in fact an $A_{I}$ correspondence. Now, consider $A_{<I}$ as the trivial $A_{<I}$ correspondence. Then
$$A_{<I}\otimes X_{I}$$

\noindent has the natural structure of an $A_{<I}\otimes A_{I}= A_{-}$ correspondence in the obvious way. We claim that in fact $A_{<I}\otimes X_{I}\cong X$ as $A_{-}$-$A_{-}$ correspondences. Define $f:A_{<I}\otimes X_{I}\rightarrow X$ by
$$f(a\otimes x):=ax=xa\in X$$

This is clearly $A_{-}$ bimodular. We now verify that it is an isometry with respect to the $A_{-}$ valued inner product by computing

\begin{align*}
\sum_{i,j}\langle f(a_{i}\otimes x_{i})\ |\ f(b_{j}\otimes y_{j})\rangle_{A_{-}}&=\sum_{i,j}\langle x_{i}a_{i}\ |\ b_{j}y_{j}\rangle_{A_{-}}\\
&=\sum_{i,j} a^{*}_{i}b_{j}\langle x_{i}\ |\ x_{j}\rangle_{A_{-}}\\
&=\sum_{i,j}\langle a_{i}\otimes x_{i}\ |\ b_{j}\otimes y_{j}\rangle_{A_{-}}
\end{align*}

Furthermore, we have $X=X_{I}A_{-}=X_{I}A_{I}A_{<I}=X_{I}A$, and since $X_{I}$ contains a projective basis, we see that $f$ is surjective. In particular, $f$ is a unitary equivalence of correspondences. 

However, $A_{I}$ is just a finite-dimensional matrix algebra, hence $X_{I}\cong A^{\oplus n}_{I}$ as $A_{I}$ correspondences (note that $n$ is finite since $X_{I}$ is finite dimensional by the existence of the $I$-localized finite projective basis). In particular, as $A_{-}$ correspondences, we see
$$X\cong A_{<I}\otimes A^{\oplus n}_{I}\cong (A_{<I}\otimes A_{I})^{\oplus n}\cong A^{\oplus n}_{-}.$$

\noindent Thus $X$ is a direct sum of trivial $A_{-}$ bimodules. Hence $\text{DHR}_{L}(A_{-})\cong \text{Hilb}_{f.d}$ as unitary tensor categories.
\end{proof}

\begin{cor}\label{cor:on-siteanomaly}
    A fusion category is realized as an on-site symmetry on a tensor product quasi-local algebra $A=\otimes_{\mathbbm{Z}} M_{d}(\mathbbm{C})$ if and only if $\mathcal{C}$ admits a fiber functor.
\end{cor}

\begin{proof}
If $\mathcal{C}$ admits a fiber functor, then there is a natural right $\mathcal{C}$-module category structure on $\text{Hilb}_{f.d.}$. Picking a generating object $H\in \mathcal{C}^{*}_{\text{Hilb}_{f.d.}}$, by \ref{ex:MPO} we have a realization of $\mathcal{C}$ on $A(\text{Hilb}_{f.d.}, H)\cong \otimes_{\mathbbm{Z}} M_{d}(\mathbbm{C})$, where $d=\text{dim}(H)$. 

Conversely, if $\mathcal{C}$ is realized on $A=\otimes_{\mathbbm{Z}} M_{d}(\mathbbm{C})$, then $\mathcal{C}$ admits a fiber functor by the previous theorem.
\end{proof}

\section{Symmetric states}\label{sec:symstates}

In this section, we perform an analysis of $\mathcal{C}$ symmetric states from our operator algebraic perspective. We propose a natural definition for a symmetric topological state\footnote{we use the term topological in place of the more commonly used ``gapped'', since we do not need to reference a Hamiltonian, and topological is usually what is meant in the physics literature}, and show how these correspond to module categories over $\mathcal{C}$. The rank of the module category measures the degree of symmetry breaking, and hence our analysis gives a formal proof that if there is a pure, symmetric topological state in the infinite volume limit, then there exists a fiber functor on $\mathcal{C}$. Thus the existence of an anomaly forces symmetry breaking or gaplessness, which is a categorical variant of the Lieb-Schultz-Mattis-type theorem, originally (to our knowledge) stated in \cite{TW24}.

To be more precise, suppose $B\subseteq A$ is a physical boundary subalgebra with symmetry category $\mathcal{C}$. By a state, we mean in the usual sense for C*-algebras, i.e. a positive linear functional $\phi:A\rightarrow \mathbbm{C}$ with $\phi(1)=1$.

\begin{defn} A state $\phi$ on $A$ is called symmetric if $\phi\circ \lambda_{X}=\phi$ for all $X\in \text{Irr}(\mathcal{C})$.
\end{defn}

\begin{remark}
Note that $\phi$ is symmetric if and only if $\phi\circ E=\phi$. This implies $\phi$ is totally determined by its values on the subalgebra $B$, and there is a bijection between symmetric states on $A$ and states on $B$. We also note that symmetric states are never far away: given an arbitrary state $\phi$ on $A$, $\phi\circ E$ is symmetric (note that this is simply an averaging procedure over the quantum channel symmetries).
\end{remark}

Even though there is a bijection between states on $B$ and symmetric states on $A$, a pure state on $B$ might lift to a mixed state on $A$. If $\phi$ is symmetric, $L^{2}(A,\phi)=A\boxtimes_{B} L^{2}(B,\phi)$.

States (and more generally, Hilbert space representations) of $B$ can be thought of boundary theories for the Symmetry TFT $\mathcal{T}$, and thus we can think of induction (sending a Hilbert space representation $H$ of $B$ to $A\boxtimes_{B}H$) as gluing a boundary condition of $\mathcal{T}$ to the whole SymTFT decomposition. This leads us to formally define a symmetric Hilbert representation of $A$ as an induction of a Hilbert space representation of $B$.

$$\begin{tikzpicture}
\draw (0,1.5)--(2,1.5);
\draw [dotted,red] (0,.7)--(2,0.7);
\draw (0,0.2)--(2,0.2);
\draw node at (1,1.8) {$\mathcal{B}_{\text{top}}$};
\draw [red] node at (-0.4,0.6) {$\mathcal{B}_{\text{phys}}$};
\draw node at (1,-0.3) {$H$};
\end{tikzpicture}$$

The basic idea is that symmetric states/representations can be described \textit{relative} to $\mathcal{T}$. In particular, any Hilbert space representation of B (or a state) gives a \textit{W*-algebra object} internal to $\text{DHR}(B)$, and we can use this information to analyze the induction. This approach was already proposed in \cite[Section 6.3]{jones2025localtopologicalorderboundary} for boundary states of Levin-Wen type models. We refer the reader to Appendix \ref{app:W*} for a detailed discussion of W*-algebra objects and their relations. 

Suppose we are given a boundary condition for $\mathcal{T}$ embodied by a state $\phi$ of $B$. Then we define an associated \textit{module category} for $\text{DHR}(B)$, which we call $\mathcal{N}_{\phi}\subseteq \text{Rep}(B)$, defined as summands of $X\boxtimes_{B} L^{2}(B,\phi)$, where $X\in \text{DHR}(B)$. $\mathcal{N}_{\phi}$ together with the choice of generator $L^{2}(B,\phi)\in \mathcal{N}_{\phi}$ defines a \textit{cyclic module category} for $\text{DHR}(B)$, hence a W*-algebra object $\mathcal{H}_{\phi}$ in $\text{DHR}(B)$.

If $\text{End}_{B}(L^{2}(B,\phi))$ is finite dimensional), then $\mathcal{N}_{\phi}$ is a finitely semisimple module category. In particular, $\mathcal{H}_{\phi}$ is \textit{compact}, hence can be represented by a Q-system $H_{\phi}$. This is unique in the case $\phi$ is pure \cite{MR4079745}, but in general the choices of lift are parameterized by module traces on $\mathcal{N}_{\phi}$ \cite{chen2024manifestlyunitaryhigherhilbert}. The different choices are irrelevant since we do not care about module traces here, but having an actual Q-system (which lives in $\text{DHR}(B)$ rather than $\text{Vec}(\text{DHR}(B))$ will be convenient.

\begin{defn}
    
Let $B\subseteq A$ be a physical boundary subalgebra and $\phi$ a symmetric state on $A$. Then $\phi$ is

\begin{enumerate}
\item 
\textit{compact} if $\text{End}_{B}(L^{2}(B,\phi))$ is finite dimensional.
\item 
\textit{connected} if $\phi|_{B}$ is pure (i.e. $\text{dim}(\text{End}_{B}(L^{2}(B,\phi)))=1$)

\end{enumerate}

\noindent In case $\phi$ is compact, we denote a choice of Q-system representing $\mathcal{H}_{\phi}$ by $H_{\phi}$.
\end{defn}

Note that every compact state is a finite direct sum of connected ones by definition. In this section, we will only be concerned with compact states.

An associated individual boundary condition corresponds to a state on $B$, but such a boundary condition can have topological defects. It is thus natural to single out a class of states which can be interpreted as vacuum states for their topological defect theories.

\begin{defn}
Let $\phi$ be a symmetric compact state on $A$. Then $\phi$ is \textit{local} if $H_{\phi}$ is commutative.
\end{defn}

As mentioned above, we claim if a state $\phi$ is local, it is reasonable to interpret $\phi$ as a ``vacuum" state for the whole boundary theory we are gluing on to $\mathcal{T}$. If $\phi$ is local, then we can equip $\mathcal{N}_{\phi}$ canonically with the structure of a unitary fusion category, whose objects can be thought of as point defects of the boundary theory defined by $\phi$ (more precisely, they are a version of superselection sectors of the state $\phi$ in the sense of \cite[Definition 6.8]{jones2025localtopologicalorderboundary}). 

More generally if $H_{\phi}$ is Morita equivalent to a commutative algebra, then it is reasonable to interpret $\phi$ as a state with a topological defect over the vacuum corresponding to the commutative algebra. More precisely, our state can be created by applying a DHR bimodule to a local state, and decomposing. To this end, we show that the locality condition follows from a type of Haag duality in the GNS representation of $B$, rephrasing the locality condition in the language of AQFT.

\begin{defn}
A pure state $\phi$ on $B$ satisfies \textit{bounded spread cone Haag duality} if there exists an $R\ge 0$ such that for any ray $L$, $B^{\prime}_{L^{c}}\cap \mathcal{B}(L^{2}(B,\phi))\subseteq B^{\prime \prime}_{L^{+R}}$.
\end{defn}

\begin{thm} Suppose $B$ is simple. If $\phi$ is a pure state on $B$ satisfying bounded spread cone Haag duality, then $\mathcal{H}_{\phi}$ is local.
\end{thm}

\begin{proof}
 Since $\mathcal{H}_{\phi}$ is a connected C*-algebra object, we can equip it with the structure of a (dual) Q-system $H_{\phi}$, and apply the realization to obtain an irreducible, finite index extension $D:=|H_{\phi}|\supseteq B$. Thus $D$ is simple \cite[Theorem 3.3]{MR1900138}. There is a natural $*$-homomorphism $\pi_{\phi}: D\rightarrow \mathcal{B}(L^{2}(B,\phi)$ which we define using the decomposition

 $$D:=\bigoplus_{X\in \text{Irr}(\text{DHR}(B))} \mathcal{H}_{\phi}(X)\otimes X=\bigoplus_{X\in \text{Irr}(\text{DHR}(B))} \text{Hom}_{B} \left(X\boxtimes_{B} L^{2}(B,\phi), L^{2}(B,\phi)\right)\otimes X.$$

 \noindent Then for $f\otimes \xi\in \text{Hom}_{B} \left(X\boxtimes_{B} L^{2}(B,\phi), L^{2}(B,\phi)\right)\otimes X $, set
 
 $$\pi_{\phi}(f\otimes \xi)(\eta):=f(\xi\boxtimes_{B} \eta).$$

\noindent It is easy to check that $\pi_{\phi}(f\otimes \xi)$ is bounded, and that the linear extension to $D$ defines a $*$-homomorphism from $D$ to $\mathcal{L}(L^{2}(B,\phi))$. By simplicity, it is injective. Now, for any sufficiently large intervals $I$, we can choose a basis for $H_{\phi}$ localized in $I$ of the form $e^{X}_{ij}=f^{X}_i\otimes b^{X}_{j}$, where for each $X\in \text{Irr}(\text{DHR}(B))$, $f^{X}_i$ runs through an orthonormal basis of the space $\text{Hom}_{B} (X\boxtimes_{B} L^{2}(B,\phi), L^{2}(B,\phi))$ (and inner product determined by choice of module trace that makes $H_{\phi}$ into a dual Q-system), and $b^{X}_{j}$ runs through a projective basis for $X$ localized in $I$.

Now, we claim that $\pi_{\phi}(e_{X,i})$ commutes with $B_{I^{c}}$. Indeed, if $b\in B_{I^{c}}$, we have 

\begin{align*}
\pi_{\phi}(e^{X}_{ij})(b(\eta))&=f^{X}_{j}(b^{X}_{i}\boxtimes_{B}(b\eta))\\
&=f^{X}_{j}(b^{X}_{i}b\boxtimes_{B}\eta)\\
&=f^{X}_{j}(bb^{X}_{i}\boxtimes_{B}\eta)\\
&=bf^{X}_{j}(b^{X}_{i}\boxtimes_{B}\eta)\\
&=b\pi_{\phi}(e^{X}_{ij})(\eta)\\
\end{align*}
\noindent In particular, $\pi_{\phi}(e^{X}_{ij})$ commutes with $B_{\mathbbm{Z}_{>I}}$, and so $\pi_{\phi}(e_{X,i})\in B^{\prime \prime}_{\mathbbm{Z}^{+R}_{<I}}$.

Similarly, if we apply the same construction and analysis to an interval $J>>I$, we can find a projective basis $f^{Y}_{rs}$ for $H_{\phi}$ localized in $J$ such that $\pi_{\phi}(f^{Y}_{rs})$ commutes with $B_{J^{c}}$, hence lies in $B^{\prime \prime}_{\mathbbm{Z}^{+R}_{>J}}$. But since $[B_{\mathbbm{Z}^{+R}_{>J}}, B_{\mathbbm{Z}^{+R}_{<I}}]=0$, $[B^{\prime \prime}_{\mathbbm{Z}^{+R}_{>J}}, B^{\prime \prime}_{\mathbbm{Z}^{+R}_{<I}}]=0$, hence

\begin{align*}
\pi_{\phi}(m\circ (e^{X}_{ij}\boxtimes_{B} f^{Y}_{rs}))&=\pi_{\phi}(e^{X}_{ij}f^{Y}_{rs})\\
&=\pi_{\phi}(e^{X}_{ij})\pi_{\phi}(f^{Y}_{rs})\\
&=\pi_{\phi}(f^{Y}_{rs})\pi_{\phi}(e^{X}_{ij})\\
&=\pi_{\phi}(m\circ \sigma_{H_{\phi},H_{\phi}}(e^{X}_{ij}\boxtimes_{B} f^{Y}_{rs})\\
\end{align*}
\noindent Since $\pi_\psi$ is injective, we obtain $m\circ (e^{X}_{ij}\boxtimes_{B} f^{Y}_{rs})=m\circ \sigma_{H_{\phi},H_{\phi}}(e^{X}_{ij}\boxtimes_{B} f^{Y}_{rs})$. However, $\{e^{X}_{ij}\boxtimes_{B} f^{Y}_{rs}\}$ is a projective basis for $H_{\phi}\boxtimes H_{\phi}$, and any Hilbert module map is determined by its action on a projective basis. Thus

$$m\circ \sigma_{H_{\phi},H_{\phi}}=m,$$

\noindent hence $H_{\phi}$ is commutative and $\phi$ is local.

\end{proof}

\begin{remark}
Another way to see the locality of states satisfying bounded spread cone duality is to use the work of \cite{bhardwaj2024superselectionsectorsposetsvon}. There, the authors show that in $\mathbbm{Z}^{2}$, satisfying bounded spread cone duality yield a braided tensor category of cone-localized superselection sectors, building on work of Naaijkens and Ogata \cite{MR2804555,MR4362722}. The same argument applied to $\mathbbm{Z}$ still works, but yields  only a tensor category of superselection sectors, without a braiding. One can show (similar in spirit to \cite{MR4790511,MR4650344}) that the free module functor from $\text{DHR}(B)$ to the category of cone-localized superselection sectors of $\psi$ is central and thus the internal hom object $\mathcal{H}_{\phi}$ is commutative. 
\end{remark}

We now introduce the concept of topological and gapless boundary states. In the physics literature, the term ``gapped" is frequently used interchangeably with ``topological", especially in the context of categorical symmetries. However, since there are no Hamiltonians in sight, we think it is better to use the term topological. Nevertheless, there does not seem to be analogous terminology for ``gapless" that is commonly used, so we will adopt this terminology here.

\begin{defn} Let $B\subseteq A$ be a physical boundary subalgebra and $\phi$ a symmetric, connected state on $A$. Then

\begin{enumerate}
\item 
\textit{topological} (or \textit{gapped}) if $H_{\phi}$ is Lagrangian.
\item 
\textit{gapless} if $H_{\phi}$ is connected but not Morita equivalent to a Lagrangian algebra object
\end{enumerate}
\end{defn}

The reason we require a gapless state to be "not Morita equivalent" to a Lagrangian algebra (rather than "not isomorphic") is that if $H_{\phi}$ is Morita equivalent to a Lagrangian algebra, as discussed above it can be interpreted as a topological defect of the ``vacuum" topological state defined by the Lagrangian algebra. More precisely, our state can be created from a topological state by applying a DHR bimodule (and decomposing).

\begin{remark}\label{rem:topmodcat} In addition to the $\text{DHR}(B)$ module category $\mathcal{N}_{\phi}$, since $H_{\phi}$ is Lagrangian, associated to a topological state $\phi$ is a $\mathcal{C}$-module category $\mathcal{M}_{\phi}$ \cite{MR3406516}.
\end{remark}

\begin{ex}{\textbf{Building topological states}}\label{ex:buildingtopstate}. Let $\mathcal{C}$ be a unitary fusion category, and $\mathcal{M}$ an indecomposable (right) $\mathcal{C}$-module category, with associated Lagrangian algebra $L\in \mathcal{Z}(\mathcal{C})$. In this example, we will construct a physical boundary system $B\subseteq A$ with $\mathcal{C}$-symmetry and a topological $\mathcal{C}$-symmetric state $\phi$ such that $\mathcal{M}_{\phi}=\mathcal{M}$.

To do this, choose $X\in \mathcal{C}^{*}_{\mathcal{M}}$ with $X$ a strong tensor generator. Let $v:\mathbbm{1}\rightarrow X^{\otimes n}$ be a choice of isometry for some value $n$ sufficiently large so that all simple objects are contained in $X^{\otimes n}$. For simplicity, we can coarse grain and assume $n=1$. Let $B$ be the quasi-local algebra corresponding to the fusion spin chain $\mathcal{A}(\mathcal{C}^{*}_{\mathcal{M}},X)$. Then there is a Lagrangian algebra $A\in \text{DHR}(B)\cong \mathcal{Z}(\mathcal{C}^{*}_{\mathcal{M}})$ with $\text{DHR}(B)_{A}\cong \mathcal{C}$, and this yields the physical boundary system $B\subseteq A$.

Define $\phi$ on $B_{I}:=\text{End}_{\mathcal{C}^{*}_{\mathcal{M}}}(X^{\otimes I})$ by 

$$\phi(f)1_{\mathbbm{1}}:=(v^{\dagger}\otimes \dots \otimes v^{\dagger})\circ f\circ (v\otimes \dots v).$$

This extends to a uniquely defined pure state $\phi$ on $B$. We use the unique conditional expectation $E:A\rightarrow B$ to define the $\mathcal{C}$-symmetric state $\phi(a)=\phi(E(a))$ for all $a\in A$. By exactly the same proof as \cite[Theorem 6.14]{jones2025localtopologicalorderboundary}, this state has $\mathcal{H}_{\phi}\cong L$.

The construction above started with a symmetry category $\mathcal{C}$ and a choice of Lagrangian algebra $L$, and \textit{then constructed} a physical boundary subalgebra $B\subseteq A$ realizing $L$ as a topological state, but the construction a-priori depends on $L$. What if we fix the $B\subseteq A$? Can we realize states associated to \textit{any} other Lagrangian algebra?

Consider the same physical boundary inclusion $B\subseteq A$ we constructed above from a module category $\mathcal{M}$ and a generator $X\in \mathcal{C}^{*}_{\mathcal{M}}$. Let us call the fusion category $\mathcal{D}:=\mathcal{C}^{*}_{\mathcal{M}}$, so that $B$ is the fusion spin chain associated to $\mathcal{D}$ and the object $X\in \mathcal{D}$. Choose another Lagrangian algebra $K\in \mathcal{Z}(\mathcal{C})$ and consider its image under the isomorphism $\gamma: \mathcal{Z}(\mathcal{C})\cong \mathcal{Z}(\mathcal{D})$. Then $K$ corresponds to an indecomposable (left) module category $\mathcal{N}$ of $\mathcal{D}$. Pick a Q-system $Q\in\mathcal{D}$ such that $\mathcal{N}\cong \mathcal{D}_{Q}$. For some $m$, $Q$ is a summand of $X^{\otimes m}$. By coarse-graining, we can assume $m=1$ (for convenience). Choose an isometry $v:Q\rightarrow X$. Now, let $m:Q\otimes Q\rightarrow Q$ denote the Q-system multiplication, and let $m^{k}:Q^{\otimes k}\rightarrow Q$ denote the $k$-fold associated product.

Then for $f\in B_{I}:=\text{End}_{\mathcal{D}}(X^{\otimes I})$ define the state 

$$\phi_{Q}(f):=\frac{1}{d_{Q}}\text{Tr}_{Q}(m^{I}\circ f \circ (m^{I})^{\dagger}),$$

\noindent where $\text{Tr}_{Q}: \text{End}_{\mathcal{D}}(Q)\rightarrow \mathbbm{C}$ is the standard, unnormalized categorical trace. It is easy to verify that this definition satisfies the desired compatibility conditions, allowing for a well-defined state in the thermodynamic limit of $B$. This state is considered as a boundary state of Levin-Wen models in \cite[Example 5.20]{jones2025localtopologicalorderboundary}. It is easy to show using the same arguments as in \cite[Section 6.3]{jones2025localtopologicalorderboundary} that $\text{H}_{\phi_{Q}}=K$. In particular, this shows that we can construct a topological state for \textit{any} topological boundary condition.

\end{ex}

\begin{remark} The many different flavors of gaplessness, in particular related to local states and their relation to the physical boundary algebras, are discussed extensively in \cite{bhardwaj2024hassediagramsgaplessspt}. 
\end{remark}

While it is straightforward to explicitly write down models of topological states, it appears to be much more difficult to present gapless states. For this reason, even if we have a candidate gapless state (for example, the ground state of some Hamiltonian at a topological critical point, e.g. the critical transverse field Ising chain), it is usually very difficult to argue this state is gapless in the sense we have described above (i.e. strongly non-topological). In the next several sections, we will develop tools in our framework (a Lieb-Schultz-Mattis theorem and an anomalous duality theorem) that allow us to conclude certain states are necessarily gapless, and combine these with results in the theory of C*-algebras to show the existence of such states.

\subsection{Decomposing inductions} We now consider the following abstract situation. Suppose we have an inclusion $B\subseteq A$ of unital C*-algebras with a finite index conditional expectation $E:A\rightarrow B$ (where by finite index, we mean the existence of a projective basis). For convenience, we will assume $A$ and $B$ are both simple. Suppose further that we are given a Hilbert space representation $H$ of $B$. Our goal is to understand the induced representation $A\boxtimes_{B} H$ of $A$, and in particular how it decomposes. In other words, we would like a description of the von Neumann algebra $\mathcal{L}(A\boxtimes_{B} H)\cap A^{\prime}$, where $\mathcal{L}(K)$ denotes the von Neumann algebra of bounded operators on the Hilbert space $K$.

Since $(B\subseteq A, E)$ is finite index, $A$ is a (dual) Q-system internal $\text{Bim}(B)$. We will assume that we have chosen some (full) unitary tensor subcategory $\mathcal{D}\subseteq \text{Bim}(B)$ such that $A\in \mathcal{D}$ as a Q-system.

Then $A$ defines a W*-algebra object in $\mathcal{D}$ via the module category $\mathcal{M}_{B\subseteq A}$ consisting of $B$-$A$ correspondences $X$ such that the $B$-$B$ correspondence $_{B} X\boxtimes_{A} A_{B}\in \mathcal{D}$. This is a semisimple C*-category (and in particular is a W*-category), with a clear left $\mathcal{D}$ module structure under relative $B$ tensor product $\boxtimes_{B}$. Furthermore, $_{B}A_{A}\in \mathcal{M}_{A\subseteq B}$ is a distinguished object, thus applying the construction from the Appendix  we obtain a W*-algebra object 

$$\mathcal{A}:\mathcal{D}^{op}\rightarrow \text{Vec}$$

Note that $A$ is a (dual) Q-system internal to $\text{Bim}(B)$. Suppose that $\mathcal{D}\subseteq \text{Bim}(B)$ is a unitary tensor category with $\mathcal{A}\in \mathcal{D}^{op}\rightarrow \text{Vec}$. More concretely, we can see $\mathcal{A}$ given by 

$$\mathcal{A}(X):=\text{Hom}_{\mathcal{D}}(X, A),$$

\noindent where $A$ is viewed as a (dual Q-system) object in $\mathcal{D}$. In this picture $\mathcal{A}(f)$ is simply precomposition by $f$. In category theory language, $\mathcal{A}$ is \textit{represented} by the object $A$.

Now, suppose $H$ is a Hilbert space representation of $B$. Since $\mathcal{D}\subseteq \text{Bim}(B)$, there is a generated a W*-module category of $\mathcal{D}$, $\mathcal{N}_{H}$, whose objects are summands of Hilbert space representations of the form $_{B}X_{B}\boxtimes_{B} H$, where $X\in \mathcal{D}$. Then $H\in \mathcal{N}_{H}$ is a distinguished object, and we obtain the W* co-algebra object $\mathcal{H}^{op}:\mathcal{D}\rightarrow \text{Vec}$ as defined in the Appendix, so that $\mathcal{H}^{op}(X)=\text{Hom}_{B}(H, X\boxtimes_{B} H)$ is the space of bounded $B$-linear maps between Hilbert space representations of $B$. We have the following theorem

\begin{thm}\label{Endo-realization} There is a canonical isomorphism of $*$-algebras $\text{End}_{A}(A\boxtimes_{B} H)\cong | \mathcal{A}\boxtimes \mathcal{H}^{op}|$. In particular, $| \mathcal{A}\boxtimes \mathcal{H}^{op}|$ is a W*-algebra.
\end{thm}

\begin{proof}

Since $\phi$ is symmetric, $L^{2}(A,\phi)\cong A\boxtimes_{B} L^{2}(B,\phi)$. Thus $\text{End}_{A}(L^{2}(A,\phi))\cong \text{End}_{A}(A\boxtimes_{B}L^{2}(B,\phi))$. We claim this latter space is naturally isomorphic to $\text{Hom}_{B}(L^{2}(B,\phi), A\boxtimes_{B} L^{2}(B,\phi)).$

Indeed, for $f\in \text{End}_{A}(A\boxtimes_{B} L^{2}(B,\phi))$, we define 

$$\widetilde{f}(\xi):=f(1_{A}\boxtimes \xi)\in \text{Hom}_{B}(L^{2}(B,\phi), A\boxtimes_{B}L^{2}(B,\phi),$$

\noindent and for $g\in \text{Hom}_{B}\left(L^{2}(B,\phi), A\boxtimes_{B}L^{2}(B,\phi)\right)$, define

$$\widehat{g}(a\boxtimes \xi):=ag(\xi)\in \text{End}_{A}(A\boxtimes_{B} L^{2}(B,\phi)),$$

\noindent which is well-defined since $g$ is $B$-modular. From the definitions we see $\widehat{\widetilde{f}}=f$ and $\widetilde{\widehat{g}}=g$, establishing the desired isomorphism.

Furthermore, since as a $B$-$B$ bimodule $A\cong \bigoplus_{X\in \text{Irr}(\text{DHR}(B))} \mathcal{A}(X)\otimes X$ we have a $B$-module decomposition $$A\boxtimes_{B} L^{2}(B,\phi)\cong \bigoplus_{X\in \text{Irr}(\text{DHR}(B))}\mathcal{A}(X)\otimes (X\boxtimes_{B} L^{2}(B,\phi)),$$

\noindent where the $A(X)$ simply functions as a finite-dimensional multiplicity space. Therefore, we have

\begin{align*}\text{Hom}_{B}\left(L^{2}(B,\phi), A\boxtimes_{B}L^{2}(B,\phi)\right)&\cong \bigoplus_{X\in \text{Irr}(\text{DHR}(B))} \mathcal{A}(X)\otimes \text{Hom}_{B}\left(L^{2}(B,\phi), X\boxtimes_{B} L^{2}(B,\phi)\right)\\
&=|\mathcal{A}\otimes \mathcal{H}^{op}|.
\end{align*}

Tracking through this isomorphism, it is straightforward to show this is a $*$-algebra isomorphism $\text{End}_{A}(L^{2}(A,\phi))\cong |\mathcal{A}\otimes \mathcal{H}^{op}|$ as desired.
    
\end{proof}

We now  specialize the results on decomposition from the previous to the case of categorically symmetric states.

\begin{thm}\label{thm:rankofmodule}
Let $B\subseteq A$ be a physical boundary subalgebra and $\phi$ a local symmetric state on $A$. Then

\begin{enumerate}
\item    
$L^{2}(A,\phi)$ decomposes as a finite direct sum of inequivalent irreducible representations. 

\item 
The number of irreducible summands is $\text{dim}(\text{Hom}_{\text{DHR}(B)}(H_{\phi},A))=\text{dim}(\text{Hom}_{\mathcal{C}}(F(H_{\phi}),\mathbbm{1}))$, where $F:\text{DHR}(B)\cong \mathcal{Z}(\mathcal{C})\rightarrow \mathcal{C}$ is the forgetful functor. 

\item 
If $\phi$ is topological, the number of summands is $\text{rank}(\mathcal{M}_{\phi})$, where $\mathcal{M}_{\phi}$ is the $\mathcal{C}$-module category associated to $H_{\phi}$ (see Remark \ref{rem:topmodcat}).
\end{enumerate}
\end{thm}

\begin{proof}
If $\phi$ is a local, $H_{\phi}$ is a Q-system representing the compact W*-algebra object $\mathcal{H}_{\phi}$. Thus

$$\text{End}_{A}(L^{2}(A,\phi))\cong |\mathcal{A}\otimes \mathcal{H}^{op}_{\phi}|\cong \text{DHR}(B)(H_{\phi},A)$$

\noindent where the latter space $\text{DHR}(B)(H,A)$ is equipped with the structure of a finite dimensional C*-algebra via the convolution product \cite{MR3308880,MR4357481, schatz2024boundarysymmetries21dtopological}

$$f\ast g:= m_{A}\circ (f\otimes g)\circ m^{\dagger}_{H}.$$

\noindent Since $A,H$ are commutative algebra objects $(\text{DHR}(B)(H,A),\ast)$ is commutative. This implies $L^{2}(A,\phi)$ decomposes as a finite direct sum of inequivalent, irreducible representations, the number of which is precisely $\text{dim}(\text{DHR}(B)(H,A))$. Furthermore, if we define $I:\mathcal{C}\rightarrow \text{DHR}(B)$ as the (bi)adjoint of the free module functor $\text{DHR}(B)\rightarrow \text{DHR}(B)_{A}=\mathcal{C}$ we have an identification $A=I(\mathbbm{1})$. Thus 

$$\text{dim}(\text{DHR}(B)(H,A))=\text{dim}(\text{DHR}(B)(H,I(1)))=\text{dim}(\mathcal{C}(\mathcal{H}_{\phi},\mathbbm{1}).$$ 

\noindent By \cite{MR4309554}, in the case of a Lagrangian algebra this is precisely $\text{rank}(\mathcal{M}_{\phi})$.

\end{proof}

\begin{remark}
There is an alternative categorical picture for this decomposition result. Let $B\subseteq A$ be a physical boundary subalgebra. Let $H$ be a Hilbert space representation of $B$ and $\mathcal{N}_{H}$ denote the W*-module category of $\text{DHR}(B)$ generated by $H$. Then this induces a module category structure on the category of $A$-$A$ bimodules $_{A} \text{DHR}(B)_{A}\cong \mathcal{C}^{mp}\boxtimes \mathcal{C}$, which can be explicitly realized as the category of left $A$-modules internal to $\mathcal{N}_{H}$, which we denote by $\mathcal{M}_{H}$. We also have a functor (the free-module functor) $F_{A}:\mathcal{N}_{H}\rightarrow \mathcal{M}_{H}$ given by $N\mapsto A\boxtimes_{B} N$. Then we have $F_{A}(L^{2}(B,\phi))\cong L^{2}(A,\phi)$. 
\end{remark}

\subsection{Lieb-Schultz-Mattis type theorm}

We can now give a mathematically rigorous version of a well known physics result (for example, \cite{TW24}, which can be interpreted as a Lieb-Schultz-Mattis type theorem for fusion category symmetries. In the case that $\phi$ is topological and $\text{rank}(\mathcal{M}_{\phi})>1$, we say $\phi$ is symmetry breaking, which is precisely the case that $\phi$ is topological and not pure.

\begin{cor}\label{cor:LSM} Let $A\subseteq B$ be a physical boundary subalgebras with symmetry category $\mathcal{C}$, and let $\phi$ be a pure symmetric state on $A$.  If $\mathcal{C}$ has no fiber functor, then $\phi$ is gapless.
\end{cor}

\begin{proof}
If $\phi$ is topological, by Theorem \ref{thm:rankofmodule}, $\mathcal{C}$ has a rank 1 module category, i.e. a fiber functor, which contradicts our assumptions. 

If $\phi$ is not topological then $H_{\phi}$ is a connected Q-system in $\text{DHR}(B)$ with $\text{dim}(\text{Hom}(H_{\phi}, A))=1$, and with $H_{\phi}$ not Lagrangian. We claim $H_{\phi}$ is \textit{not} Morita equivalent to a Lagrangian algebra. 

By \cite[Lemma 4.46]{MR3308880}, if $H_{\phi}$ is Morita equivalent to a commutative algebra $Q$, then the left center $Z_{l}(H_{\phi})\cong Q$ as objects (in fact as algebra objects, though this requires slightly more work). In particular if $H_{\phi}$ were Morita equivalent to a Lagrangian algebra $K$, since $Z_{l}(H_{\phi})\le H_{\phi}$, we have $\text{dim}(\text{Hom}_{\text{DHR}(B)}(K, A))=1$. But then $\mathcal{C}$ has a fiber functor, a contradiction.
\end{proof}

In light of the above result, fusion categories with no fiber functor are called \textit{anomalous} \cite{TW24,PhysRevB.110.035155}.
The above theorem leads naturally to the question: given a topological symmetry decomposition $B\subseteq A$ with symmetry category $\mathcal{C}$, does there exist a symmetric pure state on $A$ in general? If so, then for any fusion categorical symmetry with no fiber functor, we will be guaranteed the existence of a gapless state. 

We will show that this is indeed the case for all models arising from fusion spin chains \ref{ex:MPO}. First we review some technical ideas from the theory of C*-algebras
\cite{MR985307, MR1234394, IzumiBEK}.

\begin{defn} An inclusion of separable C*-algebras $B\subseteq A$ has property BEK if it satisfies the following equivalent conditions:

\begin{enumerate}
\item 
There exists a faithful irreducible representation of $A$ that is irreducible as a representation of $B$.
\item For all $x,y\in A$,
$$\sup_{b\in B} \frac{\|xby\|}{\|b\|}= \|x\| \|y\|.$$
\end{enumerate}
\end{defn}

\noindent The equivalence of these conditions is \cite[Theorem 3.1]{IzumiBEK}.

It is not obvious at first sight how property BEK is related to having a symmetric pure state. Indeed, states directly arising from vector states of a common irreducible representation are essentially never symmetric. However, we have the following lemma, which is a corollary of a result due to Izumi \cite{IzumiBEK}.

\begin{lem} Let $B\subseteq A$ be an inclusion of separable C*-algebras with a finite index conditional expectation $E:A\rightarrow B$ with property BEK. Then there exists a pure state $\phi$ on $A$ with $\phi\circ E=\phi$
\end{lem}

\begin{proof}
By \cite[Theorem 3.5]{IzumiBEK}, $B\subseteq A$ has property BEK if and only if the basic construction inclusion $A\subseteq A_{1}$ bas BEK. Let $H$ be a faithful irreducible representation of $A_{1}$ which restricts to an irreducible representation of $A$. Let $e\in A_{1}$ denote the Jones projection, so that $eae=E(a)e$ for all $a\in A$. Since $H$ is faithful, $e(H)\ne 0$. Take any vector $\xi\in e(H)$ so that $e(\xi)=\xi$, and consider the state on $A$ 
$$\phi(a):=\langle \xi\ | a\xi\rangle.$$

\noindent This is a pure state on $A$ since $H$ is irreducible. Furthermore, we compute
\begin{align*}
\phi(E(a))&=\langle \xi\ |\ E(a)\xi\rangle\\
&=\langle \xi)\ |\ E(a)e(\xi)\rangle\\
&=\langle \xi\ |\ eae(\xi)\rangle\\
&=\langle e(\xi)\ |\ ae(\xi)\rangle\\
&=\phi(a).\\
\end{align*}
\end{proof}
\begin{thm} Suppose $B\subseteq A$ is a physical boundary subalgebra such that $B$ is (bounded spread) isomorphic to a fusion spin chain. Then there exists a symmetric pure state on $A$. 
\end{thm}

\begin{proof}
After passing through a bounded spread isomorphism, we can assume that $B$ is the quasi-local algebra of a fusion spin chain $B:=A(\mathcal{D}, X)$, where $\mathcal{D}$ is a unitary fusion category and $X\in \mathcal{D}$ is a strong tensor generator. Furthmore $A\in \text{DHR}(B)\cong \mathcal{Z}(\mathcal{C})$ is a Lagrangian Q-system, hence $B\subseteq A$ is a categorical inclusion in the sense of Example \ref{ex:MPO}. If we let $\mathcal{M}$ denote the associated indecomposable $\mathcal{D}$-module category, then the quasi-local $A$ is the multifusion spin chain associated to $\text{Mat}_{n}(\text{Hilb}_{f.d.})\cong \text{End}(\mathcal{M})$ with generating object given by the endofunctor $\widetilde{X}:=X\triangleright \cdot$ (here, $n=\text{rank}(\mathcal{M}))$. Alternatively, in subfactor language this can be described as the ``graph planar algebra embedding" of the planar algebra associated to $\mathcal{D}$ and $X$ into the graph planar algebra built from the fusion graph of $X$ on the module $\mathcal{M}$ \cite{MR4598730}.

Now, since $X$ is strongly tensor generating and $\mathcal{M}$ is indecomposable, there exists some $k$ such that the fusion graph for $X^{k}$ on $\mathcal{M}$ connects every vertex to every other vertex by a single edge. In particular, we have that for any two intervals $J$ and $I$ with $d(I,J)>k$, the $*$-algebra homomorphism map $A_{J}\otimes A_{I}\rightarrow A_{J\cup I}$ given by by $a\otimes b\mapsto ab$ is injective, and in particular $\|ab\|=\|a\| \|b\|$ by uniques of norms on finite dimensional C*-algebras.

Now we will show $B\subseteq A$ has property BEK. Let $\epsilon>0$. Let $x,y\in A$ with $\|x\|=\|y\|=1$. It suffices to show that there is a $b\in B$ with $\|b\|=1$ such that $\|xby\|\ge 1-2\epsilon-\epsilon^{2}$.

First, we can choose an interval $I$ such that there are elements $x^{\prime},y^{\prime}\in A_{I}$ with $\|x^{\prime}\|=\|y^{\prime}\|=1$ and $\|x-x^{\prime}\|, \|y-y^{\prime}\|<\epsilon$. Now we will show there is a $b\in B$ with $\|b\|=1$ such that $\|x^{\prime}by^{\prime}\|=\|x^{\prime}\| \|y^{\prime}\|=1$.
To see this, denote $m=|I|$ and let $k$ be as above. Choose $v\in \mathcal{D}(\mathbbm{1}, X^{\otimes m+k})$, so that $v^{\dagger}\circ v=1_{\mathbbm{1}}$. Set

$$b= v\otimes 1_{k+2m}\otimes v^{\dagger}\in \mathcal{D}(X^{3m+2k},X^{3m+2k})=B_{I^{+(m+k)}}.$$

\noindent In pictures, we represent $b$ as

$$
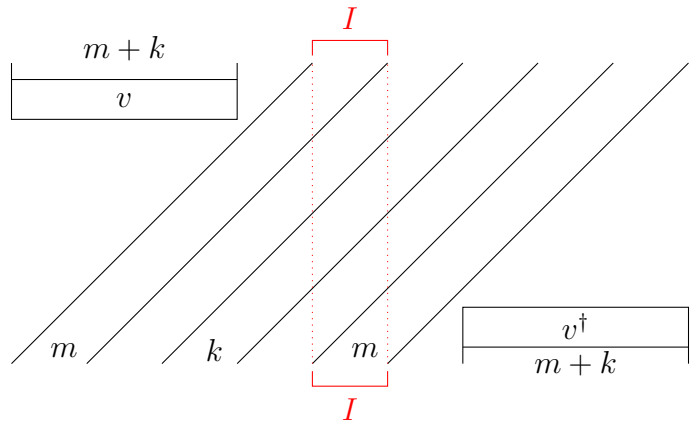
\begin{tikzpicture}
\draw (0,0)--(4,4);
\draw (1,0)--(5,4);
\draw (2,0)--(6,4);
\draw (3,0)--(7,4);
\draw (4,0)--(8,4);
\draw (5,0)--(9,4);
\draw (3,4)--(3,3.25)--(0,3.25)--(0,4);
\draw (0,3.78)--(3,3.78);
\draw (6,0)--(6,0.75)--(9,0.75)--(9,0);
\draw (6,0.22)--(9,0.22);
\draw node at (7.5,0.5) {$v^{\dagger}$};
\draw node at (1.5,3.5) {$v$};
\draw node at (1.5,4.2) {$m+k$};
\draw node at (7.5,0) {$m+k$};
\draw node at (4.7,0.18) {$m$};
\draw node at (2.7,0.18) {$k$};
\draw node at (0.7,0.18) {$m$};
\draw [red] (4,-0.1)--(4,-0.3)--(5,-0.3)--(5,-0.1);
\draw [red] node at (4.5,-0.6) {$I$};
\draw [red] (4,4.1)--(4,4.3)--(5,4.3)--(5,4.1);
\draw [red] node at (4.5,4.6) {$I$};
\draw [dotted, red] (4,0)--(4,4);
\draw [dotted, red] (5,0)--(5,4);
\end{tikzpicture}$$

\noindent

Now, using the algebra model for categorical inclusions \ref{ex:fusionspinsubalg}, we have an arbitrary element of $A_{I}$ represented as $\mathcal{D}(X^{m}, X^{m}\otimes A)$, where we view $A\in \mathcal{D}$ by applying the forgetful functor. Then we see that $x^{\prime}by^{\prime}$ is given by

$$
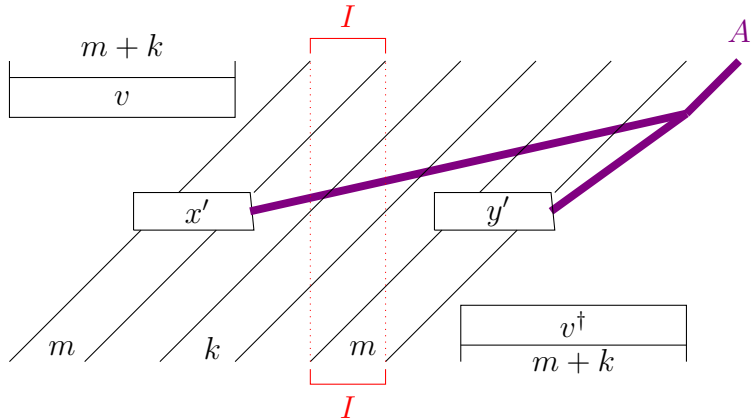
\begin{tikzpicture}
\draw (0,0)--(1.75,1.75);
\draw (1,0)--(2.75,1.75);
\draw (1.65,1.75)--(3.25,1.75)--(3.2,2.25)--(1.65,2.25)--(1.65,1.75);
\draw node at (2.5,2) {$x^{\prime}$};
\draw node at (6.5,2) {$y^{\prime}$};
\draw [violet, line width=3pt] (3.2,2)--(9,3.3);
\draw [violet, line width=3pt] (7.2,2)--(9,3.3);
\draw [violet, line width=3pt] (9,3.3)--(9.7,4);
\draw [violet] node at (9.7,4.4) {$A$};
\draw (2.25,2.25)--(4,4);
\draw (3.25,2.25)--(5,4);
\draw (2,0)--(6,4);
\draw (3,0)--(7,4);
\draw (4,0)--(5.75,1.75);
\draw (5,0)--(6.75,1.75);
\draw (5.65,1.75)--(7.25,1.75)--(7.2,2.25)--(5.65,2.25)--(5.65,1.75);
\draw (6.25,2.25)--(8,4);
\draw (7.25,2.25)--(9,4);
\draw (3,4)--(3,3.25)--(0,3.25)--(0,4);
\draw (0,3.78)--(3,3.78);
\draw (6,0)--(6,0.75)--(9,0.75)--(9,0);
\draw (6,0.22)--(9,0.22);
\draw node at (7.5,0.5) {$v^{\dagger}$};
\draw node at (1.5,3.5) {$v$};
\draw node at (1.5,4.2) {$m+k$};
\draw node at (7.5,0) {$m+k$};
\draw node at (4.7,0.18) {$m$};
\draw node at (2.7,0.18) {$k$};
\draw node at (0.7,0.18) {$m$};
\draw [red] (4,-0.1)--(4,-0.3)--(5,-0.3)--(5,-0.1);
\draw [red] node at (4.5,-0.6) {$I$};
\draw [red] (4,4.1)--(4,4.3)--(5,4.3)--(5,4.1);
\draw [red] node at (4.5,4.6) {$I$};
\draw [dotted, red] (4,0)--(4,4);
\draw [dotted, red] (5,0)--(5,4);
\end{tikzpicture}$$

\noindent
Then using that $v$ is an isometry, we compute $\|x^{\prime}by^{\prime}\|^{2}$ as the norm of the diagram operator

$$
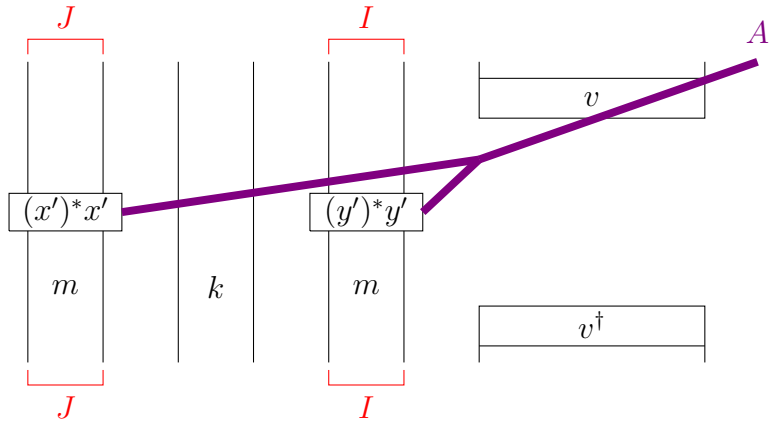
\begin{tikzpicture}
\draw (0,0)--(0,1.75);
\draw (1,0)--(1,1.75);
\draw (-0.25,1.75)--(1.25,1.75)--(1.25,2.25)--(-0.25,2.25)--(-0.25,1.75);
\draw (0,2.25)--(0,4);
\draw (1,2.25)--(1,4);
\draw (2,0)--(2,4);
\draw (3,0)--(3,4);
\draw (4,0)--(4,1.75);
\draw (5,0)--(5,1.75);
\draw (3.75,1.75)--(5.25,1.75)--(5.25,2.25)--(3.75,2.25)--(3.75,1.75);
\draw (4,2.25)--(4,4);
\draw (5,2.25)--(5,4);
\draw (6,0)--(6,0.75)--(9,0.75)--(9,0);
\draw (6,0.22)--(9,0.22);
\draw (6,4)--(6,3.25)--(9,3.25)--(9,4);
\draw (6,3.78)--(9,3.78);
\draw [red] (4,-0.1)--(4,-0.3)--(5,-0.3)--(5,-0.1);
\draw [red] node at (4.5,-0.6) {$I$};
\draw [red] (4,4.1)--(4,4.3)--(5,4.3)--(5,4.1);
\draw [red] node at (4.5,4.6) {$I$};
\draw node at (7.5, 0.5) {$v^{\dagger}$};
\draw node at (7.5, 3.5) {$v$};
\draw node at (0.5, 2) {$(x^{\prime})^{*}x^{\prime}$};
\draw node at (4.5, 2) {$(y^{\prime})^{*}y^{\prime}$};
\draw node at (0.5, 1) {$m$};
\draw node at (2.5, 1) {$k$};
\draw node at (4.5, 1) {$m$};
\draw [violet, line width=3pt] (1.25,2)--(6,2.7);
\draw [violet, line width=3pt] (5.25,2)--(6,2.7);
\draw [violet, line width=3pt] (6,2.7)--(9.7,4);
\draw [violet] node at (9.7,4.4) {$A$};
\draw [red] (0,-0.1)--(0,-0.3)--(1,-0.3)--(1,-0.1);
\draw [red] node at (0.5,-0.6) {$J$};
\draw [red] (0,4.1)--(0,4.3)--(1,4.3)--(1,4.1);
\draw [red] node at (0.5,4.6) {$J$};
\end{tikzpicture}$$

But since $v$ is an isometry, by translation symmetry of fusion categorical nets and the injectivity of the $*$-algebra homomorphism $A_{J}\otimes A_{I}\hookrightarrow A_{J\cup I}$, we obtain the norm of the above diagram is simply $\|x^{\prime}\| \|y^{\prime}\|=1$ as desired. Thus $\|x^{\prime}by^{\prime}\|=1$.

Now, we see that

\begin{align*}
\|xby-x^{\prime}by^{\prime}\|&\le \|x^{\prime}b(y-y^{\prime})\|+\|(x-x^{\prime})by^{\prime}\|+\|(x-x^{\prime})b(y-y^{\prime})\|\\
&\le 2\epsilon+\epsilon^{2}.
\end{align*}

\noindent
Thus $\|xby\|\ge \|x^{\prime}by^{\prime}\|-2\epsilon+\epsilon^{2}=1-2\epsilon-\epsilon^{2}$.
\end{proof}

\begin{cor}\label{cor:existence of gapless}
    For any fusion category $\mathcal{C}$ with no fiber functor (for example, any fusion category with at least one object of non-integer dimension) and for any a physical boundary inclusion $B\subseteq A$ with symmetry category $\mathcal{C}$, there exists a gapless symmetric pure state on $A$.
\end{cor}

\subsection{Dualities of states}\label{sec:duaofstates}

In this section we consider generalizations of Kramers-Wannier (KW) type duality. Again, fix a physical boundary subalgebra $B\subseteq A$ with categorical symmetry $\mathcal{C}$.

\begin{defn}
Given a pair of symmetric states $\phi,\psi$ on $A$, a KW-type duality is a bounded spread isomorphism $\alpha:B\rightarrow B$ such that  $\phi|_{B}\circ \alpha=\psi|_{B}$.
\end{defn}

Note that we can define an associated \textit{duality channel} $\widetilde{\alpha}:A\rightarrow A$ by

$$\widetilde{\alpha}:=\alpha\circ E,$$

\noindent where $E:A\rightarrow B$ is the unique local conditional expectation. Note that the duality channels define a \textit{bijection} on the collection of symmetric states, and are themselves locality preserving. In terms of the duality channel the duality equivalence condition reads $\phi\circ \widetilde{\alpha}=\psi$ on $A$.

Associated to a bounded spread isomorphism $\alpha:B\rightarrow B$ is an autoequivalence $\text{DHR}(\alpha)\in \text{Aut}_{br}(\text{DHR}(B))$. To define this, given a DHR bimodule $X$ of $B$, set

$$X^{\alpha}:=X$$
\noindent as a vector space, with left and right actions defined by

$$a\triangleright \xi\triangleleft b:= \alpha^{-1}(a)\xi\alpha^{-1}(b)$$
$$\langle \xi\ |\ \eta\rangle_{\alpha}:=\alpha(\langle \xi\ |\ \eta\rangle).$$

\noindent It is easy to check that $X^{\alpha}$ is a DHR-bimodule. It is straightforward to construct unitary isomorphisms $X^{\alpha}\boxtimes_{B} Y^{\alpha}\cong (X\boxtimes Y)^{\alpha}$ satisfying the appropriate coherences and which are compatible with the braiding \cite[Theorem B]{jones2024dhr}. Then the assignment $X\mapsto X^{\alpha}$ assembles into a braided autoequivalence of $\text{DHR}(B)$, which we denote by $\text{DHR}(\alpha)$.

Let $H$ be a Hilbert space representation of $B$ and $\alpha:B\rightarrow B$ a bounded spread isomorphism. Let $H^{\alpha}:=H$ as a Hilbert space, but with $B$ action $b\cdot \xi:=\alpha^{-1}(b)\xi$.

\begin{thm} Let $H$ be a Hilbert space representation of $B$ and $\alpha:B\rightarrow B$ a bounded spread isomorphism. Let $\mathcal{H}, \mathcal{H}^{\alpha}$ be W*-algebra objects in $\text{DHR}(B)$ associated to $H$ and $H^{\alpha}$ respectively. Then $\alpha(\mathcal{H})\cong \mathcal{H}^{\alpha}$ as W*-algebra objects.
\end{thm}

\begin{proof}
$$\alpha(\mathcal{H})\cong \bigoplus_{X\in \text{Irr}(\text{DHR})}\text{Hom}_{B}(X\boxtimes_{B}H, H)\otimes X^{\alpha}.$$ 

\noindent On the other hand 

$$\mathcal{H}^{\alpha}\cong \bigoplus_{X\in \text{Irr}(\text{DHR})} \text{Hom}_{B}(X\boxtimes_{B}H^{\alpha}, H^{\alpha})\otimes X .$$

But we claim $X\boxtimes_{B}H^{\alpha}\cong (X^{\alpha}\otimes H)^{\alpha}$.

First we claim that the obvious attempt at a map (using the identity on simple tensors)  $X\boxtimes_{B} H^{\alpha}\cong X^{\alpha^{-1}}\boxtimes_{B} H$ is a (well-defined) unitary of Hilbert spaces.

Indeed 
\begin{align*}
\langle x_{1}\boxtimes_{B}\xi_{1}\ |\ x_{2}\boxtimes_{B}\xi_{2}\rangle_{X\boxtimes_{B} H^{\alpha}}&=\langle \xi_{1}\ | \alpha^{-1}(\langle x_{1}\ |\ x_{2}\rangle_{X})\xi_{2}\rangle_{H}\\
&=\langle \xi_{1}\ | \langle x_{1}\ |\ x_{2}\rangle_{X^{\alpha^{-1}}}\xi_{2}\rangle_{H}\\
&= \langle x_{1}\boxtimes_{B}\xi_{1}\ |\ x_{2}\boxtimes_{B}\xi_{2}\rangle_{X^{\alpha^{-1}}\boxtimes_{B} H}
\end{align*}

Now it is easy to see the left $B$ action on $X^{\alpha^{-1}}\boxtimes_{B} H$ is composes with $\alpha$, and thus we have, as Hilbert space representations of $B$,

$$(X^{\alpha^{-1}}\boxtimes_{B} H)^{\alpha}\cong X\boxtimes_{B} H^{\alpha}.$$

\noindent Thus $\text{Hom}_{B}(X\boxtimes_{B}H^{\alpha}, H^{\alpha})\cong \text{Hom}_{B}((X^{\alpha^{-1}}\boxtimes_{B}H)^{\alpha}, H^{\alpha})=\text{Hom}_{B}((X^{\alpha^{-1}}\boxtimes_{B}H, H)$. Hence we have natural isomorphisms

\begin{align*}
\mathcal{H}^{\alpha}&\cong \bigoplus_{X\in \text{Irr}(\text{DHR})} \text{Hom}_{B}(X\boxtimes_{B}H^{\alpha}, H^{\alpha})\otimes X\\
&\cong \bigoplus_{X\in \text{Irr}(\text{DHR})} \text{Hom}_{B}(X^{\alpha^{-1}}\boxtimes_{B}H, H)\otimes X\\
&\cong \bigoplus_{X\in \text{Irr}(\text{DHR})}\text{Hom}_{B}(X\boxtimes_{B}H, H)\otimes X^{\alpha}\\
&\cong \alpha(\mathcal{H})
\end{align*}

\noindent It is straightforward to check all of these isomorphisms are compatible with the $*$-algebra structures on $\mathcal{H}^{\alpha}$ and $\alpha(\mathcal{H})$, respectively.

\end{proof}

\begin{remark} The above theorem tells us what kinds of states can be obtained from others using this result. In particular, if we know $\mathcal{H}_{\phi}$, we can work out $\mathcal{H}_{\phi\circ \alpha}$. In particular, $\mathcal{H}_{\phi\circ \alpha}\cong \alpha^{-1}(\mathcal{H})$.

We claim  $L^{2}(B,\phi\circ \alpha)^{\alpha}\cong L^{2}(B,\phi)$. Indeed, consider the map $U:L^{2}(B,\phi\circ \alpha)^{\alpha}\rightarrow L^{2}(B,\phi)$ given by 

$$U(b\Omega_{\phi\circ \alpha}):=\alpha(b)\Omega_{\alpha}.$$

For $b_{1},b_{2}$

\begin{align*}
\langle U(b_1\Omega_{\phi\circ \alpha})|U(b_2\Omega_{\phi\circ \alpha})\rangle&=\langle \alpha(b_1)\Omega_{\alpha}\ |\ \alpha(b_2)\Omega_{\alpha} \rangle\\
&=\phi(\alpha(b^{*}_{1}b_{2})\\
&=\langle b_1\Omega_{\phi\circ \alpha}|b_2\Omega_{\phi\circ \alpha}\rangle
\end{align*}

\noindent $U$ is clearly surjective so it extends to a unitary. Now

\begin{align*}
U(b\cdot b_{1}\Omega_{\phi\circ \alpha}))&=U(\alpha^{-1}(b)b_{1}\Omega_{\phi\circ \alpha})\\
&=b\cdot\alpha(b_{1})\Omega_{\phi}
\end{align*}

\noindent Thus $U$ is also an intertwiner.
\end{remark}

Recall there is a special class of bounded spread isomorphism $\alpha: B\rightarrow B$ called a \textit{finite depth circuit} (see \cite[Definition 2.13]{jones2024dhr}). Finite depth circuits are generally considered to provide a useful, Hamiltonian-free operational definition of equivalence of states that preserves long-range entanglement structure, and leads to a useful notion of ``phase". A natural question then is to understand the classification of $\mathcal{C}$-symmetric finite depth circuits. For a phyical boundary subalgebra $B\subseteq A$, a symmetric finite depth circuit is just a finite depth circuit internal to $B$. Note that symmetric finite depth circuits naturally extend to an ordinary finite depth circuit on $A$.

 We have the following corollary, which justifies that the Lagrangian algebra is an invariant of symmetric phases.

\begin{cor}\label{thm:finitedepth}
  If $\phi$ and $\psi$ are $\mathcal{C}$ symmetric states and $\alpha$ is a locally symmetric finite depth circuit such that $\phi\circ \alpha=\psi$, then we have an isomorphism of algebras $\mathcal{H}_{\phi}\cong \mathcal{H}_{\psi}$.
\end{cor}

\begin{proof}
By \cite[Theorem A]{jones2024dhr} the induced braided autoequivalence of on $\text{DHR}(B)$ is trivial. A choice of trivialization in particular gives us an isomorphism $\mathcal{H}_{\phi}\cong \mathcal{H}_{\psi}$.
\end{proof}

Thus our algebra object is an invariant of symmetric states up to finite depth circuits. This leads to the following question,

\begin{quest} Does the the W*-algebra object $\mathcal{H}_{\phi}$ completely classify symmetric topological states up to symmetric finite-depth circuits?
\end{quest}

In the setting MPOs and matrix product states (MPS), a version question has an affirmative answer \cite{GarreRubio2023classifyingphases}.

As another corollary we see that if $\alpha:B\rightarrow B$ is a duality, then the action on $\text{DHR}(B)$ permutes the isomorphism classes of algebra objects in $\text{DHR}(B)$. We will use this to obtain a new criterion for concluding that a state is gapless by its transformation properties under a duality channel. This goes beyond the LSM-type theorem, and allows us to conclude gaplessness of symmetric states even if $\mathcal{C}$ is anomaly free.

\begin{defn}
  A state on a C*-algebra $B$ is called \textit{covariant} under an automorphism $\alpha$ if there is a unitary equivalence $L^{2}(B,\phi)^{\alpha}\cong L^{2}(B,\phi)$.
\end{defn}

\begin{prop} \cite[Theorem 2]{MR905025} For any simple, separable C*-algebra $B$ and any outer automorphism $\alpha$, there exists a pure $\alpha$ covariant state.
\end{prop}

\begin{cor}\label{cor:anomalousLagrangian}
Let $B\subseteq A$ be a physical boundary subalgebra and $\alpha:B\rightarrow B$ a bounded spread isomorphism such that $\text{DHR}(\alpha)$ leaves no Lagrangian algebra in $\text{DHR}(B)$ fixed. Then any $\alpha$-covariant connected symmetric state is gapless. If $B$ is simple (for example, the quasi-local algebra of a fusion spin chain), then such a state always exists.
\end{cor}

\begin{proof}
If $\phi$ were topological then by definition $\mathcal{H}_{\phi}$ is a Lagrangian algebra, and covariance of $\phi$ under $\alpha$ implies $\alpha(\mathcal{H}_{\phi})\cong \mathcal{H}_{\phi}$. 

If $\phi$ is not topological, we need to show that $H_{\phi}$ is not Morita equivalent to a Lagrangian algebra. Since $\alpha(H_{\phi})\cong H_{\phi}$, $\alpha$ preserves the Morita class of $H_{\phi}$. In particular, if this class contained a Lagrangian algebra $L$, then $\alpha(L)$ would be another Lagrangian in the same Morita class as $H_{\phi}$, and in particular would be Morita equivalent to $L$. But by \cite[Lemma 4.46]{MR3308880}, Morita equivalent commutative algebras are isomorphic (as objects).

The existence of such a state for simple $B$ will follow if $\alpha$ is outer. But this follows, since inner automorphsims induce the trivial autoequivalence on $\text{DHR}(B)$, and such an autoequivalence will fix \textit{all} Lagrangian algebras.
\end{proof}

In light of the above result, we call bounded spread isomorphisms $\alpha:B\rightarrow B$ such that $\text{DHR}(\alpha)$ leaves no Lagrangian algebra fixed \textit{anomalous}. 

\begin{ex}{\textbf{Kramers-Wannier}}. In this example, we consider the usual Kramers-Wannier duality in our context. The symmetry category is $\text{Vec}(\mathbbm{Z}/2\mathbbm{Z})$, and is implemented by the onsite spin flip action $\mathbbm{Z}/2\mathbbm{Z}\curvearrowright\mathbbm{C}^{2}$, extended diagonally to a an action $\mathbbm{Z}/2\mathbbm{Z}\curvearrowright A=\otimes_{\mathbbm{Z}}M_{2}(\mathbbm{C})$. The symmetric subalgebra $B$ is naturally isomorphic to the fusion spin chain built from the category $\text{Rep}(\mathbbm{Z}/2\mathbbm{Z})$ with the regular representation (i.e. the spin flip action on $\mathbbm{C}^{2}$).

In this picture, we identify $\text{DHR}(B)\cong \mathcal{Z}(\text{Rep}(\mathbbm{Z}/2\mathbbm{Z}))$, which has 4 anyon types $\mathbbm{1}, e, m, f$, where $e$ and $m$ are bosons and $f$ is a fermion. We pick the convention that $\{1,e\}\cong \text{Rep}(\mathbbm{Z}/2\mathbbm{Z})$ with the standard lift of $\text{Rep}(\mathbbm{Z}/2\mathbbm{Z})$ to $\mathcal{Z}(\text{Rep}(\mathbbm{Z}/2\mathbbm{Z}))$ using the symmetric braiding as the half-braiding. There are precisely two Lagrangian algebras: $L_{1}:=1+e$ and $L_{2}:=1+m$. The quasi-local algebra is identified with $A=1\oplus e$ as a Lagrangian algebra object.

In this set up $L_{1}$ will correspond to the rank two $\text{Vec}(\mathbbm{Z}/2\mathbbm{Z})$ module category (i.e. the regular module category) and $L_{2}$ corresponds to the rank one module category given by the fiber functor. We have specific $\mathbbm{Z}/2\mathbbm{Z}$ symmetric states that realize these two Lagrangian algebras. In the standard qubit basis, we have the mixed state, symmetry breaking state 

$$\phi_{1}=\frac{1}{2} \left(|0\rangle^{\otimes \mathbbm{Z}}+ |1\rangle^{\otimes \mathbbm{Z}} \right) $$

\noindent which realizes the Lagrangian algebra $L_{1}$. Define the pure state

$$\phi_{2}:=\left(\frac{1}{\sqrt{2}}|0\rangle+\frac{1}{\sqrt{2}}|1\rangle \right)^{\otimes \mathbbm{Z}}.$$ 

\noindent This realizes the Lagrangian algebra $L_{2}$. Both of these states can be realized as ground states of the transverse field Ising model, with Hamiltonian

$$H_{J,h}=-J\sum_{i\in \mathbbm{Z}} \sigma^{z}_{i}\sigma^{z}_{i+1}-h\sum_{i\in \mathbbm{Z}} \sigma^{x}_{i}.$$

\noindent We see that $\alpha(H_{J,h})=H_{h,J}$. $\phi_{1}$ is the unique symmetric ground state for $J=0$ and describes the universality class for $h>J$, while $\phi_{2}$ is the unique symmetric ground state for $h=0$ and describes the universality class for $J \geq h$ \cite{MR727197, MR810491, MR821288}. It is clear that there is a phase transition at $J=h$, at which point we switch from a symmetry breaking phase to a unique gapped ground state.\\

As an algebra, $B$ is generated by the Pauli operators $\{\sigma^{x}_{i},\ \sigma^{z}_{i}\sigma^{z}_{i+1}\}_{i\in \mathbbm{Z}}$. Kramers-Wannier duality $\alpha:B\rightarrow B$ is defined on generators by

$$\alpha(\sigma^{x}_{i}):=\sigma^{z}_{i-1}\sigma^{z}_{i}$$
$$\alpha(\sigma^{z}_{i}\sigma^{z}_{i+1})=\sigma^{x}_{i}.$$

In \cite{JoLi24, jones2025quantumcellularautomatacategorical} it is shown that $\text{DHR}(\alpha)$ implements the unique symmetry that exchanges $e$ and $m$, $e\leftrightarrow m$. In particular $\text{DHR}(\alpha)(L_{1})=L_{2}$ and $\text{DHR}(\alpha)(L_{2})=L_{1}$.

Thus there are no fixed Lagrangian algebras, and by the previous corollary, this implies there exists a $\mathbbm{Z}_{2}$-symmetric connected state $\phi_{0}$ which is gapless. In particular, it has associated to it a connected algebra object $\mathcal{H}_{\phi_{0}}$, which is stable under $\text{DHR}(\alpha)$. This implies either $\mathcal{H}_{\phi_{0}}\cong \mathbbm{1}$ or $\mathcal{H}_{\phi_{0}}\cong \mathbbm{1}\oplus \epsilon $. If $\phi_{0}$ satisfies the ``bounded spread cone Haag duality" condition on $B$, then we must have $\mathcal{H}_{\phi_{0}}\cong \mathbbm{1}$.
\end{ex}

\begin{ex}{\textbf{Kramers-Wannier for arbitrary abelian groups}}

Here we recall the construction of \cite[Example 4.4]{JoLi24}. Let $G$ be an abelian group, and consider the regular representation $G\curvearrowright \ell^{2}(G)\cong \mathbbm{C}^{n}$, where $n=|G|$. Extending this to an on-site action of $G$ on the quasi-local algebra $A:=\otimes_{\mathbbm{Z}} M_{n}(\mathbbm{C})$, we obtain the physical boundary subalgebra of fixed points $B\subseteq A$.  But as noted in Example \ref{ex:SymsubAlg2}, $B$ is the fusion spin chain associated to the fusion category $\text{Rep}(G)\cong \text{Vec}(\widehat{G})$, with generating object $X=\mathbbm{C}[A]$ the regular representation. Viewed in $\text{Vec}(\widehat{G})$, $X=\bigoplus_{a\in \widehat{G}} a$.

Now, let $\mathcal{D}=\mathcal{TY}(\widehat{G},\beta)$ be a Tambara-Yamagami category associated to $\widehat{
G}$, where $\beta:\widehat{G}\times \widehat{G}\rightarrow \text{U}(1)$ is a non-degenerate, symmetric bicharacter\footnote{these categories also depend on a choice of sign, we will pick the $+$-version.} The simple objects are $\{\rho\}\cup \widehat{G}$, with fusion rules given by

$$a\otimes b\cong ab,$$

$$\rho\otimes \rho\cong \bigoplus_{a\in \widehat{G}} a,$$

$$a\otimes \rho\cong \rho\otimes a\cong \rho.$$

Now, from the fusion rules, we notice that $X=\rho\otimes \rho\in \text{Vec}(\widehat{G})$. In particular, the fusion spin chain for $\text{Vec}(\widehat{G})$ with $X$ is equal to the fusion spin chain for $\mathcal{TY}(\widehat{G}, \beta)$ with object $\rho \otimes \rho$. However, in the Tambara-Yamagami picture, we have a natural ``shift by $1$", which is a bounded spread isomorphism $\sigma: B\rightarrow B$. 

As explained in \cite{JoLi24}, $$\text{DHR}(\sigma)\in \text{DHR}(B)\cong \mathcal{Z}(\text{Vec}(\widehat{G}))\cong \text{Vec}(\widehat{G}\times \widehat{\widehat{G}}, q)=\text{Vec}(\widehat{G}\times G, q),$$ where $q(a,g)=a(g)$.
$\text{DHR}(\sigma)$ is the order 2 automorphism defined by sending $(a, g)\rightarrow (S_{\beta}(g),S_{\beta}(a))$. Here $S_{\beta}$ denotes the $\beta$-Fourier transform, i.e. for $a\in \widehat{G}$, $S_{\beta}(a)=\beta(a, \cdot)\in \widehat{\widehat{G}}\cong G$ and for $g\in G$, $S_{\beta}(g)$ is the unique $a\in \widehat{G}$ with $\beta(S_{\beta}(g),a)=a(g)$.

Under this identification, the Lagrangian algebras are the subgroups of $L\le \widehat{G}\times G$ with $q|_{L}=1$ and $|L|=|G|$. These are in bijective correspondences are given pairs $(H,b)$, where $H\le G$, and $b:H\times H\rightarrow \text{U}(1)$ is an \textit{antisymmetric} bi-character \cite{MR2677836}. Then 

$$L_{(H,b)}:=\{(\phi,h)\ :\ \phi(k):=b(h,k)\ \text{for all}\ k\in H\}.$$

Now suppose $G=\mathbbm{Z}/n\mathbbm{Z}$ with $n\ne m^{2}$ for any integer $m$. Then we claim $\alpha $ is anomalous. Indeed, in this case since $H^{2}(H,\text{U}(1))$ is trivial for any cyclic $H$ this implies there is a unique alternating bicharacter, and thus the Lagrangian algebras are directly parameterized by subgroups $L_{H}=L_{H,1}$, and now we see that $L_{H}:=\{(\phi,h)\ : \phi\in H^{\perp},\ h\in H\}$. Thus if we denote by $\pi_{1}$ and $\pi_{2}$ the projection onto the first and second coordinates respectively, then $|\pi_{1}(L_{H})|=|H^{\perp}|$ and $|\pi_{2}(L_{H})|=|H|$. But $|\pi_{1}(\text{DHR}(\sigma)(L_{H}))|=|\pi_{2}(L_{H})|=|H|$ and $|\pi_{2}(\text{DHR}(\sigma)(L_{H}))|=|\pi_{1}(L_{H})|=|H^{\perp}|$. Thus if $\text{DHR}(\sigma)(L_{H})\cong L_{H}$, we would have $|H|=|H^{\perp}|$. But $|G|=|H||H^{\perp}|$, which contradicts our hypothesis that $n$ is not a square. Thus the generalized Kramers-Wannier is anomalous.

\end{ex}

\begin{ex}{$\textbf{PSU}(2)_{k}$}
In this example, we discuss a natural generalization of Kramers-Wannier duality that arises in subfactor theory. We let $\mathcal{D}$ be the unitary fusion category $\textbf{PSU}(2)_{k}$, which is the trivially graded component of the category of (unitary) representations of $sl_{2}$ at level $k$. The simple objects of $SU(2)_{k}$ are indexed by $X_{l}$ where the $l$ are half-integer spins,  $l=0, \frac{1}{2},\dots \frac{k}{2}$, and $\text{PSU}(2)_{k}$ consists precisely of the integer spin objects. $\text{PSU}(2)_{k}$ is also realized as the even part of the $A_{k+1}$ subfactor, which can be expressed in terms of the semi-simplification of the TLJ planar algebra with loop parameter $\delta:=2\text{cos}(\frac{\pi}{k+2})$.

Consider the object $X:=X_{0}\oplus X_{1}$, and consider the fusion spin chain $B=A(\mathcal{D},X)$. We will define a bounded spread isomorphims $\alpha:B\rightarrow B$ which restricts to Kramers-Wannier described above when $k=2$.

We have an isomorphism $\text{End}_{\mathcal{D}}(X^{\otimes n})\cong \text{TL}_{2n}(\delta)$, where the latter is the 2n-strand Temperly-Lieb algebra with loop parameter $\delta=2\text{cos}(\frac{\pi}{k+1})$. This algebra is generated by $2n-1$ projection $\{e_{1},\dots e_{2n-1}\}$ called the \textit{Jones projections} which satisfy the relations

$$e_{i}e_{j}=e_{j}e_{i},\ |i-j|\ge 2,$$
$$e_{i}e_{i\pm1 }e_{i}=\frac{1}{\delta^{2}}e_{i}.$$

There is an additional ``semisimplification" relation depending on $k$ (amounting to setting the appropriate Jones-Wenzl idempotent equal to $0$).

This allows us to define a bounded spread isomorphism $\alpha:B\rightarrow B$, 

$$\alpha(e_{i}):=e_{i-1}.$$

\noindent One can check that for $k=2$, this recovers the Kramers-Wannier automorphism. If we define the Hamiltonian

$$H_{J,h}=-J\sum_{i} e_{2i}-h\sum_{j} e_{2j+1}$$
then $\alpha(H_{J,h})=H_{h,J}$, precisely as in the Kramers-Wannier case.

When $k = 4$, the fusion category of the even part of $SU(2)_4$ is the 
fusion ring of Rep$(S_{3})$ with three (isomorphism classes of) simple objects by $\mathbbm{1}, \pi, \epsilon$, where $\mathbbm{1}$ is trivial, $\sigma$ is the sign representation, and $\pi$ is the two dimenisonal representation of $S_{3}$. There are four Lagrangian algebras in $\mathcal{Z}(\text{Rep}(S_{3}))$ corresponding to the four subgroups $1, \mathbb{Z}_2, \mathbb{Z}_3, S_{3}$. We can pick any one of them to give a physical boundary inclusion $B\subseteq A$ (note that two will give $\text{Rep}(S_{3})$ symmetry on $A$, and two will give $\text{Vec}(S_{3})$ symmetry).

We claim that the $\alpha$ described above is anomalous in this case. By combining \cite[Section 4.2]{JoLi24} and \cite[Example 8.1]{MR3210925}, the induced homomorphism from $\mathbbm{Z}/2\mathbbm{Z}$ to \\ $\text{Aut}_{br}(\mathcal{Z}(\text{Rep}(S_{3})))\cong \text{Aut}_{br}(\mathcal{Z}(\text{Vec}(S_{3})))$ is actually an isomorphism. We know that the two Lagrangian algebras corresponding to $\text{Vec}(S_{3})$ are permuted. But each of these has a \textit{uniquely determined} ``dual" Lagrangian algebra corresponding to $\text{Rep}(S_{3})$ that arise from a fiber functor. Thus this $\mathbbm{Z}/2\mathbbm{Z}$ autoequivalence must swap these as well. In particular this Kramers-Wannier is anomalous.

\end{ex}

\begin{ex}{$\textbf{PSU}(n)_{k}?$}\ \ As in the above construction, one can take the zero graded fusion subcategory of $SU(n)$ at level $k$, which we will call $\textbf{PSU}(n)_{k}$. For example, $SU(3)$ at level $3$ has ten simple objects of which four have grade zero, corresponding to the trivial object and its rotation by the center ${\mathbb Z}_3$, as well as the one fixed object.
This fusion category is the representation ring {Rep}$(A_{4})$ of the alternating group $A_4$ \cite{MR1301620}. There are seven Lagrangian objects 
$\mathcal{Z}(\text{Rep}(A_{4}))$ corresponding to the four subgroups $1, \mathbb Z_2,  \mathbb Z_3, A_4$ of $A_4$, together with a multiplicity for the Klein subgroup $\mathbb Z_2 \times \mathbb Z_2$ and for $A_{4}$ itself coming from the two Schur multipliers in the first case, and a single additional Schur multiplier in the latter. Since there are seven Lagrangian algebras, the action of $\mathbbm{Z}/3\mathbbm{Z}$ must leave at least one fixed (since any orbit contains either 1 or 3 elements). Thus this version of Kramers-Wannier is \textit{not anomalous}.

This leads to the interesting question: when are the shifts associated to $\text{SU}(n)$ at level $k$  anomalous?

\end{ex}

\appendix
\section{W*-algebra objects and realizations}\label{app:W*} We briefly recall the idea of a W*-algebra object internal to a unitary tensor category $\mathcal{D}$, introduced in \cite{MR3687214,MR3948170}. Given a (left) $\mathcal{D}$-module category $\mathcal{M}$, and an object $m\in \mathcal{M}$, we can define the lax monoidal functor $\mathcal{A}:\mathcal{D}^{op}\rightarrow \text{Vec}$, by setting

$$\mathcal{A}(X)=\mathcal{M}(X\triangleright m, m),$$

\noindent and for $f:X\rightarrow Y$ and $\xi\in A(Y)$

$$\mathcal{A}(f)(\xi):= \xi\circ (f\triangleright 1_{m}).$$

\noindent The lax tensorator is given by $\mu_{X,Y}:\mathcal{A}(X)\otimes \mathcal{A}(Y)\rightarrow \mathcal{A}(X\otimes Y)$ defined via

$$\mu_{X,Y}(\xi\otimes \eta):=\xi\circ (1_{X}\triangleright(\eta)).$$

There is also an involutive structure on $\mathcal{A}$ in the sense of \cite{MR3687214}, which consists of a conjugate linear isomorphisms $j:A(X)\rightarrow A(\overline{X})$, defined via

$$j(\xi):=(\text{ev}_{\overline{X},X}\triangleright 1_{m})\circ (1_{\overline{X}}\triangleright \xi^{\dagger})\in A(\overline{X}).$$

\noindent This $j$ is anti-tensorial in the sense of \cite{MR3687214} and thus makes $\mathcal{A}$ into a $*$-algebra object. The fact that this comes from a W*-module category makes this into a \textit{W*-algebra} object.

Similarly, we could define a W*-coalgebra object $\mathcal{A}^{op}$ with is a lax monoidal functor $\mathcal{A}^{op}:\mathcal{D}\rightarrow \text{Vec}$,

$$\mathcal{A}^{op}(X):=\mathcal{M}(m, X\triangleright m).$$

\noindent The other structure maps follow in the obvious way.

\begin{remark} We note that if $\mathcal{D}$ is a fusion category and each vector space $\mathcal{A}(X)$ is finite dimensional, then by the (enriched) Yoneda lemma, there is a distinguished object $A\in \mathcal{D}$ such that $\mathcal{A}(X):=\mathcal{D}(X, A)$ (or in the coalgebra case, $\mathcal{D}(A,X)$). The W*-algebra structure makes $A$ into an algebra object in $\mathcal{D}$. There is one subtlety here. The algebra structure is only well-defined up to non-unitary isomorphsim. In particular, we could use the Yoneda lemma to lift the lax monoidal structure on $\mathcal{A}$ to two different algebra multiplications on $A$ which are not unitarily isomorphic in $\mathcal{D}$ (but will always be isomorphic as algebras). However, it is shown in \cite{MR4079745} that a lift can always be chosen to make $A$ into a (dual) Q-system.
\end{remark}

Now, given two  W* $\mathcal{D}$-module categories $\mathcal{M}$ and $\mathcal{N}$ with objects $m$ and $n$ respectively, consider the associated W*-algebra and co-algebra objects denoted $\mathcal{A}$ and $\mathcal{B}^{op}$, respectively. We define an (ordinary) associative $*$-algebra called the \textit{realization} of these objects, which from a categorical perspective could be defined as the (enriched) coend of the functor $\mathcal{A}\times \mathcal{B}^{op}:\mathcal{D}\times \mathcal{D}^{op}\rightarrow \text{Vec}$. We will be more concrete and give an explicit model. Let

$$|\mathcal{A}\otimes \mathcal{B}^{op}|:=\bigoplus_{X\in \text{Irr}(\mathcal{C})} \mathcal{A}(X)\otimes\mathcal{B}^{op}(X),$$

\noindent with multiplication defined on simple tensors $\xi_{1}\otimes \eta_{1}\in \mathcal{A}(X)\otimes\mathcal{B}^{op}(X)$ and $\xi_{2}\otimes \eta_{2}\in \mathcal{A}(Y)\otimes\mathcal{B}^{op}(Y)$ by 

$$(\xi_{1}\otimes \eta_{1})\cdot (\xi_{2}\otimes \eta_{2})=\sum_{Z,\alpha} \mathcal{A}(\alpha)(\xi_{1}\otimes \xi_{2})\otimes \mathcal{B}^{op}(\alpha^{*})(\eta_{1}\otimes \eta_{2}),$$

\noindent where for each $Z\in \text{Irr}(\mathcal{D})$, $\alpha$ ranges over a choice of basis for $\mathcal{D}(Z,X\otimes Y)$ such that $\beta^{\dagger}\circ \alpha=\delta_{\alpha,\beta} 1_{Z}$. The $*$ structure on the algebra is given on simple tensors simply by 

$$(\xi_{1}\otimes \eta_{1})^{*}:=j(\xi_{1})\otimes j(\eta_{1})\in \mathcal{A}(\overline{X})\otimes \mathcal{B}^{op}(\overline{X}).$$

For a graphical representation of the above construction, see \cite[Section 4]{MR3948170}. In general, the realization $|\mathcal{A}\otimes \mathcal{B}^{op}|$ is only an associative $*$-algebra, but there are various C* and W*-algebraic completions available (similar to the situation of ordinary tensor products of C* and W*-algebras, which is exactly recovered in the case $\mathcal{D}=\text{Hilb}_{f.d}$). Nevertheless, there are many situations in which the realization is already a C*/W*-algebra without completion, in particular when one of $\mathcal{A}$ or $\mathcal{B}^{op}$ is finitely supported (for example, see \cite[Corollary 5.14]{MR3948170}).

\bibliographystyle{alpha}
{\footnotesize{
\bibliography{bibliography}

\newcommand{\etalchar}[1]{$^{#1}$}
\begin{thebibliography}{BRWdOC25}

\bibitem[ABGE{\etalchar{+}}23]{ABGIH}
Fabio Apruzzi, Federico Bonetti, I{\~n}aki Garc{\'\i}a~Etxebarria, Saghar~S.
  Hosseini, and Sakura Sch{\"a}fer-Nameki.
\newblock Symmetry {TFTs} from string theory.
\newblock {\em Comm. Math. Phys.}, 402(1):895--949, 2023.

\bibitem[AE83]{MR727197}
Huzihiro Araki and David~E. Evans.
\newblock On a {$C\sp{\ast} $}-algebra approach to phase transition in the
  two-dimensional {I}sing model.
\newblock {\em Comm. Math. Phys.}, 91(4):489--503, 1983.

\bibitem[AFM20]{aasen2020topologicaldefectslatticedualities}
David Aasen, Paul Fendley, and Roger S.~K. Mong.
\newblock Topological defects on the lattice: Dualities and degeneracies.
\newblock 2020.
\newblock \arXiv{2008.08598}.

\bibitem[AKW26]{ahmad2026facesnoninvertiblesymmetries}
Shadi~Ali Ahmad, Marc~S. Klinger, and Yifan Wang.
\newblock The many faces of non-invertible symmetries, 2026.

\bibitem[AM85a]{MR821288}
Huzihiro Araki and Taku Matsui.
\newblock {$C^\ast$}-algebra approach to ground states of the {$XY$}-model.
\newblock In {\em Statistical physics and dynamical systems ({K}\"oszeg,
  1984)}, volume~10 of {\em Progr. Phys.}, pages 17--39. Birkh\"auser Boston,
  Boston, MA, 1985.

\bibitem[AM85b]{MR810491}
Huzihiro Araki and Taku Matsui.
\newblock Ground states of the {$XY$}-model.
\newblock {\em Comm. Math. Phys.}, 101(2):213--245, 1985.

\bibitem[AMF16]{MR3543452}
David Aasen, Roger S.~K. Mong, and Paul Fendley.
\newblock Topological defects on the lattice: {I}. {T}he {I}sing model.
\newblock {\em J. Phys. A}, 49(35):354001, 46, 2016.

\bibitem[BBC{\etalchar{+}}24]{bhardwaj2024superselectionsectorsposetsvon}
Anupama Bhardwaj, Tristen Brisky, Chian~Yeong Chuah, Kyle Kawagoe, Joseph
  Keslin, David Penneys, and Daniel Wallick.
\newblock Superselection sectors for posets of von {N}eumann algebras.
\newblock 2024.
\newblock \arXiv{2410.21454}.

\bibitem[BBPSN24]{PhysRevLett.133.161601}
Lakshya Bhardwaj, Lea~E. Bottini, Daniel Pajer, and Sakura Sch\"afer-Nameki.
\newblock Categorical {L}andau paradigm for gapped phases.
\newblock {\em Phys. Rev. Lett.}, 133:161601, Oct 2024.

\bibitem[BBPSN25]{MR4861493}
Lakshya Bhardwaj, Lea~E. Bottini, Daniel Pajer, and Sakura Sch\"afer-Nameki.
\newblock Gapped phases with non-invertible symmetries: {$(1+1){\rm d}$}.
\newblock {\em SciPost Phys.}, 18(1):Paper No. 032, 113, 2025.

\bibitem[BBSNT25a]{PhysRevB.111.054432}
Lakshya Bhardwaj, Lea~E. Bottini, Sakura Sch\"afer-Nameki, and Apoorv Tiwari.
\newblock Illustrating the categorical {L}andau paradigm in lattice models.
\newblock {\em Phys. Rev. B}, 111:054432, Feb 2025.

\bibitem[BBSNT25b]{bhardwaj2025latticemodelsphasestransitions}
Lakshya Bhardwaj, Lea~E. Bottini, Sakura Schafer-Nameki, and Apoorv Tiwari.
\newblock Lattice models for phases and transitions with non-invertible
  symmetries, 2025.
\newblock \arXiv{2405.05964}.

\bibitem[BE99]{MR1671970}
Jens B\"ockenhauer and David~E. Evans.
\newblock Modular invariants, graphs and {$\alpha$}-induction for nets of
  subfactors. {II}.
\newblock {\em Comm. Math. Phys.}, 200(1):57--103, 1999.

\bibitem[BE00]{MR1785458}
Jens B\"ockenhauer and David~E. Evans.
\newblock Modular invariants from subfactors: {T}ype {I} coupling matrices and
  intermediate subfactors.
\newblock {\em Comm. Math. Phys.}, 213(2):267--289, 2000.

\bibitem[BEEK89]{MR985307}
Ola Bratteli, George~A. Elliott, David~E. Evans, and Akitaka Kishimoto.
\newblock Quasi-product actions of a compact abelian group on a
  {$C^*$}-algebra.
\newblock {\em Tohoku Math. J. (2)}, 41(1):133--161, 1989.

\bibitem[BEK93]{MR1234394}
Ola Bratteli, George~A. Elliott, and Akitaka Kishimoto.
\newblock Quasi-product actions of a compact group on a {$C^*$}-algebra.
\newblock {\em J. Funct. Anal.}, 115(2):313--343, 1993.

\bibitem[BEK99]{MR1729094}
Jens B\"ockenhauer, David~E. Evans, and Yasuyuki Kawahigashi.
\newblock On {$\alpha$}-induction, chiral generators and modular invariants for
  subfactors.
\newblock {\em Comm. Math. Phys.}, 208(2):429--487, 1999.

\bibitem[BEK00]{MR1777347}
Jens B\"ockenhauer, David~E. Evans, and Yasuyuki Kawahigashi.
\newblock Chiral structure of modular invariants for subfactors.
\newblock {\em Comm. Math. Phys.}, 210(3):733--784, 2000.

\bibitem[BG17]{MR3719546}
Matthew Buican and Andrey Gromov.
\newblock Anyonic chains, topological defects, and conformal field theory.
\newblock {\em Comm. Math. Phys.}, 356(3):1017--1056, 2017.

\bibitem[Bis97]{MR1424954}
Dietmar Bisch.
\newblock Bimodules, higher relative commutants and the fusion algebra
  associated to a subfactor.
\newblock In {\em Operator algebras and their applications (Waterloo, ON,
  1994/1995), 13-63, Fields Inst. Commun., 13}. Amer. Math. Soc., Providence,
  RI, 1997.
\newblock \mathscinet{MR1424954}, \googlebooks{_InIRTO8Y7gC}.

\bibitem[Bis17]{MR3595480}
Marcel Bischoff.
\newblock Generalized orbifold construction for conformal nets.
\newblock {\em Rev. Math. Phys.}, 29(1):1750002, 53, 2017.

\bibitem[BJ22]{MR4357481}
Marcel Bischoff and Corey Jones.
\newblock Computing fusion rules for spherical {$G$}-extensions of fusion
  categories.
\newblock {\em Selecta Math. (N.S.)}, 28(2):Paper No. 26, 39, 2022.

\bibitem[BKLR15]{MR3308880}
Marcel Bischoff, Yasuyuki Kawahigashi, Roberto Longo, and Karl-Henning Rehren.
\newblock {\em Tensor categories and endomorphisms of von {N}eumann
  algebras---with applications to quantum field theory}, volume~3 of {\em
  SpringerBriefs in Mathematical Physics}.
\newblock Springer, Cham, 2015.
\newblock \mathscinet{MR3308880} \doi{10.1007/978-3-319-14301-9}.

\bibitem[BMPS12]{MR2979509}
Stephen Bigelow, Scott Morrison, Emily Peters, and Noah Snyder.
\newblock Constructing the extended {H}aagerup planar algebra.
\newblock {\em Acta Math.}, 209(1):29--82, 2012.
\newblock \mathscinet{MR2979509}, \arXiv{0909.4099},
  \doi{10.1007/s11511-012-0081-7}.

\bibitem[BMW{\etalchar{+}}17]{MR3614057}
N.~Bultinck, M.~Mari\"{e}n, D.~J. Williamson, M.~B. \c{S}ahino\u{g}lu,
  J.~Haegeman, and F.~Verstraete.
\newblock Anyons and matrix product operator algebras.
\newblock {\em Ann. Physics}, 378:183--233, 2017.

\bibitem[BO08]{MR2391387}
Nathanial~P. Brown and Narutaka Ozawa.
\newblock {\em {$C^*$}-algebras and finite-dimensional approximations},
  volume~88 of {\em Graduate Studies in Mathematics}.
\newblock American Mathematical Society, Providence, RI, 2008.

\bibitem[BPSNW24]{bhardwaj2024hassediagramsgaplessspt}
Lakshya Bhardwaj, Daniel Pajer, Sakura Schafer-Nameki, and Alison Warman.
\newblock Hasse diagrams for gapless {SPT} and {SSB} phases with non-invertible
  symmetries.
\newblock 2024.
\newblock \arXiv{2403.00905}.

\bibitem[BR87]{MR887100}
Ola Bratteli and Derek~W. Robinson.
\newblock {\em Operator algebras and quantum statistical mechanics. 1}.
\newblock Texts and Monographs in Physics. Springer-Verlag, New York, second
  edition, 1987.
\newblock $C^\ast$- and $W^\ast$-algebras, symmetry groups, decomposition of
  states.

\bibitem[BR97]{MR1441540}
Ola Bratteli and Derek~W. Robinson.
\newblock {\em Operator algebras and quantum statistical mechanics. 2}.
\newblock Texts and Monographs in Physics. Springer-Verlag, Berlin, second
  edition, 1997.
\newblock Equilibrium states. Models in quantum statistical mechanics.

\bibitem[BRWdOC25]{bols2025classificationlocalitypreservingsymmetries}
Alex Bols, Wojciech~De Roeck, Michiel~De Wilde, and Bruno de~O.~Carvalho.
\newblock Classification of locality preserving symmetries on spin chains.
\newblock 2025.
\newblock \arXiv{2503.15088}.

\bibitem[BSN25]{PhysRevLett.134.191602}
Lea~E. Bottini and Sakura Sch\"afer-Nameki.
\newblock Construction of a gapless phase with haagerup symmetry.
\newblock {\em Phys. Rev. Lett.}, 134:191602, May 2025.

\bibitem[CFH{\etalchar{+}}24]{chen2024manifestlyunitaryhigherhilbert}
Quan Chen, Giovanni Ferrer, Brett Hungar, David Penneys, and Sean Sanford.
\newblock Manifestly unitary higher {H}ilbert spaces.
\newblock 2024.
\newblock \arXiv{2410.05120}.

\bibitem[CGGH23]{MR4616673}
Sebastiano Carpi, Tiziano Gaudio, Luca Giorgetti, and Robin Hillier.
\newblock Haploid algebras in {$C^*$}-tensor categories and the {S}chellekens
  list.
\newblock {\em Comm. Math. Phys.}, 402(1):169--212, 2023.

\bibitem[CGP25]{ciamprone2025weakquasihopfalgebrasctensor}
Sergio Ciamprone, Marco~Valerio Giannone, and Claudia Pinzari.
\newblock Weak quasi-hopf algebras, {C}*-tensor categories and conformal field
  theory, and the kazhdan-lusztig-finkelberg theorem.
\newblock 2025.
\newblock \arXiv{2101.10016}.

\bibitem[CHPJ24]{chen2022ktheoretic}
Quan Chen, Roberto Hern\'{a}ndez~Palomares, and Corey Jones.
\newblock K-theoretic {C}lassification of {I}nductive {L}imit {A}ctions of
  {F}usion {C}ategories on {AF}-algebras.
\newblock {\em Comm. Math. Phys.}, 405(3):Paper No. 83, 2024.

\bibitem[CHPJP22]{MR4419534}
Quan Chen, Roberto Hern\'{a}ndez~Palomares, Corey Jones, and David Penneys.
\newblock Q-system completion for {$\rm C^*$} 2-categories.
\newblock {\em J. Funct. Anal.}, 283(3):Paper No. 109524, 2022.

\bibitem[CJW24]{chatterjee2024emergentgeneralizedsymmetrymaximal}
Arkya Chatterjee, Wenjie Ji, and Xiao-Gang Wen.
\newblock Emergent generalized symmetry and maximal symmetry-topological-order.
\newblock 2024.
\newblock \arXiv{2212.14432}.

\bibitem[CM21]{CaMa}
Horacio Casini and Javier~M. Mag{\'a}n.
\newblock On completeness and generalized symmetries in quantum field theory.
\newblock {\em Modern Physics Letters A}, 36(36):2130025, 2025/06/30 2021.

\bibitem[CN11]{COGN}
Emilio Cobanera, Gerardo ~, Ortiz, , and Zohar Nussinov.
\newblock The bond-algebraic approach to dualities.
\newblock {\em Advances in Physics}, 60(5):679--798, 10 2011.

\bibitem[CW20]{PhysRevResearch.2.043044}
Meng Cheng and Dominic~J. Williamson.
\newblock Relative anomaly in ($1+1$)d rational conformal field theory.
\newblock {\em Phys. Rev. Res.}, 2:043044, Oct 2020.

\bibitem[CW23a]{PhysRevB.108.075105}
Arkya Chatterjee and Xiao-Gang Wen.
\newblock Holographic theory for continuous phase transitions: Emergence and
  symmetry protection of gaplessness.
\newblock {\em Phys. Rev. B}, 108:075105, Aug 2023.

\bibitem[CW23b]{PhysRevB.107.155136}
Arkya Chatterjee and Xiao-Gang Wen.
\newblock Symmetry as a shadow of topological order and a derivation of
  topological holographic principle.
\newblock {\em Phys. Rev. B}, 107:155136, Apr 2023.

\bibitem[DHR71]{MR0297259}
Sergio Doplicher, Rudolf Haag, and John~E. Roberts.
\newblock Local observables and particle statistics. {I}.
\newblock {\em Comm. Math. Phys.}, 23:199--230, 1971.
\newblock \mathscinet{MR0297259}.

\bibitem[DHR74]{DHR2}
Sergio Doplicher, Rudolf Haag, and John~E. Roberts.
\newblock Local observables and particle statistics {II}.
\newblock {\em Comm. Math. Phys.}, 35(1):49--85, 1974.

\bibitem[DKR15]{MR3406516}
Alexei Davydov, Liang Kong, and Ingo Runkel.
\newblock Functoriality of the center of an algebra.
\newblock {\em Adv. Math.}, 285:811--876, 2015.

\bibitem[EG11]{MR2837122}
David~E. Evans and Terry Gannon.
\newblock The exoticness and realisability of twisted {H}aagerup-{I}zumi
  modular data.
\newblock {\em Comm. Math. Phys.}, 307(2):463--512, 2011.

\bibitem[EGNO15]{MR3242743}
Pavel Etingof, Shlomo Gelaki, Dmitri Nikshych, and Victor Ostrik.
\newblock {\em Tensor categories}, volume 205 of {\em Mathematical Surveys and
  Monographs}.
\newblock American Mathematical Society, Providence, RI, 2015.
\newblock \mathscinet{MR3242743} \doi{10.1090/surv/205}.

\bibitem[EK94]{MR1301620}
David~E. Evans and Yasuyuki Kawahigashi.
\newblock Orbifold subfactors from {H}ecke algebras.
\newblock {\em Comm. Math. Phys.}, 165(3):445--484, 1994.

\bibitem[EK98]{MR1642584}
David~E. Evans and Yasuyuki Kawahigashi.
\newblock {\em Quantum Symmetries on Operator Algebras}.
\newblock Oxford Mathematical Monographs. Oxford Science Publications. The
  Clarendon Press, Oxford University Press, New York, 1998.
\newblock xvi+829 pp. ISBN: 0-19-851175-2, \mathscinet{MR1642584}.

\bibitem[EK23]{MR4642115}
David~E. Evans and Yasuyuki Kawahigashi.
\newblock Subfactors and mathematical physics.
\newblock {\em Bull. Amer. Math. Soc. (N.S.)}, 60(4):459--482, 2023.

\bibitem[ENO10]{MR2677836}
Pavel Etingof, Dmitri Nikshych, and Victor Ostrik.
\newblock Fusion categories and homotopy theory.
\newblock {\em Quantum Topol.}, 1(3):209--273, 2010.

\bibitem[Far20]{Farrelly2020reviewofquantum}
Terry Farrelly.
\newblock A review of {Q}uantum {C}ellular {A}utomata.
\newblock {\em {Quantum}}, 4:368, November 2020.

\bibitem[FMT24]{freed2024topological}
Daniel~S. Freed, Gregory~W. Moore, and Constantin Teleman.
\newblock Topological symmetry in quantum field theory.
\newblock {\em Quantum {Topol}. 15 (2024), no}, 15(3/4):779--869, 2024.

\bibitem[FRS02]{MR1940282}
J{\"u}rgen Fuchs, Ingo Runkel, and Christoph Schweigert.
\newblock T{FT} construction of {RCFT} correlators. {I}. {P}artition functions.
\newblock {\em Nuclear Phys. B}, 646(3):353--497, 2002.

\bibitem[FRS04]{MR2076134}
J{\"u}rgen Fuchs, Ingo Runkel, and Christoph Schweigert.
\newblock T{FT} construction of {RCFT} correlators. {III}. {S}imple currents.
\newblock {\em Nuclear Phys. B}, 694(3):277--353, 2004.

\bibitem[FRS05]{FRSIV}
J{\"u}rgen Fuchs, Ingo Runkel, and Christoph Schweigert.
\newblock {TFT} construction of {RCFT} correlators iv:: Structure constants and
  correlation functions.
\newblock {\em Nuclear Physics B}, 715(3):539--638, 2005.

\bibitem[GF93]{GFr93}
Fabrizio Gabbiani and J{\"u}rg Fr{\"o}hlich.
\newblock Operator algebras and conformal field theory.
\newblock {\em Comm. Math. Phys.}, 155(3):569--640, 1993.

\bibitem[GK21]{GK21}
Davide Gaiotto and Justin Kulp.
\newblock Orbifold groupoids.
\newblock {\em J. High Energy Phys.}, 2021(2):132, 2021.

\bibitem[GKSW15]{GKSW}
Davide Gaiotto, Anton Kapustin, Nathan Seiberg, and Brian Willett.
\newblock Generalized global symmetries.
\newblock {\em J. High Energy Phys.}, 2015(2):172, 2015.

\bibitem[GMP{\etalchar{+}}23]{MR4598730}
Pinhas Grossman, Scott Morrison, David Penneys, Emily Peters, and Noah Snyder.
\newblock The extended {H}aagerup fusion categories.
\newblock {\em Ann. Sci. \'Ec. Norm. Sup\'er. (4)}, 56(2):589--664, 2023.

\bibitem[GRLM23]{GarreRubio2023classifyingphases}
Jos{\'{e}} Garre-Rubio, Laurens Lootens, and Andr{\'{a}}s Moln{\'{a}}r.
\newblock Classifying phases protected by matrix product operator symmetries
  using matrix product states.
\newblock {\em {Quantum}}, 7:927, February 2023.

\bibitem[HC25]{huang2025topologicalholographyquantumcriticality}
Sheng-Jie Huang and Meng Cheng.
\newblock Topological holography, quantum criticality, and boundary states.
\newblock 2025.
\newblock \arXiv{2310.16878}.

\bibitem[HJJY25]{hataishi2025structuredhrbimodulesabstract}
Lucas Hataishi, David Jaklitsch, Corey Jones, and Makoto Yamashita.
\newblock On the structure of {DHR} bimodules of abstract spin chains.
\newblock 2025.
\newblock \arXiv{2504.06094}.

\bibitem[HJLW24]{huang2024phasegroupcategorybimodule}
Linzhe Huang, Chunlan Jiang, Zhengwei Liu, and Jinsong Wu.
\newblock Phase group category of bimodule quantum channels.
\newblock 2024.
\newblock \arXiv{2411.02707}.

\bibitem[HLO{\etalchar{+}}22]{PhysRevLett.128.231603}
Tzu-Chen Huang, Ying-Hsuan Lin, Kantaro Ohmori, Yuji Tachikawa, and Masaki
  Tezuka.
\newblock Numerical evidence for a {H}aagerup conformal field theory.
\newblock {\em Phys. Rev. Lett.}, 128:231603, Jun 2022.

\bibitem[Hol23]{Ho23}
Stefan Hollands.
\newblock Anyonic chains -- {$\alpha$}-induction --{CFT} --defects
  --subfactors.
\newblock {\em Comm. Math. Phys.}, 399(3):1549--1621, 2023.

\bibitem[HP23]{HP23}
Andr{\'e} Henriques and David Penneys.
\newblock Representations of fusion categories and their commutants.
\newblock {\em Selecta Mathematica}, 29(3):38, 2023.

\bibitem[Ina22]{In22}
Kansei Inamura.
\newblock On lattice models of gapped phases with fusion category symmetries.
\newblock {\em J. High Energy Phys.}, 2022(3):36, 2022.

\bibitem[Izu]{IzumiUnpublished}
Masaki Izumi.
\newblock Notes on $3^{2n}$ subfactors.

\bibitem[Izu01]{MR1832764}
Masaki Izumi.
\newblock The structure of sectors associated with {L}ongo-{R}ehren inclusions.
  {II}. {E}xamples.
\newblock {\em Rev. Math. Phys.}, 13(5):603--674, 2001.
\newblock \mathscinet{MR1832764} \doi{10.1142/S0129055X01000818}.

\bibitem[Izu02]{MR1900138}
Masaki Izumi.
\newblock Inclusions of simple {$C^\ast$}-algebras.
\newblock {\em J. Reine Angew. Math.}, 547:97--138, 2002.

\bibitem[Izu24]{IzumiBEK}
Masaki Izumi.
\newblock Minimal compact group actions on {C}*-algebras with simple fixed
  point algebras.
\newblock {\em Reviews in Mathematical Physics}, page 2461002, 2025/07/07 2024.

\bibitem[JL24]{JoLi24}
Corey Jones and Junhwi Lim.
\newblock An index for quantum cellular automata on fusion spin chains.
\newblock {\em Annales Henri Poincar{\'e}}, 25(10):4399--4422, 2024.

\bibitem[JMNR21]{MR4309554}
C.~Jones, S.~Morrison, D.~Nikshych, and E.~C. Rowell.
\newblock Rank-finiteness for {$G$}-crossed braided fusion categories.
\newblock {\em Transform. Groups}, 26(3):915--927, 2021.

\bibitem[JNP25]{jones2025holographybulkboundarylocaltopological}
Corey Jones, Pieter Naaijkens, and David Penneys.
\newblock Holography for bulk-boundary local topological order.
\newblock 2025.
\newblock \arXiv{2506.19969}.

\bibitem[JNPW25]{jones2025localtopologicalorderboundary}
Corey Jones, Pieter Naaijkens, David Penneys, and Daniel Wallick.
\newblock Local topological order and boundary algebras, 2025.

\bibitem[Jon83]{MR0696688}
Vaughan F.~R. Jones.
\newblock Index for subfactors.
\newblock {\em Invent. Math.}, 72(1):1--25, 1983.
\newblock \mathscinet{MR696688}, \doi{10.1007/BF01389127}.

\bibitem[Jon22]{MR4374438}
V.~F.~R. Jones.
\newblock Planar algebras, {I}.
\newblock {\em New Zealand J. Math.}, 52:1--107, 2021 [2021--2022].

\bibitem[Jon24]{jones2024dhr}
Corey Jones.
\newblock {DHR} bimodules of quasi-local algebras and symmetric quantum
  cellular automata.
\newblock {\em Quantum {Topol}. 15 (2024), no}, 15(3/4):633--686, 2024.

\bibitem[JP17]{MR3687214}
Corey Jones and David Penneys.
\newblock Operator algebras in rigid {$\rm C^*$}-tensor categories.
\newblock {\em Comm. Math. Phys.}, 355(3):1121--1188, 2017.

\bibitem[JP19]{MR3948170}
Corey Jones and David Penneys.
\newblock Realizations of algebra objects and discrete subfactors.
\newblock {\em Adv. Math.}, 350:588--661, 2019.

\bibitem[JP20]{MR4079745}
Corey Jones and David Penneys.
\newblock Q-systems and compact {W}*-algebra objects.
\newblock In {\em Topological phases of matter and quantum computation}, volume
  747 of {\em Contemp. Math.}, pages 63--88. Amer. Math. Soc., Providence, RI,
  2020.

\bibitem[JSW25]{jones2025quantumcellularautomatacategorical}
Corey Jones, Kylan Schatz, and Dominic~J. Williamson.
\newblock Quantum cellular automata and categorical dualities of spin chains.
\newblock 2025.
\newblock \arXiv{2410.08884}.

\bibitem[JW19]{PhysRevResearch.1.033054}
Wenjie Ji and Xiao-Gang Wen.
\newblock Noninvertible anomalies and mapping-class-group transformation of
  anomalous partition functions.
\newblock {\em Phys. Rev. Res.}, 1:033054, Oct 2019.

\bibitem[JW20]{PhysRevResearch.2.033417}
Wenjie Ji and Xiao-Gang Wen.
\newblock Categorical symmetry and noninvertible anomaly in symmetry-breaking
  and topological phase transitions.
\newblock {\em Phys. Rev. Res.}, 2:033417, Sep 2020.

\bibitem[Kaw20]{MR4109480}
Yasuyuki Kawahigashi.
\newblock A remark on matrix product operator algebras, anyons and subfactors.
\newblock {\em Lett. Math. Phys.}, 110(6):1113--1122, 2020.

\bibitem[Kaw23]{Kaw23}
Yasuyuki Kawahigashi.
\newblock Projector matrix product operators, anyons and higher relative
  commutants of subfactors.
\newblock {\em Math. Annalen}, 387(3):2157--2172, 2023.

\bibitem[Kis87]{MR905025}
Akitaka Kishimoto.
\newblock Type {${\rm I}$} orbits in the pure states of a {$C^*$}-dynamical
  system. {II}.
\newblock {\em Publ. Res. Inst. Math. Sci.}, 23(3):517--526, 1987.

\bibitem[KLW{\etalchar{+}}20a]{PhysRevResearch.2.043086}
Liang Kong, Tian Lan, Xiao-Gang Wen, Zhi-Hao Zhang, and Hao Zheng.
\newblock Algebraic higher symmetry and categorical symmetry: A holographic and
  entanglement view of symmetry.
\newblock {\em Phys. Rev. Res.}, 2:043086, Oct 2020.

\bibitem[KLW{\etalchar{+}}20b]{KLWZZ}
Liang Kong, Tian Lan, Xiao-Gang Wen, Zhi-Hao Zhang, and Hao Zheng.
\newblock Classification of topological phases with finite internal symmetries
  in all dimensions.
\newblock {\em Journal of High Energy Physics}, 2020(9):93, 2020.

\bibitem[KS24]{kapustin2024anomaloussymmetriesquantumspin}
Anton Kapustin and Nikita Sopenko.
\newblock Anomalous symmetries of quantum spin chains and a generalization of
  the {L}ieb-{S}chultz-{M}attis theorem, 2024.
\newblock \arXiv{2401.02533}.

\bibitem[KWZ17]{KONG201762}
Liang Kong, Xiao-Gang Wen, and Hao Zheng.
\newblock Boundary-bulk relation in topological orders.
\newblock {\em Nuclear Physics B}, 922:62--76, 2017.

\bibitem[KZ18]{MR3763324}
Liang Kong and Hao Zheng.
\newblock Gapless edges of 2d topological orders and enriched monoidal
  categories.
\newblock {\em Nuclear Phys. B}, 927:140--165, 2018.

\bibitem[KZ22a]{LZI}
Liang Kong and Hao Zheng.
\newblock Categories of quantum liquids {I}.
\newblock {\em Journal of High Energy Physics}, 2022(8):70, 2022.

\bibitem[KZ22b]{LZIII}
Liang Kong and Hao Zheng.
\newblock Categories of quantum liquids {III}.
\newblock 2022.
\newblock \arXiv{2201.05726}.

\bibitem[KZ24]{LZII}
Liang Kong and Hao Zheng.
\newblock Categories of quantum liquids {II}.
\newblock {\em Communications in Mathematical Physics}, 405(9):203, 2024.

\bibitem[Lan95]{MR1325694}
E.~C. Lance.
\newblock {\em Hilbert {$C^*$}-modules}, volume 210 of {\em London Mathematical
  Society Lecture Note Series}.
\newblock Cambridge University Press, Cambridge, 1995.
\newblock A toolkit for operator algebraists.

\bibitem[LDOV23]{PRXQuantum.4.020357}
Laurens Lootens, Clement Delcamp, Gerardo Ortiz, and Frank Verstraete.
\newblock Dualities in one-dimensional quantum lattice models: Symmetric
  {H}amiltonians and matrix product operator intertwiners.
\newblock {\em PRX Quantum}, 4:020357, Jun 2023.

\bibitem[LDV24]{PRXQuantum.5.010338}
Laurens Lootens, Clement Delcamp, and Frank Verstraete.
\newblock Dualities in one-dimensional quantum lattice models: Topological
  sectors.
\newblock {\em PRX Quantum}, 5:010338, Mar 2024.

\bibitem[LR95]{MR1332979}
R.~Longo and K.-H. Rehren.
\newblock Nets of subfactors.
\newblock volume~7, pages 567--597. 1995.
\newblock Workshop on Algebraic Quantum Field Theory and Jones Theory (Berlin,
  1994).

\bibitem[MMT23]{10.21468/SciPostPhysCore.6.4.066}
Heidar Moradi, Seyed~Faroogh Moosavian, and Apoorv Tiwari.
\newblock {Topological holography: Towards a unification of Landau and
  beyond-Landau physics}.
\newblock {\em SciPost Phys. Core}, 6:066, 2023.

\bibitem[Moo17]{moore2017quantummechanicsnoncommutativeamplitudes}
Gregory~W. Moore.
\newblock Quantum mechanics with noncommutative amplitudes, 2017.
\newblock \arXiv{1701.07746}.

\bibitem[MYLG25]{meng2025noninvertiblesptsonsiterealization}
Chenqi Meng, Xinping Yang, Tian Lan, and Zhengcheng Gu.
\newblock An on-site realization of (1+1)d anomaly-free fusion category
  symmetry, 2025.

\bibitem[Naa11]{MR2804555}
Pieter Naaijkens.
\newblock Localized endomorphisms in {K}itaev's toric code on the plane.
\newblock {\em Rev. Math. Phys.}, 23(4):347--373, 2011.

\bibitem[Naa17]{MR3617688}
Pieter Naaijkens.
\newblock {\em Quantum spin systems on infinite lattices}, volume 933 of {\em
  Lecture Notes in Physics}.
\newblock Springer, Cham, 2017.
\newblock A concise introduction.

\bibitem[NR14]{MR3210925}
Dmitri Nikshych and Brianna Riepel.
\newblock Categorical {L}agrangian {G}rassmannians and {B}rauer-{P}icard groups
  of pointed fusion categories.
\newblock {\em J. Algebra}, 411:191--214, 2014.

\bibitem[NS97]{NSz97}
Florian Nill and Korn{\'e}l Szlach{\'a}nyi.
\newblock Quantum chains of {H}opf algebras with quantum double cosymmetry.
\newblock {\em Comm. Math. Phys.}, 187(1):159--200, 1997.

\bibitem[NT13]{MR3204665}
Sergey Neshveyev and Lars Tuset.
\newblock {\em Compact quantum groups and their representation categories},
  volume~20 of {\em Cours Sp\'ecialis\'es [Specialized Courses]}.
\newblock Soci\'et\'e Math\'ematique de France, Paris, 2013.
\newblock \mathscinet{MR3204665}.

\bibitem[Ocn88]{MR996454}
Adrian Ocneanu.
\newblock Quantized groups, string algebras and {G}alois theory for algebras.
\newblock In {\em Operator algebras and applications, Vol.\ 2}, volume 136 of
  {\em London Math. Soc. Lecture Note Ser.}, pages 119--172. Cambridge Univ.
  Press, Cambridge, 1988.
\newblock \mathscinet{MR996454}.

\bibitem[Oga22]{MR4362722}
Yoshiko Ogata.
\newblock A derivation of braided {$C^*$}-tensor categories from gapped ground
  states satisfying the approximate {H}aag duality.
\newblock {\em J. Math. Phys.}, 63(1):Paper No. 011902, 48, 2022.

\bibitem[Oga24]{MR4790511}
Yoshiko Ogata.
\newblock Boundary states of a bulk gapped ground state in 2-{D} quantum spin
  systems.
\newblock {\em Comm. Math. Phys.}, 405(9):Paper No. 213, 58, 2024.

\bibitem[Pop95]{MR1334479}
Sorin Popa.
\newblock An axiomatization of the lattice of higher relative commutants of a
  subfactor.
\newblock {\em Invent. Math.}, 120(3):427--445, 1995.
\newblock \mathscinet{MR1334479} \doi{10.1007/BF01241137}.

\bibitem[RLL00]{MR1783408}
M.~R{\o}rdam, F.~Larsen, and N.~Laustsen.
\newblock {\em An introduction to {$K$}-theory for {$C^*$}-algebras}, volume~49
  of {\em London Mathematical Society Student Texts}.
\newblock Cambridge University Press, Cambridge, 2000.

\bibitem[Sch24]{schatz2024boundarysymmetries21dtopological}
Kylan Schatz.
\newblock Boundary symmetries of (2+1)d topological orders, 2024.
\newblock \arXiv{2408.10832}.

\bibitem[Sha24]{shao2024whatsundonetasilectures}
Shu-Heng Shao.
\newblock What's done cannot be undone: {TASI} lectures on non-invertible
  symmetries, 2024.
\newblock \arXiv{2308.00747}.

\bibitem[SN24]{SCHAFERNAMEKI20241}
Sakura Sch\"afer-Nameki.
\newblock {ICTP} lectures on (non-)invertible generalized symmetries.
\newblock {\em Physics Reports}, 1063:1--55, 2024.

\bibitem[SS25]{seifnashri2025disentanglinganomalyfreesymmetriesquantum}
Sahand Seifnashri and Wilbur Shirley.
\newblock Disentangling anomaly-free symmetries of quantum spin chains.
\newblock 2025.

\bibitem[SSS24]{10.21468/SciPostPhys.16.6.154}
Nathan Seiberg, Sahand Seifnashri, and Shu-Heng Shao.
\newblock {Non-invertible symmetries and LSM-type constraints on a tensor
  product Hilbert space}.
\newblock {\em SciPost Phys.}, 16:154, 2024.

\bibitem[SSS25]{shao2025additivityhaagdualitynoninvertible}
Shu-Heng Shao, Jonathan Sorce, and Manu Srivastava.
\newblock Additivity, {H}aag duality, and non-invertible symmetries.
\newblock 2025.
\newblock \arXiv{2503.20863}.

\bibitem[Tur94]{MR1292673}
V.~G. Turaev.
\newblock {\em Quantum invariants of knots and 3-manifolds}, volume~18 of {\em
  De Gruyter Studies in Mathematics}.
\newblock Walter de Gruyter \& Co., Berlin, 1994.

\bibitem[TW24a]{TW24}
Ryan Thorngren and Yifan Wang.
\newblock Fusion category symmetry. part {I}. {A}nomaly in-flow and gapped
  phases.
\newblock {\em J. High Energy Phys.}, 2024(4):132, 2024.

\bibitem[TW24b]{TW242}
Ryan Thorngren and Yifan Wang.
\newblock Fusion category symmetry. part ii. {C}ategoriosities at c = 1 and
  beyond.
\newblock {\em Journal of High Energy Physics}, 2024(7):51, 2024.

\bibitem[VBW{\etalchar{+}}18]{PhysRevLett.121.177203}
Robijn Vanhove, Matthias Bal, Dominic~J. Williamson, Nick Bultinck, Jutho
  Haegeman, and Frank Verstraete.
\newblock Mapping topological to conformal field theories through strange
  correlators.
\newblock {\em Phys. Rev. Lett.}, 121:177203, Oct 2018.

\bibitem[Ver22]{10.1063/5.0071215}
Dominic Verdon.
\newblock A covariant {S}tinespring theorem.
\newblock {\em J. Math. Phys.}, 63(9):091705, 09 2022.

\bibitem[VLVD{\etalchar{+}}22]{PhysRevLett.128.231602}
Robijn Vanhove, Laurens Lootens, Maarten Van~Damme, Ramona Wolf, Tobias~J.
  Osborne, Jutho Haegeman, and Frank Verstraete.
\newblock Critical lattice model for a {H}aagerup conformal field theory.
\newblock {\em Phys. Rev. Lett.}, 128:231602, Jun 2022.

\bibitem[Wal23]{MR4650344}
Daniel Wallick.
\newblock An algebraic quantum field theoretic approach to toric code with
  gapped boundary.
\newblock {\em J. Math. Phys.}, 64(10):Paper No. 102301, 29, 2023.

\bibitem[Wat90]{MR996807}
Yasuo Watatani.
\newblock Index for ${C}^*$-subalgebras.
\newblock {\em Mem. Amer. Math. Soc.}, 83(424):vi+117 pp., 1990.
\newblock \mathscinet{MR996807}, \googlebooks{Bp2cmONVye0C}.

\bibitem[WP25]{PhysRevB.111.115161}
Rui Wen and Andrew~C. Potter.
\newblock Classification of $1+1\mathrm{D}$ gapless symmetry protected phases
  via topological holography.
\newblock {\em Phys. Rev. B}, 111:115161, Mar 2025.

\bibitem[Xu98]{MR1617550}
Feng Xu.
\newblock New braided endomorphisms from conformal inclusions.
\newblock {\em Comm. Math. Phys.}, 192(2):349--403, 1998.

\bibitem[ZC24]{PhysRevB.110.035155}
Carolyn Zhang and Clay C\'ordova.
\newblock Anomalies of $(1+1)$-dimensional categorical symmetries.
\newblock {\em Phys. Rev. B}, 110:035155, Jul 2024.

\end{thebibliography}

\end{document}